\documentclass[a4paper,UKenglish,cleveref, autoref, thm-restate]{lipics-v2021}

\pdfoutput=1 \hideLIPIcs

\title{Complexity of Clique-Guarded First-Order Logic with Counting}

\author{Steffen {van Bergerem}}{Humboldt-Universität zu Berlin, Germany}{steffen.van.bergerem@hu-berlin.de}{https://orcid.org/0000-0002-5212-8992}{}
\author{Johannes Friedrich Lange}{Humboldt-Universität zu Berlin, Germany}{langefri@hu-berlin.de}{https://orcid.org/0009-0001-0791-7327}{}
\author{Nicole Schweikardt}{Humboldt-Universität zu Berlin, Germany}{schweikn@hu-berlin.de}{https://orcid.org/0000-0001-5705-1675}{}

\authorrunning{S.\ van Bergerem, J.\,F.\ Lange, and N.\ Schweikardt}

\Copyright{Steffen van Bergerem, Johannes F. Lange, and Nicole Schweikardt}



\category{} 

\relatedversion{}

\funding{This work was funded by the Deutsche Forschungsgemeinschaft
(DFG, German Research Foundation) --- project number 541000908
(gefördert durch die Deutsche Forschungsgemeinschaft (DFG) --- Projektnummer 541000908).}

\nolinenumbers

\EventEditors{Michal Kouck\'{y} and Daniela Petri\c{s}an}
\EventNoEds{2}
\EventLongTitle{51st International Symposium on Mathematical Foundations of Computer Science (MFCS 2026)}
\EventShortTitle{MFCS 2026}
\EventAcronym{MFCS}
\EventYear{2026}
\EventDate{August 24--28, 2026}
\EventLocation{Paris, France}
\EventLogo{}
\SeriesVolume{386}
\ArticleNo{15}

\usepackage{stmaryrd}
\usepackage{xspace}
\usepackage{aligned-overset}
\usepackage{thmtools}
\usepackage[capitalise,noabbrev]{cleveref}
\usepackage{xcolor}

\definecolor{ibm-ultramarine}{HTML}{648fff}
\definecolor{ibm-indigo}{HTML}{785ef0}
\definecolor{ibm-magenta}{HTML}{dc267f}
\definecolor{ibm-orange}{HTML}{fe6100}
\definecolor{ibm-gold}{HTML}{ffb000}

\newcommand{\N}{\ensuremath{\mathbb{N}}}
\newcommand{\Npos}{\ensuremath{\mathbb{N}_{\geq 1}}}
\newcommand{\Z}{\ensuremath{\mathbb{Z}}}
\newcommand{\Q}{\ensuremath{\mathbb{Q}}}
\newcommand{\Qpos}{\ensuremath{\mathbb{Q}_{> 0}}}
\newcommand{\R}{\ensuremath{\mathbb{R}}}

\newcommand{\A}{\ensuremath{\mathcal{A}}}
\newcommand{\B}{\ensuremath{\mathcal{B}}}
\newcommand{\G}{\ensuremath{\mathcal{G}}}

\newcommand{\K}{\ensuremath{\mathcal{K}}}
\DeclareMathOperator{\overlap}{overlap}

\newcommand{\D}{\ensuremath{\mathcal{D}}}
\newcommand{\Hypo}{\ensuremath{\mathcal{H}}}
\DeclareMathOperator{\err}{err}

\renewcommand{\phi}{\varphi}
\renewcommand{\epsilon}{\varepsilon}

\newcommand{\logic}[1]{\ensuremath{\mathsf{#1}}}

\newcommand{\FO}{\logic{FO}}

\newcommand{\FOC}{\logic{FOC}}
\newcommand{\cgFOC}{\logic{cgFOC}}

\newcommand{\vars}{\logic{vars}}

\newcommand{\sem}[1]{\llbracket #1 \rrbracket} \newcommand{\Land}{\ensuremath{\bigwedge}}
\newcommand{\Lor}{\ensuremath{\bigvee}}
\newcommand{\FOCCount}[2]{\ensuremath{\# {#1}.{#2}}}
\newcommand{\I}{\ensuremath{\mathcal{I}}} 

\newcommand{\SType}[3]{\ensuremath{\textup{S}^{#1}_{\!#2}(#3)}}
\newcommand{\MType}[3]{\ensuremath{M^{#1}_{\!#2}(#3)}}

\newcommand{\MTypeAVW}[1]{\MType{#1}{\A}{V/W}}
\newcommand{\STypePhiA}[1]{\SType{\phi}{\A}{#1}}

\newcommand{\STypeTA}[1]{\SType{t}{\A}{#1}}

\newcommand{\MTypeTVW}[1]{\MType{t}{#1}{V/W}}
\newcommand{\STypePhiAVW}{\SType{\phi}{\A}{V/W}}
\newcommand{\STypeTAVW}{\SType{t}{\A}{V/W}}

\newcommand{\MTypeTAVW}{\MType{t}{\A}{V/W}}
\newcommand{\absSType}[3]{\bigabs{\mkern1mu\SType{#1}{#2}{#3}}}

\newcommand{\absSTypePhiA}[1]{\absSType{\phi}{\A}{#1}}

\newcommand{\absSTypePhiAVW}{\absSType{\phi}{\A}{V/W}}

\DeclareMathOperator{\ar}{ar}

\DeclareMathOperator{\qr}{qr}
\DeclareMathOperator{\tp}{tp}
\DeclareMathOperator{\ltp}{ltp}
\DeclareMathOperator*{\free}{free}
\DeclareMathOperator{\col}{col}
\DeclareMathOperator{\wcol}{wcol}

\newcommand{\closure}{\ensuremath{\textup{cl}}}

\newcommand{\complexityclass}[1]{\ensuremath{\mathsf{#1}}}

\newcommand{\AWstar}{\ensuremath{\complexityclass{AW}[*]}}

\renewcommand{\O}{\ensuremath{\mathcal{O}}}
\newcommand{\bigO}{\O}

\newcommand{\bigmid}{:}

\newcommand{\set}[1]{\ensuremath{\{#1\}}} \newcommand{\setc}[2]{\ensuremath{\set{#1:#2}}}  \newcommand{\bigset}[1]{\ensuremath{\bigl\{ #1 \bigr\}}}
\newcommand{\bigsetc}[2]{\bigset{#1 \bigmid #2}}

\newcommand{\abs}[1]{\left\lvert#1\right\rvert}
\newcommand{\bigabs}[1]{\bigl\lvert#1\bigr\rvert}
\newcommand{\smallabs}[1]{\lvert#1\rvert}
\newcommand{\norm}[1]{\left\lVert#1\right\rVert}
\newcommand{\deff}{\coloneqq}
\newcommand{\ffed}{\eqqcolon}

\DeclareMathOperator{\dist}{dist}
\DeclareMathOperator{\rank}{rank}

\newcommand{\name}[1]{#1}

\newcommand{\ie}{\mbox{i.e.}\xspace}

\newcommand{\eg}{\mbox{e.g.}\xspace}

\newcommand{\iid}{\mbox{i.i.d.}\xspace}

\newcommand{\UC}{\textup{\textsc{uc}}} \newcommand{\VC}{\textup{\textsc{vc}}}

\newcommand{\C}{\ensuremath{\mathcal{C}}}

\renewcommand{\P}{\ensuremath{\mathbb{P}}}
\newcommand{\Pred}{\ensuremath{\mathsf{P}}}

\newcommand{\CS}{\ensuremath{\mathcal{S}}}
\newcommand{\T}{\ensuremath{\mathcal{T}}}

\newcommand{\tuple}[1]{\ensuremath{\bar{#1}}}

\newcommand{\tc}{\tuple{c}}

\newcommand{\tu}{\tuple{u}}
\newcommand{\tv}{\tuple{v}}
\newcommand{\tw}{\tuple{w}}
\newcommand{\tx}{\tuple{x}}
\newcommand{\ty}{\tuple{y}}
\newcommand{\tz}{\tuple{z}}

\newcommand{\neighb}[3]{\ensuremath{N_{#1}^{#2}(#3)}} \newcommand{\neighbr}[2]{\neighb{r}{#1}{#2}} \newcommand{\neighbA}[2]{\neighb{#1}{\A}{#2}} \newcommand{\Neighb}[3]{\ensuremath{\mathcal{N}_{#1}^{#2}(#3)}} \newcommand{\Neighbr}[2]{\Neighb{r}{#1}{#2}} \newcommand{\NeighbA}[2]{\Neighb{#1}{\A}{#2}} \newcommand{\nrA}[1]{\ensuremath{\neighb{r}{\A}{#1}}}
\newcommand{\NrA}[1]{\ensuremath{\Neighb{r}{\A}{#1}}}

\usepackage{tikz}
\usetikzlibrary{backgrounds,calc,shadows,shapes.geometric,shapes.symbols}
\tikzset{dropshadow/.style={drop shadow={opacity=.4, shadow xshift=.25ex, shadow yshift=-.25ex}},
  vertex template/.style={draw, semithick, circle, inner sep=.8ex, fill=white},
  small vertex template/.style={draw, semithick, circle, inner sep=.3ex, fill=white},
  vertex/.style={vertex template, dropshadow},
  small vertex/.style={small vertex template, dropshadow},
  edge vertex/.style={vertex, inner sep=.3ex, font=\scriptsize},
  first number relation/.style={vertex template, star, star points=5, star point ratio=2, inner sep=.3ex, fill=ibm-ultramarine},
  first number vertex/.style={vertex, first number relation},
  second number relation/.style={vertex template, regular polygon, regular polygon sides=6, inner sep=.475ex, fill=ibm-orange},
  second number vertex/.style={vertex, second number relation},
  edge/.style={thick},
}
\newcommand{\firstvertexrelation}{\begin{tikzpicture}\node[first number relation] {};\end{tikzpicture}}
\newcommand{\secondvertexrelation}{\begin{tikzpicture}\node[second number relation] {};\end{tikzpicture}}
\newcommand{\smallfirstvertexrelation}{\text{\footnotesize \firstvertexrelation}}
\newcommand{\smallsecondvertexrelation}{\text{\footnotesize \secondvertexrelation}}
 
\begin{document}

\maketitle

\begin{abstract}
  We introduce \emph{clique-guarded first-order logic with counting} (\(\cgFOC\)),
  a fragment of the first-order logic with counting \(\FOC\)
  [Kuske and Schweikardt, LICS 2017],
  and we study the complexity of this fragment.
  In particular, we prove computable upper bounds on the Vapnik--Chervonenkis (VC) dimension
  of \(\cgFOC\) formulas and on the graph dimension of \(\cgFOC\) counting terms
  on nowhere dense classes of relational structures.
  Furthermore, we show algorithmic metatheorems for \(\cgFOC\)
  for query answering, enumeration, and probably approximately correct (PAC) learning
  for Boolean and multiclass classification problems
  on classes of locally bounded expansion.
  On the other hand, we show that a slight extension of \(\cgFOC\) is already intractable on trees.
\end{abstract}

\section{Introduction}
\label{sec:intro}

The complexity of first-order logic (\(\FO\))
has been studied intensively for several complexity measures over the last decades.
For example, it has been studied in terms of the \emph{Vapnik--Chervonenkis (VC) dimension},
which has been introduced in the 1970s in model theory and learning theory
\cite{VapnikChervonenkis1971,Sauer1972,Shelah1972}
and measures the complexity of set systems.
In Valiant's framework of \emph{probably approximately correct (PAC) learning} \cite{Valiant_PAC},
a Boolean concept class is learnable if and only if it has finite VC dimension.
Grohe and Tur{\'{a}}n \cite{GroheTuran_Learnability} showed that set systems
definable by \(\FO\) formulas on classes of bounded local clique-width,
including classes of bounded genus and classes of bounded degree,
have bounded VC dimension.
Adler and Adler~\cite{AdlerAdler_NowhereDenseVC2014} generalised this
to all \emph{nowhere dense} graph classes,
a robust notion of sparsity introduced by Nešetřil
and Ossona de Mendez~\cite{NesetrilOssonaDeMendez_NowhereDense2011}.
The result of \cite{AdlerAdler_NowhereDenseVC2014} has only recently been generalised
by Braunfeld, Dawar, Eleftheriadis, and Papadopoulos to nowhere dense classes
of relational structures~\cite{BraunfeldDawarEleftheriadisPapadopoulos_NIP2023}.
While the proofs in
\cite{AdlerAdler_NowhereDenseVC2014,BraunfeldDawarEleftheriadisPapadopoulos_NIP2023}
rely on non-constructive compactness arguments for first-order logic,
Pilipczuk, Siebertz, and Toru{\'{n}}czyk \cite{PilipczukSiebertzTorunczyk_Types2018}
gave a combinatorial proof for the result of \cite{AdlerAdler_NowhereDenseVC2014},
yielding a computable upper bound on the VC dimension for \(\FO\) formulas
on effectively nowhere dense graph classes.

With respect to computational complexity,
in a seminal result, Grohe, Kreutzer, and Siebertz \cite{GroheKreutzerSiebertz_NowhereDense2017}
gave an almost-linear-time algorithm for the \(\FO\) model-checking problem
on nowhere dense graph classes.
That is, for every nowhere dense graph class \(\C\),
they showed that there is an algorithm that checks
whether a given \(\FO\) sentence \(\phi\) holds in a given graph \(G \in \C\)
in time \(\bigO(n^{1+\epsilon})\) for every \(\epsilon > 0\),
where \(n\) is the number of vertices in the graph,
and the constants in the \(\bigO\)-notation may depend on \(\C\), \(\phi\), and \(\epsilon\).
This generalises previous results for classes of bounded degree
\cite{Seese_LinearTimeBoundedDegree1996},
bounded expansion,
and locally bounded expansion
\cite{DvorakKralThomas_LocallyBoundedExpansion2013}.
Kreutzer and Dawar \cite{KreutzerDawar_2009} showed that
for classes that are closed under taking subgraphs, the result is optimal,
\ie for every such graph class that is \emph{not} nowhere dense,
the \(\FO\) model-checking problem is just as hard
as the model-checking problem on the class of all graphs,
which is hard for the parameterised complexity class \(\AWstar\).
Schweikardt, Segoufin, and Vigny \cite{SchweikardtSegoufinVigny_Enumeration2022}
generalised the result of \cite{GroheKreutzerSiebertz_NowhereDense2017}
to first-order testing and enumeration on nowhere dense classes.
On a nowhere dense class \(\C\) of relational structures,
this shows that for every \(\FO\) formula \(\phi\) and every relational structure \(\A \in \C\),
after an almost-linear-time preprocessing step,
we can check whether a given tuple of elements satisfies \(\phi\) on \(\A\) in constant time,
and we can enumerate all tuples satisfying \(\phi\) on \(\A\) with constant delay.

\paragraph*{First-order logic with counting}
Such \emph{algorithmic metatheorems},
which state that problems that can be formalised in a certain logic
can be solved efficiently on certain classes of inputs,
have also been obtained beyond first-order logic.
The first-order logic with counting \(\FOC\),
introduced by Kuske and Schweikardt in \cite{KuskeSchweikardt_FOCN},
extends first-order logic by \emph{counting terms} and \emph{numerical predicates}.
Counting terms are polynomials of \emph{$\#$-terms} of the form \(\FOCCount{\tx}{\phi}\),
which provide a means of counting the number of satisfying assignments of a formula \(\phi\).
The numerical predicates allow for the comparison of counting terms;
that is, for counting terms \(t_1, \dots, t_m\) and a numerical predicate \(\Pred \subseteq \Z^m\),
\(\Pred(t_1, \dots, t_m)\) is a formula.
For example, the numerical predicate \(\Pred_= \deff \setc{(i,i)}{i \in \Z}\)
facilitates checking that the results of two counting terms are equal.
The authors gave algorithms for constant-time testing
and constant-delay enumeration for \(\FOC\) on classes of bounded degree.
Beyond classes of bounded degree, Grohe and Schweikardt~\cite{GroheSchweikardt2018}
showed that the model-checking problem for \(\FOC\)
is already \(\AWstar\)-hard on the very restrictive class of unranked trees of height \(3\).
The article \cite{GroheSchweikardt2018} also defined the fragment \(\FOC_1\) of \(\FOC\)
and showed that the model-checking problem for \(\FOC_1\) can be solved in almost-linear time
on nowhere dense classes.
The logic \(\FOC_1\) restricts \(\FOC\) by requiring that for every formula of the form
\(\Pred(t_1, \dots, t_m)\),
the counting terms \(t_1, \dots, t_m\) have only one free variable in total.
Generalising the results of \cite{PilipczukSiebertzTorunczyk_Types2018} from \(\FO\) to \(\FOC_1\),
\cite{vanBergeremSchweikardt_VC2025}~proves a computable upper bound on the VC dimension
of \(\FOC_1\) formulas on effectively nowhere dense graph classes.
Toru{\'{n}}czyk \cite{Torunczyk_Aggregate2020} introduced the logic \(\FO_G[\mathbb{C}]\),
which extends \(\FO\) by aggregation in semirings.
For \(\FO_G[\mathbb{C}]\) on classes of bounded expansion,
the author gave an algorithm for constant-time query answering
after linear-time preprocessing.

\paragraph*{Agnostic PAC learning}
Valiant's PAC-learning setting mentioned above has been generalised
by Haussler \cite{Haussler_PAC} to the \emph{agnostic-PAC-learning} setting.
In this setting, within the logical framework for Boolean classification problems
from \cite{GroheTuran_Learnability},
we assume a probability distribution \(\D\) on \(A^\ell \times \set{0,1}\)
for a relational structure \(\A\) with universe \(A\) and for some \(\ell \in \Npos\).
Our goal is to find a function \(h \colon A^\ell \to \set{0,1}\),
called a \emph{hypothesis},
that (approximately) minimises the probability of making an error on a randomly drawn example,
that is, it should minimise \(\Pr_{(\tw, \lambda) \sim \D} \bigl(h(\tw) \neq \lambda\bigr)\).
Hypotheses are represented as pairs \((\phi, \tv)\),
where \(\phi(\tx, \ty)\) is a formula with \(\abs{\ty} = \ell\) and \(\tv \in A^{\abs{\tx}}\).
This represents the hypothesis \(h \colon A^\ell \to \set{0,1}\)
with \(h(\tw) = 1\) if and only if \(\A \models \phi[\tv, \tw]\).
In an agnostic-PAC-learning algorithm, we assume oracle access to \(\D\).
The algorithm shall draw a few examples from \(\D\)
and then construct a hypothesis based on these examples.
To make the problem feasible at all,
one needs to limit the complexity of the hypotheses we are allowed to return;
we do this by setting an upper bound on the length of the formula \(\phi\),
which also upper bounds other relevant complexity measures
such as number of variables and quantifier rank.
The agnostic-PAC-learning problem for Boolean classification has been studied before
for first-order logic
on nowhere dense classes~\cite{vanBergeremGroheRitzert_Parameterized,vanBergerem_PhDThesis}
and for first-order logic and extensions of first-order logic by counting and weight aggregation
on classes of
small degree~\cite{GroheRitzert_FO,vanBergerem_FOCNLearning2025,vanBergeremSchweikardt_FOWA2021}.
In \cite{vanBergeremSchweikardt_Multiclass2025},
the authors generalise the logical framework for Boolean classification problems
from \cite{GroheTuran_Learnability} to multiclass classification problems,
where hypotheses are of the form \(h \colon A^\ell \to \Z\),
and counting terms from \(\FOC_1\) are used to define these hypotheses.
However, \cite{vanBergeremSchweikardt_Multiclass2025} only considers the so-called
\emph{consistent-learning} setting, which is conceptually much simpler than the
(agnostic-)PAC-learning setting.

There have been various approaches to generalise the VC dimension to the multiclass setting
in order to characterise multiclass learnability.
The \emph{Natarajan dimension} and the \emph{graph dimension}
were both introduced by Natarajan \cite{Natarajan_GraphDimension1989}.
As it turns out, neither of them characterise multiclass PAC learnability,
since the Natarajan dimension underestimates the complexity
(there is a non-learnable concept class of Natarajan dimension \(1\))
and the graph dimension overestimates the complexity
(there is a learnable concept class of infinite graph dimension).
Only recently, the characterisation of multiclass learnability
was resolved by~\cite{BrukhimCDMY_MulticlassLearnability2022}
via the \emph{Daniely--Shalev-Shwartz (DS) dimension},
which was introduced in \cite{DanielyShalev-Shwartz_ImproperLearning2014}.
However, the algorithmic approaches to solve the multiclass agnostic-PAC-learning problem
for classes of bounded DS dimension are much more involved than in the Boolean case.
Meanwhile, for classes of bounded graph dimension
(which implies that they also have bounded DS dimension),
the paper \cite{DanielySBS_GraphDimensionERM2015} shows that one can still use
the simpler \emph{empirical-risk-minimisation (ERM)} procedure,
which is also used in the Boolean case.

\paragraph*{Our contributions}
We introduce \emph{clique-guarded first-order logic with counting} (\(\cgFOC\)),
a fragment of \(\FOC\) that is substantially more expressive than \(\FOC_1\).
Instead of asking that the counting terms \(t_1, \dots, t_m\)
in a formula of the form \(\Pred(t_1, \dots, t_m)\) have only one free variable in total,
we require every pair of free variables of \(t_1, \dots, t_m\) to be guarded by an atom.
On graphs, this means that one can compare counting terms in formulas of the form
\(\Land_{\set{x,y} \in \binom{X}{2}} E(x,y) \land \Pred(t_1, \dots, t_m)\)
for \(X \deff \bigcup_{i=1}^m \free(t_i)\),
ensuring that the free variables of \(t_1, \dots, t_m\) form a clique in the graph.
On relational structures \(\A\), formulas in \(\cgFOC\) ensure
that the free variables of the counting terms \(t_1, \dots, t_m\)
form a clique in the Gaifman graph \(G_\A\) of the structure \(\A\).
The formal definition of the logic \(\cgFOC\) can be found in \cref{sec:cgfoc}.

In \cref{sec:graph-dimension}, as the main result of this paper,
we show that counting terms from \(\cgFOC\) have bounded graph dimension
(see \cref{sec:graph-dimension} for the definition)
on nowhere dense classes.
This result can be viewed as a bound on the expressive power of \(\cgFOC\).
\begin{restatable}{theorem}{graphDimension}
  \label{thm:graph-dimension}
  For every (effectively) nowhere dense graph class \(\C\),
  there is a (computable) function \(f \colon \cgFOC \to \N\)
  such that for every \(\cgFOC\) counting term \(t(\tx,\ty)\)
  and every \(\sigma\)-structure \(\A\) with \(\sigma \supseteq \sigma(t)\)
  and \(G_\A \in \C\),
  the graph dimension of \(t(\tx,\ty)\) in \(\A\) is at most \(f(t)\).
\end{restatable}
Our proof relies on locality arguments and a lemma that bounds the rank of a `type matrix'.
To the best of our knowledge, we are the first to study graph dimension in a logical framework.
By the connection between the graph dimension and PAC learning described above,
\cref{thm:graph-dimension} enables us to solve multiclass PAC-learning problems on sparse classes,
where the task is to find a suitable hypothesis \(h\) that
maps tuples in \(A^\ell\) to values from the \emph{infinite} set \(\Z\).

As a direct consequence of \cref{thm:graph-dimension},
we obtain (computable) upper bounds on the VC dimension of \(\cgFOC\) formulas
on (effectively) nowhere dense classes of relational structures (see \cref{cor:VCdimension}).
This generalises the computable upper bounds on the VC dimension
of~\cite{PilipczukSiebertzTorunczyk_Types2018,vanBergeremSchweikardt_VC2025}
from \(\FO\) and \(\FOC_1\) to \(\cgFOC\)
and from classes of graphs to classes of relational structures.
In particular, this gives the first \emph{computable} upper bound on the VC dimension
of \(\FO\) formulas on effectively nowhere dense classes of relational structures.
This strengthens the result from~\cite{BraunfeldDawarEleftheriadisPapadopoulos_NIP2023},
and we even prove it for the much stronger logic \(\cgFOC\).

In addition to the complexity in terms of the graph dimension and the VC dimension,
we also study the computational complexity of evaluating \(\cgFOC\).
On classes of relational structures of locally bounded expansion,
we give constant-time query-answering and constant-delay enumeration algorithms for \(\cgFOC\)
with almost-linear-time preprocessing (see \cref{thm:answering-enumeration}).
This generalises the model-checking result from~\cite{GroheSchweikardt2018}
for classes of locally bounded expansion
from \(\FOC_1\) to \(\cgFOC\) and the query-answering and enumeration results
from~\cite{SegoufinVigny_LocallyBoundedExpansion2017} from \(\FO\) to \(\cgFOC\).

In \cref{sec:learning-formulas}, we combine the computable bounds on the graph dimension
and the VC dimension with the algorithmic metatheorems for \(\cgFOC\)
to develop agnostic-PAC-learning algorithms for \(\cgFOC\)-definable concepts.
In fact, we even solve a harder problem, \emph{agnostic PAC enumeration},
where we enumerate all hypotheses up to a given complexity bound with constant delay,
sorted by their accuracy (descending) and their complexity (ascending).
The first hypothesis we output will then be a correct solution
for the agnostic-PAC-learning problem.
We solve the multiclass agnostic-PAC-enumeration problem
on classes of locally bounded expansion for concepts definable via \(\cgFOC\) counting terms.
The algorithm we present (see \cref{thm:pac-learn-locally-bounded-expansion})
preprocesses the input in time polynomial in the size of the relational structure,
but the exponent depends on the desired complexity bound of the hypotheses.
For classes of bounded degree, we give a different algorithm that has only linear time preprocessing
and even works for the full logic \(\FOC\).
Moreover, for the restriction to Boolean classification,
we present an algorithm for the agnostic-PAC-enumeration problem for \(\cgFOC\)
that has almost-linear-time preprocessing on classes of locally bounded expansion.

Finally, in addition to our complexity analysis of \(\cgFOC\),
we also study variants of this logic.
In \(\cgFOC\), whenever we want to compare counting terms \(t_1, \dots, t_m\)
using a predicate \(\Pred \subseteq \Z^m\), we need to add guards that ensure that the
free variables of \(t_1, \dots, t_m\) form a clique,
so the free variables shall have pairwise distance at most \(1\).
Slightly weakening these guards leads to a logic
that is intractable on very simple graph classes.
As we discuss at the end of \cref{sec:cgfoc},
constructs of the form \(E(x,y) \land E(y,z) \land \Pred_=(t_1(x), t_2(z))\)
allow us to interpret the class of all graphs
in the class \(\T_2\) of coloured trees of height at most \(2\).
Hence, this would give us a formula with unbounded VC dimension
and an \(\AWstar\)-hard model-checking problem on \(\T_2\).
Moreover, this would give us a counting term with unbounded graph dimension on \(\T_2\),
and it is easy to show that it also has unbounded DS dimension and Natarajan dimension.

We also show that even the very restricted fragment of \(\cgFOC\)
with constructs of the form \(E(x,y) \land \Pred_=(t_1(x), t_2(y))\)
allows us to interpret the class of all graphs
in a class of coloured graphs of shrub-depth at most \(2\).
This shows that, while \(\cgFOC\) is well-suited for sparse classes of relational structures,
even much more restricted variants are too strong for very simple dense classes.
 \section{Preliminaries}
\label{sec:prelims}

We let \(\Z\), \(\N\), \(\Npos\), \(\Q\), and \(\Qpos\) denote the sets of
integers, non-negative integers, positive integers, rationals, and positive rationals, respectively.
For \(m, n \in \Z\), we let \([m, n] \deff \setc{i \in \Z}{m \leq i \leq n}\)
and \([n] \deff [1, n]\).
For a \(k\)-tuple \(\tv = (v_1, \dots, v_k)\),
we write \(\abs{\tv}\) to denote its \emph{length} \(k\),
and we let \(\tilde{v} \deff \set{v_1, \dots, v_k}\).
For a set \(X\), we let \(2^X \deff \setc{Y}{Y \subseteq X}\),
and for \(k \in \N\), we let \(\binom{X}{k} \deff \setc{Y}{Y \subseteq X, \abs{Y} = k}\).

\paragraph*{Graphs and Sparse Graph Classes}
When speaking of \emph{graphs}, we mean undirected graphs without self-loops.
For a graph \(G\), we write \(V(G)\) and \(E(G)\)
to denote its vertex set and edge set, respectively.
For \(V' \subseteq V(G)\), we write \(G[V']\) to denote the subgraph of \(G\) induced by \(V'\).

For two vertices \(v,w\) of a graph \(G\), we let \(\dist^G(v,w)\),
called the \emph{distance} between \(v\) and \(w\) in \(G\),
be the \emph{length} (that is, the number of edges) of a shortest path between \(v\) and \(w\).
If there is no such path, then \(\dist^G(v,w) \deff \infty\).
For \(r \in \N\), we let \(\neighbr{G}{v} \deff \setc{w \in V(G)}{\dist^G(v,w) \leq r}\)
be the \emph{\(r\)-neighbourhood} of \(v\) in \(G\),
we let \(\Neighbr{G}{v} \deff G[\neighbr{G}{v}]\) be the corresponding induced subgraph,
and we extend this to tuples \(\tv \in \bigl(V(G)\bigr)^k\) for \(k \in \Npos\)
by setting \(\neighbr{G}{\tv} \deff \bigcup_{i \in [k]} \neighbr{G}{v_i}\)
and \(\Neighbr{G}{\tv} \deff G[\neighbr{G}{\tv}]\).
We say that a graph \(G\) has \emph{radius at most \(r\)} if there
exists a \(v \in V(G)\) such that \(G = \Neighbr{G}{v}\).
For sets \(V, W, S \subseteq V(G)\),
we say that \emph{\(V\) and \(W\) are \(r\)-separated by \(S\) (in \(G\))}
if every path of length \(\leq r\) in \(G\) from a vertex in \(V\) to a vertex in \(W\)
contains a vertex from \(S\).
We use the same notion for tuples \(\tv\) and \(\tw\)
by considering the sets \(V=\tilde{v}\) and \(W=\tilde{w}\),
and we extend the notion to sets of tuples \(V'\) and \(W'\)
by considering the sets \(V = \bigcup_{\tv\in V'} \tilde{v}\)
and \(W = \bigcup_{\tw\in W'} \tilde{w}\).
If \(V\) and \(W\) are \(r\)-separated by the empty set,
then we say that \(V\) and \(W\) are \(r\)-separated.
A \emph{depth-\(r\) minor} of \(G\)
is a subgraph of a graph obtained from \(G\)
by contracting mutually vertex-disjoint connected subgraphs
of radius at most \(r\) to single vertices.

\begin{definition}
  \label{def:nowhere-dense}
  Let \(\C\) be a class of graphs.
We say that \(\C\) has \emph{bounded degree} if there is a constant \(d \in \N\)
  such that every graph \(G \in \C\) has maximum degree at most \(d\).

  \(\C\) has \emph{bounded expansion} if there is a function \(g \colon \N \to \N\)
  such that for every \(G \in \C\), \(r \in \N\),
  and every depth-\(r\) minor \(H\) of \(G\),
  it holds that \(\frac{\abs{E(H)}}{\abs{V(H)}} \leq g(r)\).

  \(\C\) has \emph{locally bounded expansion} if there is a function \(g \colon \N \times \N \to \N\)
  such that for every \(d \in \N\),
  the class of graphs \(\Neighb{d}{G}{v}\) of \(d\)-neighbourhoods of vertices \(v \in V(G)\)
  in graphs \(G \in \C\) has bounded expansion via the function \(g_d \colon \N \to \N, g_d(r) \deff g(d,r)\).
Equivalently, \(\C\) has locally bounded expansion
  if there is a function \(g \colon \N \times \N \to \N\)
  such that for every \(G \in \C\), \(v \in V(G)\), \(d, r \in \N\),
  and every depth-\(r\) minor \(H\) of \(\Neighb{d}{G}{v}\),
  it holds that \(\frac{\abs{E(H)}}{\abs{V(H)}} \leq g(d,r)\).

  Lastly, \(\C\) is \emph{nowhere dense} if there is a function \(g \colon \N \to \N\)
  such that for all \(r \in \N\), there is no graph \(G \in \C\) that contains
  the complete graph \(K_{g(r)}\) on \(g(r)\) vertices as a depth-\(r\) minor.

  If the function \(g\) above is computable,
  then we say that \(\C\) has \emph{effectively bounded expansion},
  \emph{effectively locally bounded expansion},
  or \(\C\) is \emph{effectively nowhere dense}, respectively.
\end{definition}

Every class of bounded degree has effectively bounded expansion,
every class of (effectively) bounded expansion has (effectively) locally bounded expansion,
and every class of (effectively) locally bounded expansion is (effectively) nowhere dense.
We refer to \cite{NesetrilOssonaDeMendez_Sparsity2012}
for an in-depth introduction of these notions.

\paragraph*{Relational Structures}
We use standard notation for structures,
and we refer to \cref{sec:appendix-notation} for details.
All structures in this paper are finite,
all signatures are relational and finite,
and relation symbols \(R\) have a non-negative \emph{arity} \(\ar(R) \in \N\).
Unless explicitly stated otherwise,
we use \(\A\) and \(\B\) and variations such as \(\A_1\), \(\A'\), etc.~to denote structures,
we denote their universes by \(A\), \(B\), \(A_1\), \(A'\), etc., respectively,
and \(\sigma(\A)\) denotes the signature of \(\A\).
We let the \emph{representation size} of a structure \(\A\) be
\(\norm{\A} \deff \abs{A} + \sum_{R \in \sigma(\A)} \max(\ar(R),1) \cdot \abs{R(\A)}\).

For our algorithmic results, we assume a linear order \(<\) on the universe \(A\).
For example, in our enumeration result, we will output tuples in lexicographic order
based on the linear order on \(A\),
and we could assume that \(A \subset \N\).
However, we do not assume that this linear order is realised by some relation symbol \(R \in \sigma\).
In fact, the classes of structures that we consider in this paper
will rule out linear orders of arbitrary size being implemented in \(\sigma\).

The \emph{\name{Gaifman} graph} of a structure \(\A\) is the graph \(G_\A\)
with \(V(G_\A) = A\) and where,
for \(v, w \in A\), we have \(\set{v,w} \in E(G_\A)\) if and only if \(v \neq w\)
and \(v,w \in \tilde{v}\) for some \(\tv \in R(\A)\) with \(R \in \sigma\).
The degree of a structure \(\A\) is the degree of its Gaifman graph \(G_\A\).
For a number \(r \in \N\) and elements \(v, w \in A\),
the notions of \emph{distance} \(\dist^\A(v,w)\) and
\emph{neighbourhood} \(\nrA{v}\)
naturally extend from graphs to structures by considering the Gaifman graph \(G_\A\) of \(\A\).
We write \(\NrA{v}\) to denote the induced substructure
\(\A[\nrA{v}]\) of \(\A\) on the set \(\nrA{v}\).

\paragraph*{First-Order Logic with Counting}
We assume familiarity with first-order logic (\(\FO\)).
Let \(\sigma\) be a signature,
and let \(\vars\) be a fixed and countably infinite set of variables.
A \emph{\(\sigma\)-interpretation} \(\I = (\A, \beta)\) consists of a
\(\sigma\)-structure \(\A\) and an \emph{assignment} \(\beta \colon \vars \to A\).
For \(k \in \N\), pairwise distinct variables \(x_1, \dots, x_k \in \vars\),
and elements \(v_1, \dots, v_k \in A\),
we write \(\I \frac{v_1, \dots, v_k}{x_1, \dots, x_k}\)
for the interpretation \((\A, \beta \frac{v_1, \dots, v_k}{x_1, \dots, x_k})\),
where \(\beta \frac{v_1, \dots, v_k}{x_1, \dots, x_k}\) is the assignment \(\beta'\) with
\(\beta'(x_i) = v_i\) for every \(i \in [k]\)
and \(\beta'(y) = \beta(y)\) for all \(y \in \vars \setminus \set{x_1, \dots, x_k}\).

We consider the logic \(\FOC\) that Kuske and Schweikardt introduced
in~\cite{KuskeSchweikardt_FOCN}.
This logic allows building numerical statements based on
counting terms as well as numerical predicates.
A \emph{numerical predicate collection} is a triple \((\P, \ar, \sem{.})\)
where \(\P\) is a countable set of \emph{predicate names},
and, to each \(\Pred \in \P\),
\(\ar\) assigns an \emph{arity} \(\ar(\Pred) \in \Npos\)
and \(\sem{.}\) assigns a \emph{semantics}
\(\sem{\Pred} \subseteq \Z^{\ar(\Pred)}\).
For the remainder of this paper,
fix a numerical predicate collection \((\P, \ar, \sem{.})\).
In algorithms, we will assume that machines have access to oracles
for evaluating the numerical predicates in constant time.
That is, given a predicate \(\Pred \in \P\)
and a tuple \((i_1, \dots, i_{\ar(\Pred)})\) of integers,
answering the oracle call `\((i_1, \dots, i_{\ar(\Pred)}) \in \sem{\Pred}\)?'
takes time \(\bigO(1)\).

For a signature \(\sigma\),
the set of \emph{formulas} and \emph{counting terms} for \(\FOC[\sigma]\) is defined as follows.
\emph{Formulas} are built according to the same rules as formulas in \(\FO[\sigma]\)
(in particular, \(R()\) is a formula if \(R \in \sigma\) with \(\ar(R) = 0\))
and the following additional rule:

\begin{enumerate}[(\(\star\))]
  \renewcommand{\labelenumi}{\textbf{(\theenumi)}}
  \renewcommand{\theenumi}{\(\star\)}
  \item\label{def:foc-P-short}
    If \(\Pred \in \P\), \(m = \ar(\Pred)\)
    and \(t_1, \dots, t_m\) are counting terms,
    then \(\Pred(t_1, \dots, t_m)\) is a formula.
\end{enumerate}

\emph{Counting terms} are built according to the following rules:
every integer \(i \in \Z\) is a counting term;
if \(t_1\) and \(t_2\) are counting terms,
then \((t_1 + t_2)\) and \((t_1 \cdot t_2)\) are also counting terms;
and

\begin{enumerate}[(\#)]
  \renewcommand{\labelenumi}{\textbf{(\theenumi)}}
  \renewcommand{\theenumi}{\#}
  \item\label{def:foc-countingterm-short}
    If \(\phi\) is a formula and
    \(\tx = (x_1, \dots, x_k)\) is a tuple of \(k\) pairwise distinct variables
    for some \(k \in \N\),
    then \(\FOCCount{\tx}{\phi}\) is a counting term.
    This includes \(k=0\), so \(\FOCCount{()}{\phi}\) is a counting term.
\end{enumerate}

Let \(\I = (\A, \beta)\) be a \(\sigma\)-interpretation.
For a formula \(\phi\) from \(\FOC[\sigma]\),
the semantics \(\sem{\phi}^\I\) is defined in the same way as for
\(\FO[\sigma]\), extended by

\begin{enumerate}[(\(\star\))]
  \item \(\sem{\Pred(t_1, \dots, t_m)}^\I = 1\)
    if \((\sem{t_1}^\I, \dots, \sem{t_m}^\I) \in \sem{\Pred}\),
    and \(\sem{\Pred(t_1, \dots, t_m)}^\I = 0\) otherwise.
\end{enumerate}

For a counting term \(t\) from \(\FOC[\sigma]\), the semantics \(\sem{t}^\I\) is defined via
\(\sem{i}^\I = i\) for \(i \in \Z\);
\(\sem{(t_1 + t_2)}^\I = \sem{t_1}^\I + \sem{t_2}^\I\)
and \(\sem{(t_1 \cdot t_2)}^\I = \sem{t_1}^\I \cdot \sem{t_2}^\I\);
and

\begin{enumerate}[(\#)]
  \item \(\sem{\FOCCount{\tx}{\phi}}^\I =
    \bigabs{\bigsetc{(v_1, \dots, v_k) \in A^k}
    {\sem{\phi}^{\I \frac{v_1, \dots, v_k}{x_1, \dots, x_k}} = 1}}\),
    where \(\tx = (x_1, \dots, x_k)\).
\end{enumerate}

We refer to \cref{def:foc} in \cref{sec:appendix-notation} for a more detailed definition
of \(\FOC[\sigma]\) that also includes the definition of \(\FO[\sigma]\),
and we refer to \cref{ex:cgfoc,ex:social-network}
as well as \cref{rem:weaker-guards,rem:dense-classes}
for examples of formulas and counting terms from \(\FOC[\sigma]\).
For counting terms \(t_1\) and \(t_2\),
we write \((t_1 - t_2)\) as a shorthand for \(\bigl(t_1 + ((-1) \cdot t_2)\bigr)\).
Counting terms of the form \(\FOCCount{\tx}{\phi}\) are also called \emph{$\#$-terms}.
An \emph{expression} is a formula or a counting term.
For an expression \(\xi\), we write \(\sigma(\xi)\) for the set of all
relation symbols that occur in \(\xi\).
The \emph{length} \(\abs{\xi}\) of an expression \(\xi\)
is its length when viewed as a word over the alphabet
\(\sigma \cup \vars \cup \P \cup \Z \cup \set{,} \cup
\set{=, \neg, \lor, (, ), \exists, \#, ., +, \cdot}\).

The \emph{free variables} \(\free(\xi)\) of an expression \(\xi\) are
inductively defined in the same way as for first-order logic, extended by
\(\free\bigl(\Pred(t_1, \dots, t_m)\bigr) = \bigcup_{i=1}^m
\free(t_i)\);
\(\free\bigl(\FOCCount{(x_1, \dots, x_k)}{\phi}\bigr)
= \free(\phi) \setminus \set{x_1, \dots, x_k}\);
\(\free(i) = \emptyset\) for \(i \in \Z\); and
\(\free\bigl((t_1 + t_2)\bigr) = \free\bigl((t_1 \cdot t_2)\bigr) = \free(t_1) \cup \free(t_2)\).
We write \(\xi(z_1, \dots, z_k)\) to indicate that
\(\free(\xi) \subseteq \set{z_1, \dots, z_k}\).
A \emph{sentence} is a formula without free variables, and
a \emph{ground term} is a counting term without free variables.

By \(\FOC\), we denote the union of all \(\FOC[\sigma]\) for arbitrary signatures \(\sigma\);
this applies analogously to \(\FO\).

For a formula \(\phi\) and a \(\sigma\)-interpretation \(\I\),
we write \(\I \models \phi\) if \(\sem{\phi}^\I = 1\).
Likewise, \(\I \not\models \phi\) indicates that
\(\sem{\phi}^\I = 0\).
For a formula \(\phi(x_1, \dots, x_k)\), a \(\sigma\)-structure \(\A\),
and a tuple \(\tv = (v_1, \dots, v_k) \in A^k\),
we write \(\A \models \phi[\tv]\) or \((\A, \tv) \models \phi\)
if \((\A, \beta) \models \phi\) for all assignments \(\beta\)
with \(\beta(x_i) = v_i\) for all \(i \in [k]\).
Furthermore, we set \(\sem{\phi(\tv)}^\A \deff 1\) if \(\A \models \phi[\tv]\),
and \(\sem{\phi(\tv)}^\A \deff 0\) otherwise.
If \(t(x_1, \dots, x_k)\) is a counting term,
we set \(\sem{t(\tv)}^\A \deff \sem{t}^{(\A,\beta)}\)
for one (and hence every) assignment \(\beta\)
with \(\beta(x_i) = v_i\) for all \(i \in [k]\).
Two expressions \(\xi, \xi'\) are \emph{equivalent}
if \(\sem{\xi}^\I = \sem{\xi'}^\I\) for all \(\sigma\)-interpretations \(\I\).
For \(d \in \N\), the expressions are called \emph{\(d\)-equivalent}
if \(\sem{\xi}^\I = \sem{\xi'}^\I\) for all \(\sigma\)-interpretations \(\I = (\A, \beta)\)
for all structures \(\A\) of degree at most \(d\).
For an \(\FOC\) expression \(\xi(\tx)\) with \(\abs{\tx} > 0\)
and for an \(r \in \N\),
we say that \(\xi(\tx)\) is \emph{\(r\)-local}
if for every \(\sigma(\xi)\)-structure \(\A\)
and every tuple \(\tv \in A^{\abs{\tx}}\),
it holds that \(\sem{\xi(\tv)}^\A = \sem{\xi(\tv)}^{\NrA{\tv}}\).
We say that \(\xi(\tx)\) is \emph{local} if it is \(r\)-local for some \(r \in \N\).
 \section{Clique-Guarded FOC}
\label{sec:cgfoc}

In this section, we introduce the logic
\emph{clique-guarded first-order logic with counting} (\(\cgFOC\)).
Moreover, we show that \(\cgFOC\) expressions can be efficiently evaluated
on classes of locally bounded expansion,
while variants of the logic are already intractable on very simple classes.

For a signature \(\sigma\),
the set of \emph{formulas} and \emph{counting terms} for \(\cgFOC[\sigma]\)
(clique-guarded first-order logic with counting) is built according to
the same rules as formulas and counting terms in \(\FOC[\sigma]\),
with the following modified version of rule~\eqref{def:foc-P-short}.

\begin{enumerate}[(\(\star_{cg}\))]
  \renewcommand{\labelenumi}{\textbf{(\theenumi)}}
  \renewcommand{\theenumi}{\(\star_{\text{cg}}\)}
  \item\label{def:cgfoc-P-short}
    If \(\Pred \in \P\), \(m = \ar(\Pred)\),
    \(k \in \N\),
    \(R_1, \dots, R_k \in \sigma\),
    \(\tz_i \in \vars^{\ar(R_i)}\) for all \(i \in [k]\),
    and \(t_1, \dots, t_m\) are counting terms
    such that for all \(z, z' \in \bigcup_{p=1}^m \free(t_p)\) with \(z \neq z'\),
    there is some \(i \in [k]\) with \(z, z' \in \tilde{z}_i\),
    then \(\bigl(\Land_{i=1}^k R_i(\tz_i) \land \Pred(t_1, \dots, t_m)\bigr)\) is a formula.
\end{enumerate}

\begin{remark}
  \label{rem:clique-guarded}
  Intuitively, this means that every pair of free variables
  of a subformula of the form \(\Pred(t_1, \dots, t_m)\) is guarded by an atom.
  Hence, the vertices that the free variables of \(\Pred(t_1, \dots, t_m)\)
  are interpreted with need to form a clique in the Gaifman graph of the structure.
In fact, for every signature \(\sigma\), every \(\Pred \in \P\), \(m \deff \ar(\Pred)\),
  and for all \(\cgFOC[\sigma]\) counting terms \(t_1, \dots, t_m\),
  there is a \(\cgFOC[\sigma]\) formula \(\phi\)
  such that for all \(\sigma\)-interpretations \(\I = (\A, \beta)\), it holds that \(\I \models \phi\)
  if and only if \((\sem{t_1}^\I, \dots, \sem{t_m}^\I) \in \sem{\Pred}\)
  and \(\setc{\beta(x)}{x \in \bigcup_{i=1}^m \free(t_i)}\) forms a clique in \(G_\A\).
  See \cref{subsec:appendix-cgfoc-clique-guarded} for such a formula.
\end{remark}

\begin{example}
  \label{ex:cgfoc}
  Let \(\sigma \deff \set{\firstvertexrelation, \secondvertexrelation, E, R}\)
  for unary relation symbols \(\firstvertexrelation, \secondvertexrelation\),
  a binary relation symbol \(E\), and a ternary relation symbol \(R\).
  Furthermore, let \(\P\) contain the binary predicate name \(\Pred_=\)
  with \(\sem{\Pred_=} \deff \setc{(i,i)}{i \in \Z}\).
  The \(\cgFOC[\sigma]\) counting terms
  \(t_{\smallfirstvertexrelation}(x)
  \deff \FOCCount{(z)}{\bigl(E(x,z) \land \firstvertexrelation(z)\bigr)}\) and
  \(t_{\smallsecondvertexrelation}(x)
  \deff \FOCCount{(z)}{\bigl(E(x,z) \land \secondvertexrelation(z)\bigr)}\)
  evaluate to the number of \firstvertexrelation{}- and \secondvertexrelation{}-coloured
  neighbours of \(x\), respectively.
  Furthermore, the \(\cgFOC[\sigma]\) formula
  \(\phi(x,y) \deff E(x,y) \land \Pred_{=}\bigl(t_{\smallfirstvertexrelation}(x),
  t_{\smallsecondvertexrelation}(y)\bigr)\)
  expresses that \(x\) and \(y\) are neighbours and
  the number of \firstvertexrelation{}-coloured neighbours of \(x\) is equal to
  the number of \secondvertexrelation{}-coloured neighbours of \(y\).

  In addition, \(u(x) \deff \FOCCount{(y,z)}{\bigl(\firstvertexrelation(y) \land
  \secondvertexrelation(z) \land R(x,y,z)\bigr)}\) is a \(\cgFOC[\sigma]\) counting term,
  and
  \[\psi(x_1, \dots, x_6) \deff
    \Bigl(R(x_1, x_2, x_4) \land R(x_2, x_3, x_5) \land R(x_1, x_3, x_6)
  \land \Pred_=\bigl(u(x_1), u(x_2) + u(x_3)\bigr)\Bigr)\]
  and \(\psi'(x_1, x_2, x_3) \deff \exists x_4 \exists x_5 \exists x_6\, \psi(x_1, \dots, x_6)\)
  are \(\cgFOC[\sigma]\) formulas.
  The free variables of \(\Pred_=\bigl(u(x_1), u(x_2) + u(x_3)\bigr)\) are \(x_1, x_2, x_3\).
  The guards ensure that \(x_1, x_2, x_3\) form a clique in the Gaifman graph.
  For this, the guards use additional variables \(x_4, x_5, x_6\) that do not need to be contained in any clique.
\end{example}

\begin{example}
  \label{ex:social-network}
  We consider a slight variation of the formula \(\phi(x, y)\) from \cref{ex:cgfoc}
  in the context of social networks.
  For this, let \(\sigma \deff \set{F}\) be a signature for directed graphs,
  \ie \(F\) is a binary relation symbol.
  A \(\sigma\)-structure \(\A\) then represents a social network,
  where the elements of \(A\) are the users of the network,
  and we have \((v,w) \in F(\A)\) if the user \(v\) is a follower of the user \(w\).
  Let \(\P\) contain the binary predicate name \(\Pred_>\)
  with \(\sem{\Pred_>} \deff \setc{(i,j)}{i, j \in \Z, i > j}\).
  Then
  \[\phi'(x,y) \deff F(x,y) \land \Pred_>\bigl(\FOCCount{(z)}{F(z,x)},
  10 \cdot \FOCCount{(z)}{F(z,y)}\bigr)\]
  is a \(\cgFOC[\sigma]\) formula.
  It expresses that \(x\) is a follower of \(y\)
  and \(x\) has more than ten times as many followers as \(y\).
  Therefore, \(x\) may be able to help \(y\) increase their popularity.
\end{example}

\begin{remark}
  The logic \(\FOC_1\) of \cite{GroheSchweikardt2018}
  is the fragment of \(\FOC\) where rule~\eqref{def:foc-P-short}
  may only be applied if \(\abs{\bigcup_{i=1}^m \free(t_i)} \leq 1\).
  It holds that \(\FOC_1 \subseteq \cgFOC\).
  The \(\cgFOC\) counting terms \(t_{\smallfirstvertexrelation}\),
  \(t_{\smallsecondvertexrelation}\), and \(u\) from \cref{ex:cgfoc}
  are also \(\FOC_1\) counting terms.
  However, the \(\cgFOC\) formulas \(\phi\), \(\psi\), \(\psi'\), and \(\phi'\)
  from \cref{ex:cgfoc,ex:social-network}
  are not contained in \(\FOC_1\), and there are no \(\FOC_1\) formulas equivalent to them.
\end{remark}

\begin{restatable}{lemma}{cgfocPreprocessing}
  \label{lem:cgfoc-preprocessing}
  For every signature \(\sigma\) and every \(\cgFOC[\sigma]\) formula \(\phi(\tx)\),
  there is a signature \(\sigma_\phi \supseteq \sigma\) and an \(\FO[\sigma_\phi]\) formula \(\phi'(\tx)\)
  such that for every \(\sigma\)-structure \(\A\),
  there is a \(\sigma_\phi\)-expansion \(\A_\phi\) of \(\A\)
  with \( G_{\A_\phi} = G_\A\) such that
  `\(\A \models \phi[\tv] \iff \A_\phi \models \phi'[\tv]\)'
  holds for all \(\tv \in A^{\abs{\tx}}\).
  Moreover, \(\sigma_\phi,\phi'\) can be computed from \(\sigma,\phi\),
  and \(\A_\phi\) can be computed from \(\phi,\A\).
\end{restatable}

This lemma provides a `computable reduction' from $\cgFOC$ formulas
to $\FO$ formulas; for our purposes it is crucial that the
Gaifman graph of the structure remains unchanged.
The proof of \cref{lem:cgfoc-preprocessing} is a straightforward structural induction on \(\phi\).
In this induction, constructs of the form \(\Pred(t_1, \dots, t_m)\)
are replaced by fresh relation symbols,
and for the corresponding relation we use the set of all tuples that satisfy \(\Pred(t_1, \dots, t_m)\)
and that form a clique in the Gaifman graph.
The formal proof is given in \cref{subsec:appendix-cgfoc-preprocessing}.
Enforcing in \(\cgFOC\) that the free variables in constructs of the form \(\Pred(t_1, \dots, t_m)\)
form a clique in the Gaifman graph is essential to obtain an enriched structure
\(\A_\phi\) with the same Gaifman graph as \(\A\);
we will heavily rely on this property throughout the paper.
On classes of locally bounded expansion,
the computations in \cref{lem:cgfoc-preprocessing} can be done efficiently,
as stated in the following result.

\begin{restatable}{lemma}{cgfocPreprocessingBoundedExpansion}
  \label{lem:cgfoc-preprocessing-bounded-expansion}
  Let \(\C\) be a graph class of locally bounded expansion.
  There is a function \(f\) such that the algorithm from \cref{lem:cgfoc-preprocessing}
  computing \(\A_\phi\) from \(\phi\) and any \(\A\) with \(G_\A \in \C\),
  given \(\epsilon \in \Qpos\) as an additional input,
  can be chosen so that it runs in time \(f(\phi, \sigma, \epsilon) \cdot \abs{A}^{1+\epsilon}\).
\end{restatable}

We prove this result in \cref{subsec:appendix-cgfoc-preprocessing-bounded-expansion}.
Based on \cref{lem:cgfoc-preprocessing-bounded-expansion},
we can prove the following algorithmic metatheorem for \(\cgFOC\).

\begin{restatable}[Query Answering and Enumeration]{theorem}{answeringAndEnumeration}
  \label{thm:answering-enumeration}
  Let \(\C\) be a graph class of locally bounded expansion.
  There is a function \(f\) and an algorithm that does the following.
  Given a \(\cgFOC\) expression \(\xi(\tx)\),
  a \(\sigma\)-structure \(\A\) for some \(\sigma \supseteq \sigma(\xi)\)
  with \(G_\A \in \C\),
  and given an \(\epsilon \in \Qpos\),
  after preprocessing in time \(f(\xi, \sigma, \epsilon) \cdot \abs{A}^{1+\epsilon}\),
  \begin{alphaenumerate}
    \item\label{item:answering}
      the algorithm can answer the following queries in time \(f(\xi, \sigma, \epsilon)\):
      given a tuple \(\tv \in A^{\abs{\tx}}\), output \(\sem{\xi(\tv)}^\A\), and
    \item\label{item:enumeration}
      if \(\xi\) is a formula, then the algorithm can enumerate
      all tuples \(\tv \in A^{\abs{\tx}}\) such that \(\A \models \xi[\tv]\)
      with \(f(\xi, \sigma, \epsilon)\) delay in lexicographic order, without duplicates.
  \end{alphaenumerate}
\end{restatable}

Our main tool in the proof of \cref{lem:cgfoc-preprocessing-bounded-expansion}
is the restriction of \cref{thm:answering-enumeration}\ref{item:answering} to \#-terms
of the form \(\FOCCount{\tx}{\phi(\tx, \ty)}\) for an \(\FO\) formula \(\phi\)
(see \cref{thm:fo-counting-testing-locally-bounded-expansion} in
\cref{subsec:appendix-cgfoc-preprocessing-bounded-expansion}).
This generalises a special case of \cite[Theorem~8]{Torunczyk_Aggregate2020}
from classes of bounded expansion to classes of locally bounded expansion.
In addition, we use the following result,
which states that the number of cliques of a given size in graphs from a nowhere dense class
is small and which we prove in \cref{subsec:appendix-nowhere-dense-cliques}.

\begin{restatable}{lemma}{nowhereDenseCliques}
  \label{lem:nowhere-dense-cliques}
  Let \(\C\) be a nowhere dense graph class.
  There is a function \(f \colon \N \times \Qpos \to \N\) such that
  \begin{enumerate}
    \item
      \label{item:nowhere-dense-cliques}
      for every \(k \in \Npos\), \(\epsilon \in \Qpos\), and \(G \in \C\),
      it holds that \(G\) contains at most \(f(k,\epsilon) \cdot \abs{V(G)}^{1+\epsilon}\)
      cliques of size \(k\), and
    \item
      \label{item:nowhere-dense-size}
      for every signature \(\sigma\),
      every \(\epsilon \in \Qpos\),
      and every \(\sigma\)-structure \(\A\) with \(G_\A \in \C\),
      it holds that \(\norm{\A} \leq \abs{\sigma} \cdot f(k, \epsilon) \cdot \abs{A}^{1+\epsilon}\),
      where \(k \in \N\) is the maximum arity of a relation symbol in \(\sigma\).
  \end{enumerate}
  If \(\C\) is effectively nowhere dense, then \(f\) is computable.
\end{restatable}

We prove \cref{thm:answering-enumeration}\ref{item:answering}
by combining \cref{lem:cgfoc-preprocessing-bounded-expansion}
with an analogous result for checking whether a tuple satisfies an \(\FO\) formula
in nowhere dense classes in constant time after almost-linear-time preprocessing
\cite[Corollary~2.4]{SchweikardtSegoufinVigny_Enumeration2022}.
Furthermore, we prove \cref{thm:answering-enumeration}\ref{item:enumeration}
by combining \cref{lem:cgfoc-preprocessing-bounded-expansion}
with the enumeration result \cite[Corollary~2.5]{SchweikardtSegoufinVigny_Enumeration2022}
for \(\FO\) in nowhere dense classes.
The proof can be found in \cref{subsec:appendix-evaluation}.

\Cref{thm:answering-enumeration}\ref{item:answering} generalises the \(\FO\) model-checking result
for classes of locally bounded expansion
\cite[Corollary~1.2]{DvorakKralThomas_LocallyBoundedExpansion2013}
and the corresponding \(\FO\) testing result
\cite[Corollary~15]{SegoufinVigny_LocallyBoundedExpansion2017}
to the stronger logic \(\cgFOC\),
and it generalises the \(\FO\) counting result
for classes of locally bounded expansion
\cite[Theorem~6]{SegoufinVigny_LocallyBoundedExpansion2017},
to a result on constant-time evaluation of \(\cgFOC\) counting terms
after almost-linear time preprocessing.
\Cref{thm:answering-enumeration}\ref{item:enumeration} generalises the \(\FO\) enumeration result
on classes of locally bounded expansion
\cite[Corollary~5]{SegoufinVigny_LocallyBoundedExpansion2017}
to the stronger logic \(\cgFOC\).

\begin{remark}
  \label{rem:weaker-guards}
  Rule~\eqref{def:cgfoc-P-short} enforces that vertices
  interpreting the free variables of \(\Pred(t_1, \dots, t_m)\)
  have pairwise distance at most \(1\) in the Gaifman graph.
  Weakening this to pairwise distance at most \(2\)
  yields a logic that is intractable on very simple classes of graphs.
  More precisely, as we discuss below,
  there is a transduction that uses constructs of the form
  \(E(x,y) \land E(y,z) \land \Pred_=(t_1(x), t_2(z))\) and that produces the class of all graphs
  from the class of coloured trees of height at most \(2\), using only two colours.
  This shows that the model-checking problem
  on the class of coloured trees of height at most \(2\)
  for a logic that allows such constructs
  is at least as hard as the first-order model-checking problem
  on the class of all graphs. The latter is known to be \(\AWstar\)-hard
  \cite{FlumGrohe_Parameterized}.

  We consider the signature \(\sigma \deff \set{E, \firstvertexrelation, \secondvertexrelation}\)
  of coloured graphs, where \(E\) is a binary relation symbol,
  and \firstvertexrelation{} and \secondvertexrelation{} are unary relation symbols.
Let \(n \in \Npos\), and let \(G\) be a graph with \(V(G) = [n]\).
  We define two formulas \(\phi_V\) and \(\phi_E\) and a coloured tree \(T\)
  with colours \firstvertexrelation{} and \secondvertexrelation{}
  such that \(\phi_V\) and \(\phi_E\) do not depend on \(G\) or \(n\),
  we have \(V(G) \subseteq V(T)\),
  and for all \(v \in V(T)\), we have \(T \models \phi_V[v]\) if and only if \(v \in V(G)\).
  Moreover, for all \(v_1, v_2 \in V(G)\),
  we will have \(T \models \phi_E[v_1, v_2]\) if and only if \(\set{v_1, v_2} \in E(G)\).
  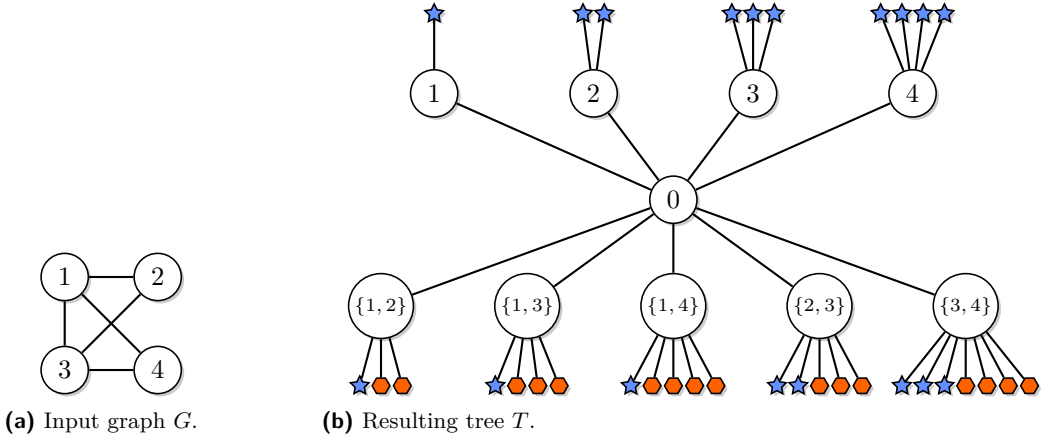
\begin{figure}
    \begin{subfigure}{.2\textwidth}
      \centering
      \begin{tikzpicture}
        \node[vertex]                                (1) {\(1\)};
\node[vertex, right of=1, node distance=3.5em] (2) {\(2\)};
\node[vertex, below of=1, node distance=3.5em] (3) {\(3\)};
\node[vertex, right of=3, node distance=3.5em] (4) {\(4\)};
\draw[edge] (1) -- (2) -- (3) -- (1) -- (4) -- (3);
       \end{tikzpicture}\caption{Input graph \(G\).}
    \end{subfigure}
    \hfill
    \begin{subfigure}{.7\textwidth}
      \centering
      \begin{tikzpicture}
        \coordinate (topcenter);
\node[vertex, left of=topcenter, node distance=9em]  (1) {\(1\)};
\node[vertex, left of=topcenter, node distance=3em]  (2) {\(2\)};
\node[vertex, right of=topcenter, node distance=3em] (3) {\(3\)};
\node[vertex, right of=topcenter, node distance=9em] (4) {\(4\)};

\node[vertex, below of=topcenter, node distance=4em] (0) {\(0\)};

\node[edge vertex, below of=topcenter, node distance=8em] (e14) {\(\set{1,4}\)};
\node[edge vertex, left of=e14, node distance=11em]        (e12) {\(\set{1,2}\)};
\node[edge vertex, left of=e14, node distance=5.5em]         (e13) {\(\set{1,3}\)};
\node[edge vertex, right of=e14, node distance=5.5em]        (e23) {\(\set{2,3}\)};
\node[edge vertex, right of=e14, node distance=11em]       (e34) {\(\set{3,4}\)};

\node[first number vertex] at ($(1) + ( 0.0em, 3em)$) (1v1) {};
\node[first number vertex] at ($(2) + (-0.4em, 3em)$) (2v1) {};
\node[first number vertex] at ($(2) + ( 0.4em, 3em)$) (2v2) {};
\node[first number vertex] at ($(3) + (-0.8em, 3em)$) (3v1) {};
\node[first number vertex] at ($(3) + ( 0.0em, 3em)$) (3v2) {};
\node[first number vertex] at ($(3) + ( 0.8em, 3em)$) (3v3) {};
\node[first number vertex] at ($(4) + (-1.2em, 3em)$) (4v1) {};
\node[first number vertex] at ($(4) + (-0.4em, 3em)$) (4v2) {};
\node[first number vertex] at ($(4) + ( 0.4em, 3em)$) (4v3) {};
\node[first number vertex] at ($(4) + ( 1.2em, 3em)$) (4v4) {};
\draw[edge] (1) -- (1v1);
\draw[edge] (2) -- (2v1);
\draw[edge] (2) -- (2v2);
\draw[edge] (3) -- (3v1);
\draw[edge] (3) -- (3v2);
\draw[edge] (3) -- (3v3);
\draw[edge] (4) -- (4v1);
\draw[edge] (4) -- (4v2);
\draw[edge] (4) -- (4v3);
\draw[edge] (4) -- (4v4);

\node[first number vertex]  at ($(e12) + (-0.8em, -3em)$) (e12v11) {};
\node[second number vertex] at ($(e12) + ( 0.0em, -3em)$) (e12v21) {};
\node[second number vertex] at ($(e12) + ( 0.8em, -3em)$) (e12v22) {};
\node[first number vertex]  at ($(e13) + (-1.2em, -3em)$) (e13v11) {};
\node[second number vertex] at ($(e13) + (-0.4em, -3em)$) (e13v21) {};
\node[second number vertex] at ($(e13) + ( 0.4em, -3em)$) (e13v22) {};
\node[second number vertex] at ($(e13) + ( 1.2em, -3em)$) (e13v23) {};
\node[first number vertex]  at ($(e14) + (-1.6em, -3em)$) (e14v11) {};
\node[second number vertex] at ($(e14) + (-0.8em, -3em)$) (e14v21) {};
\node[second number vertex] at ($(e14) + ( 0.0em, -3em)$) (e14v22) {};
\node[second number vertex] at ($(e14) + ( 0.8em, -3em)$) (e14v23) {};
\node[second number vertex] at ($(e14) + ( 1.6em, -3em)$) (e14v24) {};
\node[first number vertex]  at ($(e23) + (-1.6em, -3em)$) (e23v11) {};
\node[first number vertex]  at ($(e23) + (-0.8em, -3em)$) (e23v12) {};
\node[second number vertex] at ($(e23) + ( 0.0em, -3em)$) (e23v21) {};
\node[second number vertex] at ($(e23) + ( 0.8em, -3em)$) (e23v22) {};
\node[second number vertex] at ($(e23) + ( 1.6em, -3em)$) (e23v23) {};
\node[first number vertex]  at ($(e34) + (-2.4em, -3em)$) (e34v11) {};
\node[first number vertex]  at ($(e34) + (-1.6em, -3em)$) (e34v12) {};
\node[first number vertex]  at ($(e34) + (-0.8em, -3em)$) (e34v13) {};
\node[second number vertex] at ($(e34) + ( 0.0em, -3em)$) (e34v21) {};
\node[second number vertex] at ($(e34) + ( 0.8em, -3em)$) (e34v22) {};
\node[second number vertex] at ($(e34) + ( 1.6em, -3em)$) (e34v23) {};
\node[second number vertex] at ($(e34) + ( 2.4em, -3em)$) (e34v24) {};
\draw[edge] (e12) -- (e12v11);
\draw[edge] (e12) -- (e12v21);
\draw[edge] (e12) -- (e12v22);
\draw[edge] (e13) -- (e13v11);
\draw[edge] (e13) -- (e13v21);
\draw[edge] (e13) -- (e13v22);
\draw[edge] (e13) -- (e13v23);
\draw[edge] (e14) -- (e14v11);
\draw[edge] (e14) -- (e14v21);
\draw[edge] (e14) -- (e14v22);
\draw[edge] (e14) -- (e14v23);
\draw[edge] (e14) -- (e14v24);
\draw[edge] (e23) -- (e23v11);
\draw[edge] (e23) -- (e23v12);
\draw[edge] (e23) -- (e23v21);
\draw[edge] (e23) -- (e23v22);
\draw[edge] (e23) -- (e23v23);
\draw[edge] (e34) -- (e34v11);
\draw[edge] (e34) -- (e34v12);
\draw[edge] (e34) -- (e34v13);
\draw[edge] (e34) -- (e34v21);
\draw[edge] (e34) -- (e34v22);
\draw[edge] (e34) -- (e34v23);
\draw[edge] (e34) -- (e34v24);

\draw[edge] (0) -- (1)
            (0) -- (2)
            (0) -- (3)
            (0) -- (4);

\draw[edge] (0) -- (e12)
            (0) -- (e13)
            (0) -- (e14)
            (0) -- (e23)
            (0) -- (e34);
       \end{tikzpicture}\caption{Resulting tree \(T\).}
    \end{subfigure}
    \caption{Construction of a coloured tree \(T\) based on a graph \(G\)
      such that the transduction \((\phi_V, \phi_E)\) from \cref{rem:weaker-guards},
    applied on \(T\), produces \(G\).}
    \label{fig:hardness-example-tree}
  \end{figure}

  See \cref{fig:hardness-example-tree} for an example for the construction of the tree \(T\).
  The tree has a vertex \(0\) as its root
  with a child \(i\) for every vertex \(i \in [n] = V(G)\)
  and a child \(e\) for every edge \(e \in E(G)\).
  In addition, every vertex \(i \in V(G)\) has \(i\) \firstvertexrelation-neighbours,
  and every vertex \(\set{i,j} \in E(G) \subseteq V(T)\)
  with \(i < j\)
  has \(i\) \firstvertexrelation-neighbours
  and \(j\) \secondvertexrelation-neighbours.
  Note that \(T\) can be computed from \(G\) in polynomial time.

  The formula \(\phi_V(x)\) then defines the set of vertices \(V(G)\) in \(T\)
  by stating that \(x\) should have a \firstvertexrelation{}-coloured neighbour
  but no \secondvertexrelation{}-coloured neighbour.
  The formula \(\phi_E(x,y)\) defines the set of edges \(E(G)\) in \(T\)
  by stating that \(x\) and \(y\) form an edge in \(G\) if there is a vertex \(e\) in \(T\)
  such that the number of \firstvertexrelation{}-coloured neighbours of \(x\)
  coincides with the number of \firstvertexrelation{}-coloured neighbours of \(e\)
  and the number of \firstvertexrelation{}-coloured neighbours of \(y\)
  coincides with the number of \secondvertexrelation{}-coloured neighbours of \(e\),
  or the same holds with reversed roles for \(x\) and \(y\).
  For this, one can use the \(\cgFOC\) counting terms
  \(t_{\smallfirstvertexrelation}(x)
  \deff \FOCCount{(y)}{\bigl(E(x,y) \land \firstvertexrelation(y)\bigr)}\)
  and
  \(t_{\smallsecondvertexrelation}(x)
  \deff \FOCCount{(y)}{\bigl(E(x,y) \land \secondvertexrelation(y)\bigr)}\)
  to count the number of neighbours of a certain colour.
  To compare the number of neighbours of two vertices,
  we can use the fact that the distance between every pair of different vertices
  from \(V(G) \cup E(G) \subset V(T)\) is exactly \(2\) in \(T\).
  Hence, for example, to check that the number of \firstvertexrelation{}-coloured neighbours
  of \(x\) and \(e\) coincide, we can use the formula
  \(\exists z_0\ \bigl( E(x,z_0) \land E(z_0, e) \land \Pred_=\bigl(
    t_{\smallfirstvertexrelation}(x),
    t_{\smallfirstvertexrelation}(e)
  \bigr)\bigr)\).
  Note that this formula is an \(\FOC\) formula, but not a \(\cgFOC\) formula,
  since the pair \(x, e\) of free variables is not guarded by an edge.

  We refer to \cref{subsec:appendix-weaker-guards} for more details on the construction.
\end{remark}

\begin{remark}
  \label{rem:dense-classes}
  A similar argument shows that even stricter guards than the ones used for
  rule~\eqref{def:cgfoc-P-short} are not suitable for dense classes:
  using only constructs of the form
  \(E(x,y) \land \Pred_=(t_1(x), t_2(y))\),
  that is, where the free variables are guarded by a single edge,
  we give a transduction in \cref{subsec:appendix-dense-classes}
  that produces the class of all graphs
  from a class of coloured graphs of shrub-depth at most \(2\).
  Hence, the model-checking problem for this restricted logic is already \(\AWstar\)-hard
  on a class of coloured graphs of shrub-depth at most \(2\),
  which is a very simple dense class \cite{GanianHNOMR_Shrubdepth}.
  See \cref{subsec:appendix-dense-classes} for details.
\end{remark}
 \section{Graph Dimension of \texorpdfstring{$\cgFOC$}{cgFOC} on Nowhere Dense Classes}
\label{sec:graph-dimension}

In this section, we generalise the results of
\cite{PilipczukSiebertzTorunczyk_Types2018,vanBergeremSchweikardt_VC2025}
from \(\FO\) and \(\FOC_1\) on graphs to \(\cgFOC\) on relational structures.
Motivated by the multiclass learning setting as introduced
in \cite{vanBergeremSchweikardt_Multiclass2025},
which we study in \cref{sec:learning-formulas}, we also generalise these results
from formulas to counting terms,
proving that for every nowhere dense class \(\C\)
and every \(\cgFOC\) counting term \(t\), the \emph{graph dimension} of \(t\) in \(\C\),
a generalisation of the \emph{VC dimension},
both of which are defined below, is bounded by a value depending only on \(\C\) and \(t\).
If the class is effectively nowhere dense,
the bound will be computable.

The general notions of \emph{VC dimension} and \emph{graph dimension}
are defined as follows \cite{VapnikChervonenkis1971,Natarajan_GraphDimension1989}.
Let \(X\) be a set.
For a set \(\Hypo \subseteq 2^X\) and \(X' \subseteq X\),
we say that \(\Hypo\) \emph{shatters} \(X'\) if \(\setc{X' \cap H}{H \in \Hypo} = 2^{X'}\).
The \emph{VC dimension} of \(\Hypo\) is the maximum cardinality of a set \(X' \subseteq X\)
that is shattered by \(\Hypo\).
For a set \(\Hypo\) of functions \(h \colon X \to \Z\) and \(X' \subseteq X\),
we say that \(\Hypo\) \emph{graph shatters} \(X'\) if there is a function \(g \colon X' \to \Z\)
such that for every \(Y \subseteq X'\), there is a function \(h \in \Hypo\) with
\(g(x) = h(x)\) for all \(x \in Y\) and \(g(x) \neq h(x)\) for all \(x \in X' \setminus Y\).
The \emph{graph dimension} of \(\Hypo\) is the maximum cardinality of a set \(X' \subseteq X\)
that is graph shattered by \(\Hypo\).

\begin{lemma}
  \label{lem:vc-dim-vs-graph-dim}
  Let \(X\) be a set, and let \(\Hypo \subseteq 2^X\).
  We let \(\Hypo' \deff \setc{h_Y}{Y \in \Hypo}\),
  where, for a set \(Y \subseteq X\),
  we let \(h_Y \colon X \to \Z\) be the indicator function for \(Y\)
  with \(h_Y(x) \deff 1\) if \(x \in Y\), and \(h_Y(x) \deff 0\) otherwise.

  A set \(X' \subseteq X\) is shattered by \(\Hypo\)
  if and only if it is graph shattered by \(\Hypo'\).
  This implies that the VC dimension of \(\Hypo\) is equal to the graph dimension of \(\Hypo'\).
\end{lemma}
\begin{proof}
  For the forward direction, let \(X' \subseteq X\),
  and let \(g \colon X' \to \Z\) with \(g(x) = 1\) for all \(x \in X'\).
  If \(X'\) is shattered by \(\Hypo\),
  then \(\setc{X' \cap H}{H \in \Hypo} = 2^{X'}\).
  Hence, for every \(Y \subseteq X'\),
  there is a set \(H \in \Hypo\) such that \(X' \cap H = Y\).
  This implies that, for all \(x \in X'\),
  we have \(h_H(x) = 1 = g(x)\) if and only if \(x \in Y\),
  and we know that \(h_H \in \Hypo'\).
  Since this holds for every \(Y \subseteq X'\),
  this shows that \(X'\) is graph shattered by \(\Hypo'\).

  Now, for the backward direction,
  suppose that \(X' \subseteq X\) is a set that is graph shattered by \(\Hypo'\),
  and let \(g \colon X' \to \Z\) be a witness for this,
  that is, for every \(Y \subseteq X'\),
  there is a function \(h \in \Hypo'\) with \(h(x) = g(x)\) for all \(x \in Y\)
  and \(h(x) \neq g(x)\) for all \(x \in X' \setminus Y\).
  In particular, there needs to be such a function \(h \in \Hypo'\) for \(Y = X'\),
  where \(h(x) = g(x)\) for all \(x \in X'\).
  This shows that \(g(x) = h(x) \in \set{0,1}\) for all \(x \in X'\).
  Based on \(g\), we define a mapping \(\pi \colon 2^{X'} \to 2^{X'}\),
  where \(\pi(Y) \deff \setc{x \in Y}{g(x) = 1} \cup \setc{x \in X' \setminus Y}{g(x) = 0}\).
  Now let \(Y \subseteq X'\), \(Y' \deff \pi(Y) \subseteq X'\),
  and let \(h \in \Hypo'\) with \(h(x) = g(x)\) for all \(x \in Y'\)
  and \(h(x) \neq g(x)\) for all \(x \in X' \setminus Y'\).
  Then there is some \(H \in \Hypo\) such that \(h = h_H\).
Let \(x \in X'\).
  If \(g(x) = 1\), then \(x \in Y\) if and only if \(x \in Y'\)
  if and only if \(h(x) = g(x)\) if and only if \(h(x) = 1\)
  if and only if \(x \in H\).
  On the other hand, if \(g(x) = 0\),
  then \(x \in Y\) if and only if \(x \in X' \setminus Y'\)
  if and only if \(h(x) \neq g(x)\)
  if and only if \(h(x) = 1\)
  if and only if \(x \in H\).
  Combined, this shows that \(H \cap X' = Y\).
  Since \(Y\) was chosen arbitrarily, this shows that \(X'\) is shattered by \(\Hypo\).
\end{proof}

To define the notions of VC dimension and graph dimension for \(\cgFOC\),
we need the following notation.
For a \(\cgFOC\) formula \(\phi(\tx, \ty)\),
a \(\sigma\)-structure \(\A\) for some \(\sigma \supseteq \sigma(\phi)\),
and a set \(W \subseteq A\),
we let \(\tp^\phi_\A(\tv/W) \deff \setc{\tw \in W^{\abs{\ty}}}{\A \models \phi[\tv, \tw]}\)
for all \(\tv \in A^{\abs{\tx}}\).
Alternatively, we can view \(\tp^\phi_\A(\tv/W)\)
as a function from \(W^{\abs{\ty}}\) to \(\set{0,1}\) that maps a tuple
\(\tw\) to \(\sem{\phi(\tv, \tw)}^\A\).
Accordingly, for a \(\cgFOC\) counting term \(t(\tx, \ty)\), we let
\(\tp^t_\A(\tv/W) \colon W^{\abs{\ty}} \to \Z\)
with \(\bigl(\tp^t_\A(\tv/W)\bigr)(\tw) \deff \sem{t(\tv, \tw)}^\A\)
for all \(\tw \in W^{\abs{\ty}}\).
For a \(\cgFOC\) expression \(\xi(\tx, \ty)\) and sets \(V, W \subseteq A\), we define
\(\SType{\xi}{\A}{V/W} \deff\ \setc{\tp^\xi_\A(\tv/W)}{\tv \in V^{\abs{\tx}}}.\)
For a \(\cgFOC\) counting term \(t(\tx, \ty)\),
a \(\sigma\)-structure \(\A\) for some \(\sigma \supseteq \sigma(t)\),
two sets \(V, W \subseteq A\), and \(X \deff W^{\abs{\ty}}\),
we study the graph dimension of
\[\Hypo \deff \STypeTAVW = \setc{\tp^t_\A(\tv/W)}{\tv \in V^{\abs{\tx}}}
= \bigsetc{W^{\abs{\ty}} \to \Z, \tw \mapsto \sem{t(\tv, \tw)}^\A}{\tv \in V^{\abs{\tx}}}.\]
We let the \emph{graph dimension of \(t(\tx, \ty)\) in \(\A\)}
be the graph dimension of \(\STypeTA{A/A}
= \bigsetc{A^{\abs{\ty}} \to \Z, \tw \mapsto \sem{t(\tv, \tw)}^\A}{\tv \in A^{\abs{\tx}}}\).
The \emph{VC dimension} of a \(\cgFOC\) formula \(\phi(\tx,\ty)\) in \(\A\)
is defined as the VC dimension of \(\STypePhiA{A/A}
= \bigsetc{\setc{\tw \in A^{\abs{\ty}}}{\A \models \phi[\tv, \tw]}}{\tv \in A^{\abs{\tx}}}\).
Next, we prove our main result of this paper, which we repeat for convenience.

\graphDimension*

In fact, we prove a stronger result on the \emph{ladder index} of \(\cgFOC\) counting terms.
For a counting term \(t(\tx, \ty)\) and \(L \in \Npos\),
a \emph{\(t\)-ladder} of length \(L\) in a structure \(\A\)
is a sequence \(\tv_1, \dots, \tv_L, \tw_1, \dots, \tw_L\)
such that \(\tv_i \in A^{\abs{\tx}}\) and \(\tw_i \in A^{\abs{\ty}}\) for all \(i \in [L]\)
and there is a function \(g \colon \set{\tw_1, \dots, \tw_L} \to \Z\)
with \(\sem{t(\tv_i, \tw_j)}^\A = g(\tw_j)\) iff \(i \leq j\), for all \(i, j \in [L]\).
The smallest \(L\) for which there is no \(t\)-ladder of length \(L\) in \(\A\)
is called the \emph{ladder index of \(t(\tx,\ty)\) in \(\A\)}.

\begin{restatable}{theorem}{ladderIndexTerm}
  \label{thm:ladder-index-term}
  There are computable functions \(f \colon \cgFOC \times \N \to \N\)
  and \(g \colon \cgFOC \to \N\) such that,
  for every \(\cgFOC\) counting term \(t(\tx,\ty)\),
  every \(p \in \N\),
  and every \(\sigma\)-structure \(\A\) with \(\sigma \supseteq \sigma(t)\)
  such that \(G_\A\) excludes \(K_p\) as a depth-\(g(t)\) minor,
  the ladder index of \(t(\tx,\ty)\) in \(\A\) is at most \(f(t, p)\).
\end{restatable}

In our proofs, it will be convenient to represent \(\STypeTAVW\) as a matrix.
For this, we let \(\MTypeTAVW \in \Z^{V^{\abs{\tx}} \times W^{\abs{\ty}}}\)
be the matrix with \(\bigl(\MTypeTAVW\bigr)_{\tv, \tw} \deff \sem{t(\tv, \tw)}^\A\).
The bound from \cite[Corollary~3.2]{PilipczukSiebertzTorunczyk_Types2018} for formulas
would correspond to a bound on the number of different functions in \(\STypeTAVW\)
or, equivalently, to the number of different rows in \(\MTypeTAVW\).
However, in the case of \(\cgFOC\) counting terms, in general,
the number of different functions cannot be constantly bounded.
This even holds when replacing all non-zero entries by ones.

In our main technical tool in this section,
which is the following generalisation of
\cite[Corollary~3.2]{PilipczukSiebertzTorunczyk_Types2018} and
\cite[Lemma~4.1]{vanBergeremSchweikardt_VC2025}
for counting terms instead of formulas,
we instead prove an upper bound on the \emph{rank} of the matrix \(\MTypeTAVW\) over the reals.
Remarkably, this rank bound will then suffice to prove a variety of results
concerning the graph dimension and ladder index of \(\cgFOC\) counting terms on nowhere dense classes.

\begin{restatable}{lemma}{cgfocRank}
  \label{lem:cgfoc-rank}
  There are computable functions \(T \colon \cgFOC \times \N \to \N\)
  and \(r \colon \cgFOC \to \N\) such that,
  for every \(\cgFOC\) counting term \(t(\tx, \ty)\),
  every \(m \in \N\),
  every \(\sigma\)-structure \(\A\) for some \(\sigma \supseteq \sigma(t)\),
  and all \(V, W \subseteq A\) that are \(r(t)\)-separated in \(G_\A\)
  by a set of size at most \(m\),
  we have \(\rank\bigl(\MTypeTAVW\bigr) \leq T(t, m)\).
\end{restatable}

\begin{proof}[Proof sketch]
  Every \(\cgFOC\) counting term is a polynomial of integers and \#-terms.
  For a counting term \(t = i\) for an integer \(i\),
  and for \(r(t) \deff 0\), all entries of \(\MTypeTAVW\) are equal to \(i\),
  so the rank of \(\MTypeTAVW\) is at most \(1\).
  For a counting term \(t = t_1 + t_2\) or \(t = t_1 \cdot t_2\),
  and for \(r(t) \deff \max(r(t_1), r(t_2))\), it can be checked that we have
  \(\rank\bigl(\MTypeTAVW\bigr) \leq
  \rank\bigl(\MTypeAVW{t_1}\bigr) + \rank\bigl(\MTypeAVW{t_2}\bigr)\)
  or
  \(\rank\bigl(\MTypeTAVW\bigr) \leq
  \rank\bigl(\MTypeAVW{t_1}\bigr) \cdot \rank\bigl(\MTypeAVW{t_2}\bigr)\),
  respectively.
  Hence, it remains to handle the case of \#-terms.
For a \#-term \(t(\tx, \ty)\) of the form \( \FOCCount{\tz}{\phi(\tx, \ty, \tz)}\),
  by \cref{lem:cgfoc-preprocessing} applied to \(\phi\),
  we can show that it suffices to consider the case where \(\phi\) is an \(\FO\) formula.
  Moreover, by using Gaifman's Theorem \cite{Gaifman_Locality},
  we can show that it suffices to consider local formulas \(\phi\),
  and we define \(r(t)\) based on the locality radius of \(\phi\).

  We prove the result by induction on the size \(m\)
  of the set that \(r(t)\)-separates \(V\) and \(W\).
  In the induction step, we remove one of the separating vertices from the structure
  and encode the removed information using new relation symbols.
  This leaves us with the case where \(V\) and \(W\) are \(r(t)\)-separated by the empty set.

  Since the formula \(\phi\) is \(r'\)-local for some number \(r' \in \N\),
  the evaluation of \(\sem{\phi(\tv, \tw, \tc)}^\A\) for tuples \(\tv \in V^{\abs{\tx}}\),
  \(\tw \in W^{\abs{\ty}}\), and \(\tc \in A^{\abs{\tz}}\)
  only depends on the neighbourhood structure \(\NeighbA{r'}{\tv \tw \tc}\).
  We transform \(\phi\) into an equivalent formula that is a disjunction
  of mutually exclusive formulas and that corresponds to a case distinction
  on which of the vertices from \(\tv\), \(\tw\), and \(\tc\) appear together
  in the same connected components of \(\NeighbA{r'}{\tv \tw \tc}\).

  Since the formulas are mutually exclusive,
  we can replace the disjunction by a sum of counting terms
  and then use a Feferman--Vaught decomposition \cite{FefermanVaught}
  for each of these counting terms to express it as a sum of two-factor products,
  where each of the factors only depends on the neighbourhood of either \(V\) or \(W\).
  This implies that the matrix \(\MTypeTAVW\) can be interpreted as a sum of matrices
  where the rank of each of the matrices can be bounded
  by some value that does not depend on \(\A\).
The formal proof can be found in \cref{subsec:appendix-cgfoc-rank}.
\end{proof}

We prove \cref{thm:ladder-index-term} by assuming the existence of a long \(t\)-ladder
and showing that this contradicts \cref{lem:cgfoc-rank}.
Intuitively, starting with a long \(t\)-ladder \(\tv_1, \dots, \tv_L, \tw_1, \dots, \tw_L\),
we find a subsequence that is still long,
but where every pair of tuples in the sequence
is \(r\)-separated by the same small set, for a suitable \(r\).
Based on this subsequence, we construct \(r\)-separated sets \(V\) and \(W\),
and we argue that the rank of \(\MTypeTAVW\) has to be large,
which contradicts \cref{lem:cgfoc-rank}.
We refer to \cref{subsec:appendix-ladder-index-term}
for proof details and for the proof of \cref{thm:graph-dimension}
based on \cref{thm:ladder-index-term}.

For a \(\cgFOC\) formula \(\phi(\tx,\ty)\) and a structure \(\A\),
by applying \cref{lem:vc-dim-vs-graph-dim},
it is easy to show that the VC dimension of \(\phi(\tx,\ty)\) in \(\A\)
is equal to the graph dimension of the counting term \(t(\tx,\ty)\deff\FOCCount{()}{\phi}\) in \(\A\).
Hence, \cref{thm:ladder-index-term,thm:graph-dimension} imply the following
(see \cref{subsec:appendix-VC-dimension}).

\begin{restatable}{corollary}{VCdimension}
  \label{cor:VCdimension}
  For every (effectively) nowhere dense graph class \(\C\),
  there is a (computable) function \(f \colon \cgFOC \to \N\)
  such that for every \(\cgFOC\) formula \(\phi(\tx,\ty)\)
  and every \(\sigma\)-structure \(\A\) with \(\sigma \supseteq \sigma(\phi)\) and \(G_\A \in \C\),
  the VC dimension and the ladder index of \(\phi(\tx,\ty)\) in \(\A\) are at most \(f(\phi)\).
\end{restatable}

In addition, for a \(\cgFOC\) formula \(\phi(\tx,\ty)\),
the bound on the rank of \(\STypeTAVW\) for \(t(\tx,\ty) \deff \FOCCount{()}{\phi}\)
from \cref{lem:cgfoc-rank} implies a bound on the size of \(\STypePhiAVW\).
Based on this, it is straightforward to adapt the proof of
\cite[Theorem~1.3]{PilipczukSiebertzTorunczyk_Types2018}
to obtain a bound on the \emph{VC density} of \(\cgFOC\) formulas.
The \emph{VC density} of a formula \(\phi\) on a class of structures \(\C\)
is the infimum of all reals \(\alpha > 0\)
such that \(\absSTypePhiA{A/W} \in \bigO(\abs{W}^\alpha)\)
for all \(\A \in \C\) and all \(W \subseteq A\),
where constants hidden in the \(\bigO\)-notation may depend on \(\alpha\)
(see, \eg, \cite{AschenbrennerDolichHaskellMacphersonStarchenko_2016}).
In \cref{subsec:appendix-number-of-types-nowhere-dense}, we prove the following result:

\begin{restatable}{theorem}{VCdensity}
  \label{thm:VCdensity}
  Let \(\C\) be a nowhere dense graph class,
  let \(\phi(\tx, \ty)\) be a \(\cgFOC\) formula,
  let \(\sigma \supseteq \sigma(\phi)\) be a signature,
  and let \(\C'\) be the class of all \(\sigma\)-structures \(\A\)
  with \(G_\A \in \C\).
  The VC density of \(\phi(\tx, \ty)\) on \(\C'\) is at most \(\abs{\tx}\).
\end{restatable}
 \section{Learning Concepts Defined by \texorpdfstring{$\cgFOC$}{cgFOC}}
\label{sec:learning-formulas}

Next, we apply the results of \cref{sec:graph-dimension} to show that concepts
that can be defined using \(\cgFOC\) expressions can be learned efficiently
on classes of effectively locally bounded expansion.
For this, we consider the logical framework introduced in \cite{GroheTuran_Learnability}
for Boolean classification problems and its extension to multiclass classification from \cite{vanBergeremSchweikardt_Multiclass2025},
more specifically, the setting of \emph{agnostic probably approximately correct (PAC) learning}
\cite{Haussler_PAC}, a generalisation of the \emph{PAC-learning} setting \cite{Valiant_PAC}.
We briefly introduce the considered algorithmic problems here,
and we refer to \cite{GroheRitzert_FO,vanBergerem_FOCNLearning2025} for an in-depth introduction.
Let \(\A\) be a structure, let \(\ell \in \Npos\),
and let \(\D\) be a probability distribution on \(A^\ell \times \Z\).
We are given oracle access to \(\D\),
\ie each oracle query yields an instance from \(A^\ell \times \Z\)
drawn from \(\D\).
We want to find a \emph{hypothesis},
consisting of a counting term \(t(\tx,\ty)\) with \(\abs{\ty} = \ell\)
and a tuple \(\tv\in A^{|\tx|}\),
with small \emph{generalisation error}
\[\err_\D(t, \tv) \deff {\textstyle\Pr_{(\tw, \lambda) \sim \D}}
\bigl(\sem{t(\tv, \tw)}^\A \neq \lambda\bigr),\]
that is, with small probability that our hypothesis disagrees with a randomly drawn instance.
Since a generalisation error of \(0\) might not be possible
(\eg, when \((\tw, \lambda)\) and \((\tw, \lambda')\) both have probability greater than \(0\)
for some \(\tw\) and \(\lambda \neq \lambda'\)), we want to find a hypothesis with a generalisation error close
to the best possible one.
In order to make this possible, it is essential to avoid \emph{overfitting},
where the output hypothesis perfectly fits the examples
drawn from \(\D\),
but the algorithm simply memorised the examples instead of learning an underlying rule.
Hence, we limit the complexity of the hypotheses we are allowed to return
by setting a bound \(m \in \N\) on the length of the counting term \(t\),
setting a bound on the length \(k\) of the tuple \(\tv\),
and fixing a finite set of integers \(I \subset \Z\) that may occur explicitly in \(t\).

Formally, in the \emph{agnostic-PAC-learning problem} for \(\cgFOC\),
we are given a structure \(\A\), \(k \in \N\), \(\ell, m \in \Npos\),
\(\epsilon, \delta \in \Qpos\),
a finite set \(I \subset \Z\),
and oracle access to a probability distribution \(\D\) on \(A^\ell \times \Z\).
We define \(T[\sigma, k, \ell, m, I]\) to be the set of all \(\cgFOC[\sigma]\) counting terms
\(t(x_1, \dots, x_{k'}, y_1, \dots, y_\ell)\)
with \(k' \leq k\), \(\abs{t} \leq m\)
and such that \(t\) only uses variables from
\(\set{x_1, \dots, x_k, y_1, \dots, y_\ell, z_1, \dots, z_m}\)
and integers from \(I\).
Note that \(T[\sigma, k, \ell, m, I]\) is finite and computable from
\(\sigma\), \(k\), \(\ell\), \(m\), \(I\).
An algorithm solving the agnostic-PAC-learning problem may access the oracle
to draw examples \((\tw_1, \lambda_1), (\tw_2, \lambda_2), \ldots \in A^\ell \times \Z\)
independently and identically distributed (\iid) from \(\D\).
Let \(S = \bigl((\tw_1, \lambda_1), (\tw_2, \lambda_2), \dots\bigr)\) be the sequence
of examples drawn by the algorithm.
Depending on the drawn sequence \(S\),
the algorithm should return a counting term \(t_S(\tx, \ty) \in T[\sigma, k, \ell, m, I]\)
and a tuple \(\tv_S \in A^{\abs{\tx}}\) such that,
with probability at least \(1-\delta\)
(over the choice of examples drawn \iid from \(\D\)),
it holds that \(\err_\D(t_S, \tv_S) \leq \epsilon' + \epsilon\),
where \(\epsilon'\) is the minimum generalisation error \(\err_\D(t, \tv)\)
of counting terms \(t(\tx, \ty) \in T[\sigma, k, \ell, m, I]\)
and tuples \(\tv \in A^{\abs{\tx}}\).

We also consider a stronger variant of this problem which we call \emph{agnostic PAC enumeration},
where the input is the same as for the learning problem,
but instead of asking for one hypothesis that is close to optimal,
we ask for an enumeration of \emph{all} hypotheses, sorted by a parameter \(p\)
that differs from the generalisation error of the hypothesis by at most \(\epsilon\).

The parameter \(\epsilon\), also called the accuracy parameter (`approximately correct'),
describes how far the hypothesis we return is allowed to be from an optimal hypothesis.
This allows the returned hypothesis to make a few mistakes,
for example in case of outliers that are manually handled by an optimal solution
but that we do not see in the limited number of examples that we draw in our algorithm.
The other parameter \(\delta\), also called the confidence parameter (`probably'),
describes how confident we are to return a good hypothesis on a randomly drawn sequence of samples.
This handles cases where the randomly drawn samples are not representative for \(\D\).

As shown in \cite[Theorem~5]{DanielySBS_GraphDimensionERM2015},
if a class of hypotheses has bounded graph dimension,
then there is a straightforward way to reduce the agnostic-PAC-learning problem
to a so-called \emph{empirical-risk-minimisation} (ERM) problem,
where the goal is to find a hypothesis with close-to-minimal error on the given sample instead of the entire domain.
By combining a refined version of this result
with our computable bound on the graph dimension of \(\cgFOC\) counting terms
from \cref{thm:graph-dimension}, we can show the following.

\begin{restatable}{theorem}{PACLearnLocallyBoundedExpansion}
  \label{thm:pac-learn-locally-bounded-expansion}
  For every graph class \(\C\) that has effectively locally bounded expansion,
  there is a function \(f\) and an algorithm that does the following.
  Given a signature \(\sigma\), a \(\sigma\)-structure \(\A\) with \(G_\A \in \C\),
  natural numbers \(k \in \N\) and \(\ell, m \in \Npos\) with \(k + \ell \leq m\),
  rational numbers \(\alpha, \epsilon, \delta \in \Qpos\),
  a finite set of integers \(I\),
  and given oracle access to a probability distribution \(\D\) on \(A^\ell \times \Z\),
  the algorithm runs in time
  \(f(\alpha, \epsilon, \delta, \sigma, I, m) \cdot \abs{A}^{\max(k,1 + \alpha)}\)
  and outputs a list consisting of triples \((t, \tv, p)\)
  in ascending lexicographic order on \((p, \abs{\tv}, \abs{t})\)
  such that the list contains exactly one such triple for every
  \(t(\tx, \ty) \in T[\sigma, k, \ell, m, I]\)
  and every tuple \(\tv \in A^{\abs{\tx}}\),
  and it holds that \(p \in \Q\).
Moreover, with probability at least \(1-\delta\) over the oracle calls to \(\D\),
  we have that
  \(\bigabs{\err_\D(t, \tv) - p} \leq \epsilon\)
  holds for every triple \((t, \tv, p)\) in the list.
\end{restatable}

Our algorithm for \cref{thm:pac-learn-locally-bounded-expansion} will draw \(s\) examples
\iid from \(\D\) and then run an ERM algorithm
on these examples.
The main challenge for the proof is to find a suitable number~\(s\).
While \(s\) should be small enough to yield the stated running time,
it should also be large enough to ensure that the output of our algorithm is accurate enough
with high probability.
For this, we use a refinement of \cite[Theorem~5]{DanielySBS_GraphDimensionERM2015},
which is analogous to the result on the uniform convergence property
in~\cite[Sections~6.4--6.5]{Shalev-ShwartzBen-David_UnderstandingMachineLearning}.
We refer to \cref{subsec:appendix-pac-learn} for details.

Next, we give a much more efficient algorithm on classes of bounded degree
that guarantees linear-time preprocessing and constant-delay enumeration.
In particular, the exponent in the running time does not depend on \(k\) any more.
Moreover, the algorithm even works for the much more expressive logic \(\FOC\).
Analogously to \(T[\sigma, k, \ell, m, I]\),
we let \(T_{\FOC}[\sigma, k, \ell, m, I]\) be the set of all \(\FOC[\sigma]\) counting terms
\(t(x_1, \dots, x_{k'}, y_1, \dots, y_\ell)\)
with \(k' \leq k\), \(\abs{t} \leq m\)
and such that \(t\) only uses variables from
\(\set{x_1, \dots, x_k, y_1, \dots, y_\ell, z_1, \dots, z_m}\)
and integers from \(I\),
and this set is also finite and computable from
\(\sigma\), \(k\), \(\ell\), \(m\), and  \(I\).

\begin{restatable}{theorem}{PACLearnBoundedDegree}
  \label{thm:pac-learn-bounded-degree}
  If \(\A\) is known to have degree at most \(d\),
  the algorithm from \cref{thm:pac-learn-locally-bounded-expansion} can be modified
  to solve agnostic-PAC-enumeration for \(\FOC\) counting terms
  with delay \(f(\epsilon, \delta, \sigma, I, m, d)\)
  after preprocessing in time \(f(\epsilon, \delta, \sigma, I, m, d) \cdot \abs{A}\).
\end{restatable}

\begin{proof}[Proof sketch]
  The full proof can be found in \cref{subsec:appendix-pac-learn-bounded-degree}.
Using the Hanf normal form from \cite[Theorem~3.2]{KuskeSchweikardt_FOCN} for \(\FOC\),
  we can transform an \(\FOC\) counting term \(t\) into a set of \(\cgFOC\) counting terms,
  where for every structure \(\A\) of degree at most \(d\),
  one of them is equivalent to \(t\) on \(\A\).
  Since \cref{thm:graph-dimension} bounds the graph dimension of \(\cgFOC\) counting terms,
  \(\FOC\) counting terms also have bounded graph dimension on classes of bounded degree.
  Making use of the fact that the graph dimension bounds the number of required examples,
  we can translate the agnostic-PAC-enumeration problem
  into a standard enumeration problem for \(\FOC\),
  which can be solved using the result from \cite[Corollary~5.7]{KuskeSchweikardt_FOCN}.
\end{proof}

For the last result of this section,
we go back to the well-studied \emph{Boolean} agnostic-PAC-learning problem
(see, \eg, \cite{GroheRitzert_FO}), \ie we consider \(\cgFOC\) formulas instead of counting terms
and let the probability distribution \(\D\) be on \(A^\ell \times \set{0,1}\),
as discussed in \cref{sec:intro}.
Again, the exponent in the running time does not depend on \(k\).
\begin{restatable}{theorem}{learningFormulasPacEnumeration}
  \label{thm:learning-formulas-pac-enumeration}
  The algorithm from \cref{thm:pac-learn-locally-bounded-expansion} can be modified
  so that it solves agnostic-PAC-enumeration for \(\cgFOC\) formulas
  with delay \(f(\alpha, \epsilon, \delta, \sigma, I, m)\) after preprocessing
  in time \(f(\alpha, \epsilon, \delta, \sigma, I, m) \cdot \abs{A}^{1+\alpha}\).
\end{restatable}

The proof relies on the bound on the VC dimension of \(\cgFOC\) formulas from \cref{cor:VCdimension}
and a translation to a standard enumeration problem, this time for \(\cgFOC\) formulas,
which can be solved using \cref{thm:answering-enumeration}.
For details, we refer to \cref{subsec:appendix-learning-formulas}.

 \section{Final Remarks}
\label{sec:conclusion}

The discussion in \cref{rem:weaker-guards} shows that the logic \(\cgFOC\) is in some sense optimal,
as slightly strengthening the logic makes it intractable on very simple sparse classes.
Moreover, the discussion in \cref{rem:dense-classes} shows that such guarded logics
only work well on sparse classes and are too powerful for dense classes.
This, however, does not apply to the logic \(\FOC_1\).
In fact, it is easy to show that \(\FOC_1\) has bounded ladder index
or bounded VC dimension on a class if and only if \(\FO\) has;
and this holds for all classes, including dense ones.
Meanwhile, the computational complexity of \(\FOC_1\) on dense classes remains open.

On sparse classes, we expect the computational complexity of \(\cgFOC\) and \(\FOC_1\) to coincide,
and we believe that the techniques in the recently updated full version
of~\cite{GroheSchweikardt2018} for the evaluation of \(\FOC_1\) on nowhere dense classes
can be generalised to \(\cgFOC\).
It remains an interesting open question whether one can obtain
a multiclass PAC-enumeration algorithm for \(\cgFOC\) on classes of locally bounded expansion,
as in the Boolean case, with almost-linear-time preprocessing and constant delay.

\bibliography{main}

\clearpage
\appendix
\crefalias{section}{appendix}
\crefalias{subsection}{appendix}
\section*{Appendices}

\section{Notation}
\label{sec:appendix-notation}

A \emph{(relational) signature} is a finite set of relation symbols.
Every relation symbol \(R\) has an \emph{arity} \(\ar(R) \in \N\).
Let \(\sigma\) be a signature.
A \emph{\(\sigma\)-structure \(\A\)}
consists of a finite non-empty set \(A\), called the \emph{universe} of \(\A\),
and a relation \(R(\A) \subseteq A^{\ar(R)}\) for every \(R \in \sigma\).
A \emph{(relational) structure} is a \(\sigma\)-structure for some signature \(\sigma\),
and we denote the signature of a structure \(\A\) by \(\sigma(\A)\).
Unless explicitly stated otherwise,
we use \(\A\) and \(\B\) and variations such as \(\A_1\), \(\A'\), etc.~to denote structures,
and we denote their universes by \(A\), \(B\), \(A_1\), \(A'\), etc., respectively.

For a set \(X \subseteq A\), we let the \emph{induced substructure of \(\A\) on \(X\)}
be the \(\sigma\)-structure \(\A[X]\) with universe \(X\)
and \(R(\A[X]) \deff R(\A) \cap X^{\ar(R)}\) for every \(R \in \sigma\).
For a signature \(\sigma' \supseteq \sigma\), a \(\sigma'\)-structure \(\A'\)
is a \emph{\(\sigma'\)-expansion} of \(\A\) if the universe of \(\A'\) is \(A' = A\)
and it holds that \(R(\A') = R(\A)\) for all \(R \in \sigma\).
If \(\A'\) is a \(\sigma'\)-expansion of \(\A\),
then \(\A\) is the \emph{\(\sigma\)-reduct} of \(\A'\).
The \emph{union} of two \(\sigma\)-structures \(\A\) and \(\B\)
is the \(\sigma\)-structure \(\A \cup \B\) with universe \(A \cup B\)
and relations \(R(\A \cup \B) \deff R(\A) \cup R(\B)\).
If $A \cap B = \emptyset$, we write $\A \uplus \B$ instead and call this the \emph{disjoint union} of $\A$ and $\B$.

\begin{definition}[\(\FOC{[}\sigma{]}\)]
  \label[definition]{def:foc}
  The set of \emph{formulas} and \emph{counting terms} for \(\FOC[\sigma]\)
  is built according to the following rules.
\begin{bracketenumerate}
    \item\label{def:fo-atomic}
      \(x_1{=}x_2\) and \(R(x_1, \dots, x_k)\) are formulas for
      \(R \in \sigma\), \(k \deff \ar(R)\),
      and \(x_1, x_2, \dots, x_k \in \vars\).
      In particular, if \(\ar(R) = 0\), then \(R()\) is a formula.
    \item\label{def:fo-bool}
      If \(\phi\) and \(\psi\) are formulas,
      then \(\neg \phi\) and \((\phi \lor \psi)\) are formulas.
    \item\label{def:fo-exists}
      If \(\phi\) is a formula and \(x \in \vars\),
      then \(\exists x\, \phi\) is a formula.
    \item\label{def:foc-countingterm}
      If \(\phi\) is a formula and
      \(\tx = (x_1, \dots, x_k)\) is a tuple of \(k\) pairwise distinct variables
      for some \(k \in \N\),
      then \(\FOCCount{\tx}{\phi}\) is a counting term.
      This includes \(k=0\), so \(\FOCCount{()}{\phi}\) is a counting term.
    \item\label{def:foc-constterm}
      Every integer \(i \in \Z\) is a counting term.
    \item\label{def:foc-plustimesterm}
      If \(t_1\) and \(t_2\) are counting terms,
      then \((t_1 + t_2)\) and \((t_1 \cdot t_2)\) are also counting terms.
    \item\label{def:foc-P}
      If \(\Pred \in \P\), \(m = \ar(\Pred)\)
      and \(t_1, \dots, t_m\) are counting terms,
      then \(\Pred(t_1, \dots, t_m)\) is a formula.
  \end{bracketenumerate}
Let \(\I = (\A, \beta)\) be a \(\sigma\)-interpretation.
  For a formula or counting term \(\xi\) from \(\FOC[\sigma]\),
  the semantics \(\sem{\xi}^\I\) is defined as follows:
\begin{bracketenumerate}
    \item
      \(\sem{x_1{=}x_2}^\I = 1\) if \(\beta(x_1) = \beta(x_2)\),
      and \(\sem{x_1{=}x_2}^\I = 0\) otherwise;
      \(\sem{R(x_1, \dots, x_k)}^\I = 1\) if
      \(\bigl(\beta(x_1), \dots, \beta(x_k)\bigr) \in R(\A)\), and
      \(\sem{R(x_1, \dots, x_k)}^\I = 0\) otherwise.
    \item
      \(\sem{\neg \phi}^\I = 1 - \sem{\phi}^\I\)
      and \(\sem{(\phi \lor \psi)} = \max \set{\sem{\phi}^\I, \sem{\psi}^\I}\).
    \item
      \(\sem{\exists x\, \phi}^\I = \max \setc{\sem{\phi}^{\I\frac{v}{x}}}{v \in A}\).
    \item \(\sem{\FOCCount{\tx}{\phi}}^\I =
      \bigabs{\bigsetc{(v_1, \dots, v_k) \in A^k}
      {\sem{\phi}^{\I \frac{v_1, \dots, v_k}{x_1, \dots, x_k}} = 1}}\),
      where \(\tx = (x_1, \dots, x_k)\).
    \item \(\sem{i}^\I = i\) for \(i \in \Z\).
    \item \(\sem{(t_1 + t_2)}^\I = \sem{t_1}^\I + \sem{t_2}^\I\)
      and \(\sem{(t_1 \cdot t_2)}^\I = \sem{t_1}^\I \cdot \sem{t_2}^\I\).
    \item \(\sem{\Pred(t_1, \dots, t_m)}^\I = 1\)
      if \((\sem{t_1}^\I, \dots, \sem{t_m}^\I) \in \sem{\Pred}\),
      and \(\sem{\Pred(t_1, \dots, t_m)}^\I = 0\) otherwise.
  \end{bracketenumerate}
\end{definition}

We write \((\phi \land \psi)\), \((\phi \rightarrow \psi)\), and \(\forall x\, \phi\)
as shorthands for \(\neg(\neg \phi \lor \neg \psi)\), \((\neg \phi \lor \psi)\),
and \(\neg \exists x\, \neg \phi\).
For counting terms \(t_1\) and \(t_2\),
we write \((t_1 - t_2)\) as a shorthand for \(\bigl(t_1 + ((-1) \cdot t_2)\bigr)\).
Counting terms \(t\) of the form \(\FOCCount{\tx}{\phi}\)
(\ie, obtained by applying rule~\eqref{def:foc-countingterm})
are also called \emph{$\#$-terms}.
An \emph{expression} is a formula or a counting term.
For an expression \(\xi\), we write \(\sigma(\xi)\) for the set of all
relation symbols that occur in \(\xi\).
The \emph{length} \(\abs{\xi}\) of an expression \(\xi\)
is its length when viewed as a word over the alphabet
\(\sigma \cup \vars \cup \P \cup \Z \cup \set{,} \cup
\set{=, \neg, \lor, (, ), \exists, \#, ., +, \cdot}\).

The \emph{free variables} \(\free(\xi)\) of an expression \(\xi\) are inductively defined as follows:
\eqref{def:fo-atomic}
\(\free(x_1{=}x_2) = \set{x_1, x_2}\)
and \(\free\bigl(R(x_1, \dots, x_k)\bigr) = \set{x_1, \dots, x_k}\);
\eqref{def:fo-bool}
\(\free(\neg \phi) = \free(\phi)\)
and \(\free\bigl((\phi \lor \psi)\bigr) = \free(\phi) \cup \free(\psi)\);
\eqref{def:fo-exists}
\(\free(\exists x\, \phi) = \free(\phi) \setminus \set{x}\);
\eqref{def:foc-countingterm}
\(\free\bigl(\FOCCount{(x_1, \dots, x_k)}{\phi}\bigr)
= \free(\phi) \setminus \set{x_1, \dots, x_k}\);
\eqref{def:foc-constterm}
\(\free(i) = \emptyset\) for \(i \in \Z\);
\eqref{def:foc-plustimesterm}
\(\free\bigl((t_1 + t_2)\bigr) = \free\bigl((t_1 \cdot t_2)\bigr) = \free(t_1) \cup \free(t_2)\);
\eqref{def:foc-P} \(\free\bigl(\Pred(t_1, \dots, t_m)\bigr) = \bigcup_{i=1}^m \free(t_i)\).
We write \(\xi(z_1, \dots, z_k)\) to indicate that
\(\free(\xi) \subseteq \set{z_1, \dots, z_k}\).
A \emph{sentence} is a formula without free variables, and
a \emph{ground term} is a counting term without free variables.

The set of formulas for \(\FO[\sigma]\)
is built according to the rules \eqref{def:fo-atomic}--\eqref{def:fo-exists}.
The \emph{quantifier rank} \(\qr(\phi)\) of an \(\FO[\sigma]\) formula \(\phi\)
is the maximum nesting depth of constructs using rule~\eqref{def:fo-exists}
in order to construct \(\phi\).
By \(\FOC\), we denote the union of all \(\FOC[\sigma]\) for arbitrary signatures \(\sigma\).
This applies analogously to \(\FO\).
 \section{Further Preliminaries Used in the Appendices}
\label{sec:appendix-further-prelims}

Throughout the appendices,
for every signature \(\sigma\) and every \(\sigma\)-structure \(\A\),
we will assume that \(\top, \bot \in \sigma\) are \(0\)-ary relation symbols
with \(\top(\A) = \set{()}\) and \(\bot(\A) = \emptyset\).
This assumption does not weaken the results stated in the main part of the paper,
and they also hold without the assumption.
For the statements and proofs in the appendices, however,
it will be convenient to work with the formulas \(\top()\) and \(\bot()\) of quantifier rank \(0\),
where we always have \(\A \models \top()\) and \(\A \not\models \bot()\).
For convenience, we also use $\top$ or $\bot$ directly as formulas to mean $\top()$ or $\bot()$, respectively.

For a \(\cgFOC[\sigma]\) formula \(\phi(\tx, \ty)\), a \(\sigma\)-structure \(\A\),
and a tuple \(\tw \in A^{\abs{\ty}}\), we let \[\sem{\FOCCount{\tx}{\phi(\tx, \tw)}}^\A
\deff \sem{\FOCCount{\tx}{\phi(\tx, \ty)}}^{(\A, \beta)}\]
for one (and hence every) assignment \(\beta\)
with \(\beta(y_i) = w_i\) for all \(i \in [\abs{\ty}]\).

\begin{theorem}[Feferman--Vaught \cite{FefermanVaught}]
  \label{thm:FV}
  For every \(\FO\) formula \(\phi(\tx, \ty)\),
  where \(\tx\) and \(\ty\) are tuples of pairwise distinct variables,
  there is a finite set \(\Delta_{\phi, \tx, \ty}\) of pairs of \(\FO[\sigma(\phi)]\) formulas
  that is computable from \(\phi(\tx, \ty)\) such that the following holds.

  For all \(\sigma\)-structures \(\A, \B\) for some \(\sigma \supseteq \sigma(\phi)\)
  with \(A \cap B = \emptyset\)
  and for all \(\tv \in A^{\abs{\tx}}\) and \(\tw \in B^{\abs{\ty}}\),
  it holds that \(\A \uplus \B \models \phi[\tv, \tw]\) if and only if
  there is a pair \((\alpha(\tx), \beta(\ty)) \in \Delta_{\phi, \tx, \ty}\)
  with \(\A \models \alpha[\tv]\) and \(\B \models \beta[\tw]\).
  Moreover, for every pair \((\alpha, \beta) \in \Delta_{\phi, \tx, \ty}\),
  it holds that \(\qr(\alpha) \leq \qr(\phi)\) and \(\qr(\beta) \leq \qr(\phi)\),
  and the pairs in \(\Delta_{\phi, \tx, \ty}\) are mutually exclusive,
  that is, for two distinct pairs \((\alpha, \beta), (\alpha', \beta') \in \Delta_{\phi, \tx, \ty}\),
  it holds that \(\alpha \land \alpha'\) is unsatisfiable.

  We call \(\Delta_{\phi, \tx, \ty}\) a \emph{Feferman--Vaught decomposition}
  of \(\phi\) with respect to \(\tx, \ty\).
\end{theorem}

For a signature \(\sigma\), a \(\sigma\)-structure \(\A\), numbers \(k,q \in \N\),
and a tuple \(\tv \in A^k\), we let
\[\tp_q(\A, \tv) \deff \bigsetc{\phi(x_1, \dots, x_k) \in \FO[\sigma]}
{\qr(\phi) \leq q, \A \models \phi[\tv]}.\]
Moreover, for \(r \in \N\),
we let \(\ltp_{q, r}(\A, \tv) \deff \tp_q\bigl(\NrA{\tv}, \tv\bigr)\).
As a direct consequence of \cref{thm:FV}, we have the following
(cf.~\cite[Lemma~2.3]{Grohe_LogicGraphsAlgorithms}).

\begin{lemma}
  \label{lem:FV}
  For every signature \(\sigma\),
  for all \(k, \ell, q \in \N\),
  for all \(\sigma\)-structures \(\A, \B\) with \(A \cap B = \emptyset\),
  and for all tuples \(\tv \in A^k\) and \(\tw \in B^\ell\),
  \(\tp_q(\A \uplus \B, \tv \tw)\)
  is uniquely determined by
  \(\tp_q(\A, \tv)\) and \(\tp_q(\B, \tw)\).
\end{lemma}

For \(r \in \N\) and a signature \(\sigma\),
let \(\dist^\sigma_{\leq r}(x,y)\) be an \(\FO[\sigma]\) formula
such that for every \(\sigma\)-structure \(\A\)
and all \(v,w \in A\),
we have \(\A \models \dist^\sigma_{\leq r}[v,w]\) if and only if
\(\dist^{G_\A}(v,w) \leq r\).
Moreover, we set \(\dist^\sigma_{> r}(x,y) \deff \neg \dist^\sigma_{\leq r}(x,y)\).
We can define the formulas in such a way that
\(\qr\bigl(\dist^\sigma_{\leq r}\bigr) = \qr\bigl(\dist^\sigma_{> r}\bigr)
\leq \max_{R \in \sigma} \ar(R) \cdot \lceil\log(r+1)\rceil\).

A \emph{basic local sentence} is an \(\FO\) sentence of the form
\[\exists y_1 \dots \exists y_k\;\bigl(\Land_{1 \leq i < j \leq k} \dist^\sigma_{> 2r}(y_i, y_j)
\land \Land_{i \in [k]} \psi(y_i)\bigr),\]
where \(k \in \Npos\), \(r \in \N\), and \(\psi(y)\) is an \(r\)-local formula.
The number \(r\) is called the \emph{radius} of the basic local sentence.

For \(r \in \N\), an \(\FO\) formula is in \emph{Gaifman normal form of radius at most \(r\)}
if it is a Boolean combination of \(r\)-local formulas
and basic local sentences of radius at most \(r\).
We say that an \(\FO\) formula is in Gaifman normal form
if it is in Gaifman normal form of radius at most \(r\) for some \(r \in \N\).

\begin{theorem}[Gaifman \cite{Gaifman_Locality}]
  \label{thm:gaifman}
  There is an algorithm that, given an \(\FO\) formula \(\phi\),
  outputs an \(\FO\) formula \(\phi'\) in Gaifman normal form with radius at most \(7^{\qr(\phi)}\)
  such that \(\phi'\) is equivalent to \(\phi\).
\end{theorem}

As a consequence of \cref{thm:gaifman}, we obtain the following.

\begin{lemma}
  \label{lem:gaifman}
  There is a computable function \(q' \colon \FO \to \N\)
  such that for every signature \(\sigma\),
  every \(k \in \Npos\),
  every formula \(\phi(x_1, \dots, x_k) \in \FO[\sigma]\),
  for \(q \deff q'(\phi)\) and \(r \deff 7^{\qr(\phi)}\),
  for every \(\sigma\)-structure \(\A\),
  and all \(\tv, \tv' \in A^k\)
  with \(\ltp_{q, r}(\A, \tv) =\ltp_{q, r}(\A, \tv')\),
  it holds that \(\A \models \phi[\tv]\) if and only if \(\A \models \phi[\tv']\).
\end{lemma}
\begin{proof}
  Let \(\phi(x_1, \dots, x_k) \in \FO[\sigma]\),
  and let \(\phi'(x_1, \dots, x_k)\) be the formula computed via \cref{thm:gaifman} from \(\phi\)
  such that \(\phi'\) is equivalent to \(\phi\)
  and \(\phi'\) is in Gaifman normal form with radius at most \(r = 7^{\qr(\phi)}\).
  We let \(q \deff q'(\phi) \deff \qr(\phi')\), which is computable from \(\phi\).

  Now let \(\A\) be a \(\sigma\)-structure,
  and let \(\phi'_\A(\tx)\) be the formula obtained from \(\phi'(\tx)\)
  by replacing every basic local sentence \(\chi\) occurring in \(\phi'\) by \(\top\)
  if \(\A \models \chi\), and replacing \(\chi\) by \(\bot\) if \(\A \not\models \chi\).
  Then, for all \(\tv \in A^k\), it holds that \(\A \models \phi[\tv]\)
  if and only if \(\A \models \phi'[\tv]\) if and only if \(\A \models \phi'_\A[\tv]\).
  Moreover, \(\phi'_\A(\tx)\) is an \(r\)-local formula of quantifier rank
  \(\qr(\phi'_\A) \leq \qr(\phi') = q\).
  Hence, \(\A \models \phi[\tv]\) if and only if \(\phi'_\A \in \ltp_{q, r}(\A, \tv)\).

  This shows that
  we have \(\A \models \phi[\tv]\) if and only if \(\A \models \phi[\tv']\)
  for all \(\tv, \tv' \in A^k\)
  with \(\ltp_{q, r}(\A, \tv) =\ltp_{q, r}(\A, \tv')\).
\end{proof}

\begin{remark}
  \label{rem:finite-fo}
  We can define a syntactic normal form for first-order logic
  such that for all signatures \(\sigma\) and numbers \(k, q \in \N\),
  there are only finitely many \(\FO[\sigma]\) formulas \(\phi\) in this normal form
  with \(\qr(\phi) \leq q\) and \(\free(\phi) \subseteq \set{x_1, \dots, x_k}\).
  Let \(\Phi[\sigma, k, q]\) denote the set of all these formulas.
  The normal form can be defined such that
  \(\Phi[\sigma, k, q]\) is computable from \(\sigma\), \(k\), and \(q\)
  and there is an algorithm that, given a first-order formula \(\phi\)
  with \(\free(\phi) \subseteq \set{x_1, \dots, x_k}\),
  outputs an equivalent formula \(\phi' \in \Phi[\sigma(\phi), k, \qr(\phi)]\).
\end{remark}
 \section{Details Omitted in Section~\ref{sec:cgfoc}}
\label{sec:appendix-cgfoc}

Based on \cref{def:foc}, the definition of \(\cgFOC\) looks as follows.

\begin{definition}[\(\cgFOC{[}\sigma{]}\)]
  \label[definition]{def:cgfoc}
  The set of \emph{formulas} and \emph{counting terms} for \(\cgFOC[\sigma]\)
  (clique-guarded first-order logic with counting) is built according to
  rules~\eqref{def:fo-atomic}--\eqref{def:foc-plustimesterm} of \cref{def:foc}
  and the following modified version of rule~\eqref{def:foc-P}.
  \begin{enumerate}[(cg1)]
    \setcounter{enumi}{6}
    \renewcommand{\labelenumi}{\textbf{(\theenumi)}}
    \renewcommand{\theenumi}{cg\arabic{enumi}}
    \item\label{def:cgfoc-P}
      If \(\Pred \in \P\), \(m = \ar(\Pred)\),
      \(k \in \N\),
      \(R_1, \dots, R_k \in \sigma\),
      \(\tz_i \in \vars^{\ar(R_i)}\) for all \(i \in [k]\),
      and \(t_1, \dots, t_m\) are counting terms
      such that for all \(z, z' \in \bigcup_{p=1}^m \free(t_p)\) with \(z \neq z'\),
      there is some \(i \in [k]\) with \(z, z' \in \tilde{z}_i\),
      then \(\bigl(\Land_{i=1}^k R_i(\tz_i) \land \Pred(t_1, \dots, t_m)\bigr)\) is a formula.
  \end{enumerate}
\end{definition}

\subsection{Details on Remark~\ref{rem:clique-guarded}}
\label{subsec:appendix-cgfoc-clique-guarded}

Let \(\sigma\) be a signature, let \(\Pred \in \P\), \(m \deff \ar(\Pred)\),
and let \(t_1, \dots, t_m\) be \(\cgFOC[\sigma]\) counting terms.
We set \(X \deff \bigcup_{i=1}^m \free(t_i)\)
and \(\phi' \deff \Land_{\set{x,y} \in \binom{X}{2}} \chi(x,y)\) for
\[\chi(x,y) \deff\ x=y \lor \Lor_{R \in \sigma} \exists z_1 \dots \exists z_{\ar(R)}
\Lor_{\substack{i,j \in [\ar(R)],\\ i \neq j}} R(\tz_{i \to x, j \to y}),\]
where the tuple \(\tz_{i \to x, j \to y}\) is obtained from \((z_1, \dots, z_{\ar(R)})\)
by replacing \(z_i\) with \(x\) and \(z_j\) with \(y\).
Then, for every \(\sigma\)-interpretation \(\I = (\A, \beta)\),
it holds that \(\I \models \phi'\) if and only if \(\setc{\beta(x)}{x \in X}\)
forms a clique in \(G_\A\).

We transform \(\phi'\) into an equivalent formula \(\phi''=\Lor_{i \in I} \exists z_1 \dots \exists z_{k_i}
\Land_{\set{x,y} \in \binom{X}{2}} \psi_{i,x,y}\),
where \(I\) is a finite index set,
\(k_i \in \N\) for all \(i \in I\),
and \(\psi_{i,x,y}\) is an atomic formula
with \(x,y \in \free(\psi_{i,x,y})\) for all \(i \in I\) and \(\set{x,y} \in \binom{X}{2}\).

Then, for all \(i \in I\), we can transform the \(\FOC[\sigma]\) formula
\(\phi'_i \deff \Land_{\set{x,y} \in \binom{X}{2}} \psi_{i,x,y} \land P(t_1, \dots, t_m)\)
into an equivalent \(\cgFOC[\sigma]\) formula \(\phi_i\)
by iteratively dropping atomic formulas of the form \(x=y\) and replacing \(y\) with \(x\)
in the remaining formula.

Hence, \(\phi \deff \Lor_{i \in I} \exists z_1 \dots \exists z_{k_i} \phi_i\)
is a \(\cgFOC[\sigma]\) formula such that, for all \(\sigma\)-interpretations \(\I = (\A, \beta)\),
it holds that \(\I \models \phi\)
if and only if \((\sem{t_1}^\I, \dots, \sem{t_m}^\I) \in \sem{\Pred}\)
and \(\setc{\beta(x)}{x \in \bigcup_{i=1}^m \free(t_i)}\) forms a clique in \(G_\A\).

\subsection{Proof of Lemma~\ref{lem:cgfoc-preprocessing}}
\label{subsec:appendix-cgfoc-preprocessing}

\cgfocPreprocessing*

\begin{proof}
  We prove this result by structural induction on \(\phi\).
If \(\phi\) is built without the use of rule~\eqref{def:cgfoc-P} of \cref{def:cgfoc},
  then \(\phi\) is an \(\FO[\sigma]\) formula,
  and we can simply set \(\sigma_\phi \deff \sigma\),
  \(\phi' \deff \phi\), and \(\A_\phi \deff \A\).

  If \(\phi\) is the result of applying rule~\eqref{def:cgfoc-P},
  then \(\phi\) is of the form
  \(\Land_{i=1}^k R_i(z_{i,1}, \dots, z_{i, \ar(R_i)}) \land \Pred(t_1, \dots, t_m)\).
  Let \(\tx'\) be a tuple of pairwise distinct variables such that
  \(\tilde{x}' = \bigcup_{i=1}^m \free(t_i)\).
We set \(\sigma_\phi \deff \sigma \uplus \set{R_\phi}\)
  for a fresh relation symbol \(R_\phi\) of arity \(\ell \deff \ar(R_\phi) \deff \abs{\tx'}\),
  and we set \(\phi'(\tx) \deff
  \Land_{i=1}^k R_i(z_{i,1}, \dots, z_{i, \ar(R_i)}) \land R_\phi(\tx')\).
  Furthermore, we let \(\A_\phi\) be the \(\sigma_\phi\)-expansion of \(\A\)
  where \(R_\phi(\A_\phi)\) is the set of all tuples \(\tw \in A^\ell\)
  that form a clique in \(G_\A\) and where
  \((\sem{t_1}^{(\A, \beta)}, \dots, \sem{t_m}^{(\A, \beta)}) \in \sem{\Pred}\) for \(\beta(x'_i) \deff w_i\) for all \(i \in [\ell]\).
  Since all tuples of \(R_\phi(\A_\phi)\) form a clique in \(G_\A\),
  it holds that \(G_\A = G_{\A_\phi}\).
  Moreover, for all \(\tv \in A^{\abs{\tx}}\),
  it holds that \(\A \models \phi[\tv]\) if and only if \(\A_\phi \models \phi'[\tv]\).
  It is easy to see that \(\A_\phi\) can be computed from \(\A\) and \(\phi\) by brute force.

  If \(\phi\) is of the form \(\neg \psi\),
  then we apply the induction hypothesis on \(\psi\) to obtain
  a signature \(\sigma_\psi \supseteq \sigma\),
  an \(\FO[\sigma_\psi]\) formula \(\psi'(\tx)\),
  and a \(\sigma_\psi\)-expansion \(\A_\psi\) of \(\A\)
  as described in \cref{lem:cgfoc-preprocessing}.
  We set \(\sigma_\phi \deff \sigma_\psi\),
  \(\phi' \deff \neg \psi'\),
  and \(\A_\phi \deff \A_\psi\).

  If \(\phi\) is of the form \(\psi_1 \lor \psi_2\),
  then we apply the induction hypothesis on \(\psi_1\) and \(\psi_2\) to obtain
  signatures \(\sigma_{\psi_1}, \sigma_{\psi_2} \supseteq \sigma\),
  formulas \(\psi_1'(\tx), \psi_2'(\tx)\),
  and structures \(\A_{\psi_1}, \A_{\psi_2}\)
  as described in \cref{lem:cgfoc-preprocessing}.
  Without loss of generality, by possibly renaming the relation symbols in \(\sigma_{\psi_2} \setminus \sigma\),
  we may assume that \(\sigma_{\psi_1} \cap \sigma_{\psi_2} = \sigma\).
  We set \(\sigma_\phi \deff \sigma_{\psi_1} \cup \sigma_{\psi_2}\),
  \(\phi' \deff \psi'_1 \lor \psi'_2\),
  and we let \(\A_\phi\) be the \(\sigma_\phi\)-expansion of \(\A_{\psi_1}\) and \(\A_{\psi_2}\).

  Finally, if \(\phi\) is of the form \(\exists x\, \psi\),
  we proceed analogously to the case \(\neg \psi\),
  and we set \(\sigma_\phi \deff \sigma_\psi\),
  \(\phi' \deff \exists x\, \psi'\),
  and \(\A_\phi \deff \A_\psi\).
\end{proof}

\subsection{Proof of Lemma~\ref{lem:nowhere-dense-cliques}}
\label{subsec:appendix-nowhere-dense-cliques}

In this subsection, we prove \cref{lem:nowhere-dense-cliques}.
Intuitively, the lemma says that the number of cliques of a given size in graphs from a nowhere dense class
is almost linear.
Moreover, also the representation size of structures from a nowhere dense class is almost linear.

\nowhereDenseCliques*

For \(d \in \N\), a graph \(G\) is \emph{\(d\)-degenerate}
if every non-empty subgraph of \(G\) contains a vertex of degree at most \(d\).
\begin{lemma}[{\cite[Lemma~3.1]{NesetrilOssonaDeMendez_Sparsity2012}}]
  \label{lem:degeneracy-cliques}
  Let \(d \in \N\), \(k \in \Npos\), and let \(G\) be a \(d\)-degenerate graph.
  Then \(G\) includes at most \(\binom{d}{k-1} \abs{V(G)}\) cliques of size \(k\).
\end{lemma}

The following result gives a bound on the so-called \emph{weak \(r\)-colouring number}
\(\wcol_r(G)\) of graphs \(G\) in nowhere dense classes.

\begin{lemma}[\cite{NesetrilOssonaDeMendez_NowhereDense2011,Zhu_Colouring2009},
  cf.~{\cite[Lemma~6.3]{GroheKreutzerSiebertz_NowhereDense2017}}]
  \label{lem:nowhere-dense-weak-colouring}
  Let \(\C\) be a nowhere dense graph class.
  There is a function \(f \colon \N \times \Qpos \to \N\)
  such that for every \(r \in \N\),
  every \(\epsilon \in \Qpos\),
  and every graph \(G \in \C\) with \(n \deff \abs{V(G)} \geq f(r, \epsilon)\),
  it holds that \(\wcol_r(G) \leq n^\epsilon\).
  Furthermore, if \(\C\) is effectively nowhere dense, then \(f\) is computable.
\end{lemma}

For every graph \(G\), it holds that \(\wcol_1(G) = \col(G)\),
and \(G\) is \((\col(G)-1)\)-degenerate.
For details, including definitions of \(\wcol_r(G)\) and \(\col(G)\),
see \cite[Section~4.9]{NesetrilOssonaDeMendez_Sparsity2012}.
Hence, \cref{lem:nowhere-dense-weak-colouring} implies the following
bound on the degeneracy.

\begin{lemma}
  \label{lem:nowhere-dense-degeneracy}
  Let \(\C\) be a nowhere dense graph class.
  There is a function \(f \colon \Qpos \to \N\)
  such that for every \(\epsilon \in \Qpos\)
  and every graph \(G \in \C\) with \(n \deff \abs{V(G)} \geq f(\epsilon)\),
  it holds that \(G\) has degeneracy at most \(n^\epsilon\).
  Furthermore, if \(\C\) is effectively nowhere dense, then \(f\) is computable.
\end{lemma}

By combining \cref{lem:degeneracy-cliques,lem:nowhere-dense-degeneracy},
we can now prove \cref{lem:nowhere-dense-cliques}.

\begin{proof}[Proof of \cref{lem:nowhere-dense-cliques}]
  Let \(\C\) be a nowhere dense graph class,
  and let the function \(f' \colon \Qpos \to \N\) be the one from \cref{lem:nowhere-dense-degeneracy}.
  Without loss of generality, \(f'\) is monotonically decreasing,
  and \(f'(\epsilon) \geq 1\) for every \(\epsilon \in \Qpos\).

  For every \(\epsilon \in \Qpos\)
  and every graph \(G \in \C\) with \(n \deff \abs{V(G)} \geq f'(\epsilon)\),
  it holds that \(G\) has degeneracy at most \(n^\epsilon\).
Furthermore, for every graph \(G \in \C\) with \(\abs{V(G)} < f'(\epsilon)\),
  it holds that \(G\) has degeneracy less than \(f'(\epsilon)\).
  Hence, since all graphs we consider are non-empty,
  the degeneracy of \emph{every} graph \(G \in \C\)
  is bounded by \(f'(\epsilon) \cdot \abs{V(G)}^\epsilon\).

  Let \(G \in \C\), \(n \deff \abs{V(G)}\),
  let \(k \in \Npos\), \(\epsilon \in \Qpos\), and let \(\delta \deff \epsilon/k\).
  Then \(G\) has degeneracy at most \(f'(\delta) \cdot n^{\delta}\).
  By \cref{lem:degeneracy-cliques}, we know that \(G\) includes at most
  \[\binom{f'(\delta) \cdot n^{\delta}}{k-1} \cdot n
  \leq (f'(\delta))^k \cdot n^{1 + k \cdot \delta}
  = f_1(k, \epsilon) \cdot n^{1+\epsilon}\]
  cliques of size \(k\) for \(f_1(k, \epsilon) \deff (f'(\epsilon/k))^k\).
  This proves \cref{item:nowhere-dense-cliques} of \cref{lem:nowhere-dense-cliques}.
  Note that \(f_1(k, \epsilon)\) is monotonically decreasing in \(\epsilon\)
  and monotonically increasing in \(k\).

  Now let \(\sigma\) be a signature, and let \(\A\) be a \(\sigma\)-structure
  with \(G_\A \in \C\).
  Let \(R \in \sigma\) with \(\ar(R) \geq 1\).
  For every tuple \((v_1, \dots, v_{\ar(R)}) \in R(\A)\),
  the vertices \(v_1, \dots, v_{\ar(R)}\) form a clique in \(G_\A\) of size at most \(\ar(R)\).
  Furthermore, for every clique in \(G_\A\) of size at most \(\ar(R)\),
  there are at most \((\ar(R))^{\ar(R)}\) many tuples in \(R(\A)\)
  that only contain vertices from the clique.
  This shows that
  \(\abs{R(\A)} \leq (\ar(R))^{\ar(R)} \cdot f_1(\ar(R), \epsilon) \cdot \abs{A}^{1+\epsilon}\)
  for every \(\epsilon \in \Qpos\).
Let \(k\) be the maximum arity of a relation symbol in \(\sigma\).
  Then
  \begin{align*}
    \norm{\A}
    &= \abs{A} + \sum_{R \in \sigma} \max\set{\ar(R),1} \cdot \abs{R(\A)}\\
    &\leq \abs{A} + \abs{\sigma} \cdot (k+1) \cdot \max\setc{\abs{R(\A)}}{R\in\sigma}\\
    &\leq \abs{A} + \abs{\sigma} \cdot (k+1) \cdot (\max\{k,1\})^k \cdot f_1(k, \epsilon)
    \cdot \abs{A}^{1+\epsilon}\\
    &\leq \abs{A} + \abs{\sigma} \cdot (k+1) \cdot (k+1)^k \cdot f_1(k, \epsilon)
    \cdot \abs{A}^{1+\epsilon}.
  \end{align*}
  \Cref{item:nowhere-dense-size} of \cref{lem:nowhere-dense-cliques}
  follows with
  \(f(k, \epsilon) \deff (k+1)^{k+1} \cdot f_1(k, \epsilon) + 1\).
Finally, if \(\C\) is effectively nowhere dense,
  then \(f'\) is computable, and so are \(f_1\) and \(f\).
\end{proof}

\subsection{Proof of Lemma~\ref{lem:cgfoc-preprocessing-bounded-expansion}}
\label{subsec:appendix-cgfoc-preprocessing-bounded-expansion}

In the proof of \cref{lem:cgfoc-preprocessing-bounded-expansion},
we will use the following result.

\begin{theorem}
  \label{thm:fo-counting-testing-locally-bounded-expansion}
  Let \(\C\) be a graph class of locally bounded expansion.
  There is a function \(f\) and an algorithm that does the following.
  Given an \(\FO\) formula \(\phi(\tx, \ty)\),
  a \(\sigma(\phi)\)-structure \(\A\) with \(G_\A \in \C\),
  and given an \(\epsilon \in \Qpos\),
  after preprocessing in time \(f(\phi, \epsilon) \cdot \abs{A}^{1+\epsilon}\),
  the algorithm can answer the following queries in time \(f(\phi, \epsilon)\):
  given a tuple \(\tw \in A^{\abs{\ty}}\), output \(\sem{\FOCCount{\tx}{\phi(\tx, \tw)}}^\A\).
\end{theorem}

This generalises the following theorem,
which is a special case of \cite[Theorem~8]{Torunczyk_Aggregate2020},
from classes of bounded expansion to classes of locally bounded expansion.

\begin{theorem}
  \label{thm:fo-counting-testing-bounded-expansion}
  Let \(\C\) be a graph class of bounded expansion.
  There is a function \(f\) and an algorithm that does the following.
  Given an \(\FO\) formula \(\phi(\tx, \ty)\)
  and a \(\sigma\)-structure \(\A\) for some \(\sigma \supseteq \sigma(\phi)\)
  with \(G_\A \in \C\),
  after preprocessing in time \(f(\phi, \sigma) \cdot \abs{A}\),
  the algorithm can answer the following queries in time \(f(\phi, \sigma)\):
  given a tuple \(\tw \in A^{\abs{\ty}}\), output \(\sem{\FOCCount{\tx}{\phi(\tx, \tw)}}^\A\).
\end{theorem}

In the proof of \cref{thm:fo-counting-testing-locally-bounded-expansion},
we make use of neighbourhood covers, which are defined as follows.
For \(r \in \N\), a \emph{distance-\(r\) neighbourhood cover} of a graph \(G\)
is a mapping \(\K \colon V(G) \to 2^{V(G)}\)
of vertices of \(G\) to subsets of vertices of \(G\), called \emph{clusters},
such that, for each vertex \(v \in V(G)\),
it holds that \(\neighbr{G}{v} \subseteq \K(v)\).
We write \(C \in \K\) to express that \(C\) is a cluster of \(\K\),
that is, \(C = \K(v)\) for some \(v \in V(G)\).
The \emph{radius} of \(\K\) is the maximum radius of all subgraphs \(G[C]\)
for a cluster \(C \in \K\).
The \emph{overlap} of \(\K\) is the maximum number of clusters in \(\K\) that contain the same vertex,
that is, \(\overlap(\K) \deff \max_{v \in V(G)} \bigabs{\setc{C \in \K}{v \in C}}\).
Note that \(\sum_{C \in \K} \abs{C} \leq \abs{V(G)} \cdot \overlap(\K)\).

As shown in \cite{GroheKreutzerSiebertz_NowhereDense2017}, for nowhere dense classes,
we can efficiently compute neighbourhood covers with small radius and small overlap.

\begin{theorem}[{\cite[Theorem~6.2]{GroheKreutzerSiebertz_NowhereDense2017}}]
  \label{thm:sparse-neighbourhood-cover}
  Let \(\C\) be a nowhere dense graph class.
  There is a function \(f\) and an algorithm that does the following.
  Given an \(\epsilon \in \Qpos\), an \(r \in \N\),
  and a graph \(G \in \C\),
  the algorithm computes a distance-\(r\) neighbourhood cover of \(G\)
  of radius at most \(2r\) and overlap at most \(f(r, \epsilon) \cdot n^\epsilon\)
  in time \(f(r, \epsilon) \cdot n^{1+\epsilon}\).
  Furthermore, if \(\C\) is effectively nowhere dense, then \(f\) is computable.
\end{theorem}

As remarked in \cite[Section~8.1]{GroheSchweikardt2018},
for a given neighbourhood cover \(\K\) of \(G\),
we can compute in linear time a data structure that associates with each \(C \in \K\)
the list of all \(v \in V(G)\) with \(\K(v) = C\).
We note that \cite[Theorem~6.2]{GroheKreutzerSiebertz_NowhereDense2017}
states \cref{thm:sparse-neighbourhood-cover} for graphs of size \(\abs{V(G)} \geq f(r, \epsilon)\).
However, for graphs of size \(\abs{V(G)} < f(r, \epsilon)\),
a neighbourhood cover contains less than \(f(r, \epsilon)\) clusters,
so every vertex has overlap less than \(f(r, \epsilon)\).
Hence, the result also holds for graphs of size less than \(f(r, \epsilon)\).

As the last ingredient for the proof of
\cref{thm:fo-counting-testing-locally-bounded-expansion},
we will use the following result that allows us to test
whether a given tuple satisfies a formula
on nowhere dense classes.

\begin{theorem}[Testing, {\cite[Lemma~2.2, Corollary~2.4]{SchweikardtSegoufinVigny_Enumeration2022}}]
  \label{thm:fo-testing}
  Let \(\C\) be a nowhere dense graph class.
  There is a function \(f\) and an algorithm that does the following.
  Given an \(\FO\) formula \(\phi(\tx)\),
  a \(\sigma\)-structure \(\A\) for some \(\sigma \supseteq \sigma(\phi)\)
  with \(G_\A \in \C\),
  and a rational \(\epsilon \in \Qpos\),
  after preprocessing in time \(f(\phi, \sigma, \epsilon) \cdot \abs{A}^{1+\epsilon}\),
  the algorithm can answer the following queries in time \(f(\phi, \sigma, \epsilon)\):
  given a tuple \(\tv \in A^{\abs{\tx}}\), decide whether \(\A \models \phi[\tv]\).
  If \(\C\) is effectively nowhere dense, then \(f\) is computable.
\end{theorem}

We remark that \cite[Corollary~15]{SegoufinVigny_LocallyBoundedExpansion2017}
states a similar result for classes of locally bounded expansion.
However, \cite[Corollary~15]{SegoufinVigny_LocallyBoundedExpansion2017}
only gives a uniform algorithm over all \(\epsilon \in \Qpos\)
if the class \(\C\) has effectively locally bounded expansion.
Thus, we decided to use the uniform result from \cite{SchweikardtSegoufinVigny_Enumeration2022}
instead.

For \(k \in \N\),
let \(\G_k\) be the set of all graphs \(G\)
with vertex set \(V(G) = [k]\).
Furthermore, for \(k, r \in \N\), for \(G \in \G_k\),
and for a signature \(\sigma\), we let
\[\delta^\sigma_{G, r}(x_1, \dots, x_k)
  \deff \Land_{\set{i,j} \in E(G)}\mkern-18mu \dist^\sigma_{\leq r}(x_i, x_j)
\ \land\ \mkern-18mu\Land_{\set{i,j} \not\in E(G)}\mkern-18mu \dist^\sigma_{> r}(x_i, x_j).\]

We first prove the following special case of
\cref{thm:fo-counting-testing-locally-bounded-expansion}.

\begin{lemma}
  \label{lem:fo-dist-counting}
  Let \(\C\) be a graph class of locally bounded expansion.
  There is a function \(f\) and an algorithm that does the following.
  Given an \(\FO\) formula
  \(\psi(\tx, \ty) = \delta^{\sigma(\phi)}_{G, 2r+1}(\tx, \ty) \land \phi(\tx, \ty)\),
  where \(r \in \N\),
  \(G \in \G_{\abs{\tx} + \abs{\ty}}\),
  and \(\phi\) is an \(r\)-local formula,
  given a \(\sigma(\phi)\)-structure \(\A\) with \(G_\A \in \C\),
  and given an \(\epsilon \in \Qpos\),
  after preprocessing in time \(f(\psi, \epsilon) \cdot \abs{A}^{1+\epsilon}\),
  the algorithm can answer the following queries in time \(f(\psi, \epsilon)\):
  given a tuple \(\tw \in A^{\abs{\ty}}\), output \(\sem{\FOCCount{\tx}{\psi(\tx, \tw)}}^\A\).
\end{lemma}
\begin{proof}
  Let \(k \deff \abs{\tx}\) and \(\ell \deff \abs{\ty}\).
  If \(k = 0\), then \(\sem{\FOCCount{\tx}{\psi(\tx, \tw)}}^\A = \sem{\psi(\tw)}^\A\).
  Thus, in this case, the statement of \cref{lem:fo-dist-counting}
  is a direct consequence of \cref{thm:fo-testing}.
  Hence, in the following, let \(k > 0\).

  Let \(r' \deff (2r+1) \cdot (k + \ell)\),
  and let \(\C'\) be the class of all subgraphs of \(2r'\)-neighbourhoods
  of vertices \(v \in V(G')\) in graphs \(G' \in \C\).
  Then \(\C'\) has bounded expansion by \cref{def:nowhere-dense}.

  Let \(f_{\ref*{thm:sparse-neighbourhood-cover}}\) be the function from
  \cref{thm:sparse-neighbourhood-cover} for \(\C\),
  let \(f_{\ref*{thm:fo-counting-testing-bounded-expansion}}\) be the function from
  \cref{thm:fo-counting-testing-bounded-expansion} for \(\C'\),
  and let \(f_{\ref*{thm:fo-testing}}\) be the function from
  \cref{thm:fo-testing} for \(\C\).
  We let \(\sigma \deff \sigma(\phi)\).

  Using the algorithm from \cref{thm:sparse-neighbourhood-cover},
  we compute a distance-\(r'\) neighbourhood cover \(\K\) of \(G_\A\)
  of radius at most \(2r'\)
  and overlap at most
  \(f_{\ref*{thm:sparse-neighbourhood-cover}}(r', \epsilon) \cdot \abs{A}^\epsilon\)
  in time
  \(f_{\ref*{thm:sparse-neighbourhood-cover}}(r', \epsilon) \cdot \abs{A}^{1+\epsilon}\).

  We set \(\sigma' \deff \sigma \uplus \set{Q}\) for a fresh unary relation symbol \(Q\),
  and, for every cluster \(C \in \K\),
  we let \(\A_C\) be the \(\sigma'\)-expansion of \(\A[C]\)
  with \(Q(\A_C) \deff \setc{v \in C}{\K(v) = C}\).
  Note that \(G_{\A_C} \in \C'\).
  Hence, the graphs \(G_{\A_C}\) come from a class of bounded expansion,
  so we can run the algorithm from \cref{thm:fo-counting-testing-bounded-expansion}
  to test \(\FO\)-definable properties within a cluster as follows.

  Let \(\Phi[\sigma', k+\ell, \qr(\psi)]\) be the computable set of \(\FO[\sigma']\) formulas
  in the syntactic normal form from \cref{rem:finite-fo} of quantifier rank at most \(\qr(\psi)\)
  with free variables among \(x_1, \dots, x_k, y_1, \dots, y_\ell\).
For every \(C \in \K\) and every \(\FO[\sigma']\) formula
  \(\psi' \in \Phi[\sigma', k+\ell, \qr(\psi)]\),
  we run the algorithm of \cref{thm:fo-counting-testing-bounded-expansion}
  on \(\A_C\) and \(\psi'\).
  For every \(C \in \K\) and formula \(\psi'\),
  the preprocessing runs in time
  \(f_{\ref*{thm:fo-counting-testing-bounded-expansion}}(\psi', \sigma') \cdot \abs{C}\).
  Since \(\sum_{C \in \K} \abs{C} \leq \abs{A} \cdot \overlap(\K) \leq
  f_{\ref*{thm:sparse-neighbourhood-cover}}(r', \epsilon) \cdot \abs{A}^{1+\epsilon}\),
  and since \(\Phi[\sigma', k+\ell, \qr(\psi)]\) is computable from \(\psi\),
  there is a function \(f'\) such that the preprocessing over all clusters and formulas
  in total takes time at most \(f'(\psi) \cdot \abs{A}^{1+\epsilon}\).
  Moreover, after this preprocessing phase,
  the algorithm can answer the queries as described in
  \cref{thm:fo-counting-testing-bounded-expansion} for any cluster
  and any formula as described above in time \(f'(\psi)\).

  In addition, for every \(\FO[\sigma]\) formula \(\psi' \in \Phi[\sigma', k+\ell, \qr(\psi)]\),
  we run the algorithm of \cref{thm:fo-testing}
  on \(\A\) and \(\psi'\).
  In total, the preprocessing takes time at most \(f''(\psi, \epsilon) \cdot \abs{A}^{1+\epsilon}\)
  for a computable function \(f''\) that we can define based on \(f_{\ref*{thm:fo-testing}}\).
  After this preprocessing phase,
  we can check whether a tuple satisfies a formula as described above
  in time \(f''(\psi, \epsilon)\).

  In the remainder of this proof,
  we proceed by induction on the number of connected components of \(G\).
  First, we consider the case where \(G\) is connected.
  If \(\ell = 0\), then our goal is to compute \(\sem{\FOCCount{\tx}{\psi(\tx)}}^\A\).
  We then just store this value and return it in constant time when we are given the empty tuple.
  Let \(\tv \in A^k\).
  If \(\A \models \psi[\tv]\),
  then \(\A \models \delta^\sigma_{G, 2r+1}[\tv]\).
  Since \(G\) is connected, this implies that
  \(v_1, \dots, v_k \in \neighbA{(2r+1) \cdot (k-1)}{v_1}\).
  Furthermore, since \(\delta^\sigma_{G, 2r+1}\) and \(\phi\) are \((2r+1)\)-local,
  and \(\neighbA{2r+1}{\tv} \subseteq \neighbA{r'}{v_1} \subseteq \K(v_1)\),
  we have
  \(\A \models \psi[\tv]
  \iff \NeighbA{2r+1}{\tv} \models \psi[\tv]
  \iff \A_{\K(v_1)} \models \psi[\tv]\).
  In addition, for a cluster \(C \in \K\),
  it holds that \(v_1 \in Q(\A_C)\) if and only if \(C = \K(v_1)\).

  This shows that both of the values
  \(\sem{\FOCCount{\tx}{\psi(\tx)}}^\A
  = \sum_{C \in \K} \sem{\FOCCount{\tx}{\bigl(Q(x_1) \land \psi(\tx)\bigr)}}^{\A_C}\),
  and \(\sem{\FOCCount{\tx}{\bigl(Q(x_1) \land \psi(\tx)\bigr)}}^{\A_C}\)
  can be computed in time \(f'(\psi)\) because of the preprocessing above.
  Hence, for the case that \(G\) is connected and \(\ell = 0\),
  we can set \(f(\psi, \epsilon)\) such that the total running time is at most
  \(f(\psi, \epsilon) \cdot \abs{A}^{1+\epsilon}\).

  If \(\ell > 0\), then, analogously to the case \(\ell = 0\) above,
  for \(\tv \in A^k\) and \(\tw \in A^\ell\),
  we have that \(\A \models \psi[\tv, \tw]\)
  implies \(v_1, \dots, v_k \in \neighbA{(2r+1) \cdot (k + \ell - 1)}{w_1}\)
  and \(\neighbA{2r+1}{\tv \tw} \subseteq \neighbA{r'}{w_1} \subseteq \K(w_1)\).
  Thus, it holds that \(\A \models \psi[\tv, \tw]
  \iff \NeighbA{2r+1}{\tv \tw} \models \psi[\tv, \tw]
  \iff \A_{\K(w_1)} \models \psi[\tv, \tw]\).
  This shows that
  \(\sem{\FOCCount{\tx}{\psi(\tx, \tw)}}^\A
  = \sem{\FOCCount{\tx}{\psi(\tx, \tw)}}^{\A_{\K(w_1)}}\),
  which can be computed in time \(f'(\psi)\) because of the preprocessing above.
  Hence, also for the case \(\ell > 0\),
  we can set \(f(\psi, \epsilon)\) such that the preprocessing takes time at most
  \(f(\psi, \epsilon) \cdot \abs{A}^{1+\epsilon}\)
  and queries can be answered in time \(f(\psi, \epsilon)\).

  It remains to consider the case where \(G\) is not connected.
  For this, suppose that \(G\) has \(p > 1\) connected components
  and that the statement of \cref{lem:fo-dist-counting} holds for all graphs \(G'\)
  with at most \(p-1\) connected components.
  Let \(V_1 \subseteq V(G)\) be the vertices of the connected component of \(G\)
  that contains \(k+1\), corresponding to the variable \(y_1\),
  and let \(V_2 \deff V(G) \setminus V_1\).
  Then \(G_1 \deff G[V_1]\) is connected, and \(G_2 \deff G[V_2]\) has \(p-1\) connected components.

  For \(i \in [2]\),
  let \(\tx_i\) be the tuple obtained from \(\tx\) by keeping only those entries
  with indices \(j\) where \(j \in V_i\),
  and let \(\ty_i\) be the tuple obtained from \(\ty\) by keeping only those entries
  with indices \(j\) where \(k + j \in V_i\).
  For tuples \(\tv \in A^k\) and \(\tw \in A^\ell\),
  we analogously define \(\tv_1\), \(\tv_2\), \(\tw_1\), and \(\tw_2\).
  Let \(k_i \deff \abs{\tx_i}\) and \(\ell_i \deff \abs{\ty_i}\) for \(i \in [2]\).

  We use the induction hypothesis on \(G_1\) and \(G_2\)
  to run the preprocessing for all formulas of the form
  \(\delta^\sigma_{G_i, 2r+1}(\tx_i, \ty_i) \land \phi'(\tx_i, \ty_i)\) for \(i \in [2]\),
  where \(\phi' \in \FO[\sigma]\) has quantifier rank at most
  \(q \deff \qr(\phi) + \max_{R \in \sigma} \ar(R) \cdot \lceil\log(r+1)\rceil\).
  In addition, we also apply the induction hypothesis on every graph \(G' \in \G_{k+\ell}\)
  with at most \(p-1\) connected components for formulas of the form
  \(\delta^\sigma_{G', 2r+1}(\tx, \ty) \land \phi'(\tx, \ty)\)
  with \(\phi' \in \FO[\sigma]\) having quantifier rank at most \(q\).
  In total, the preprocessing takes time
  \(f'''(\psi, \epsilon) \cdot \abs{A}^{1+\epsilon}\)
  for some function \(f'''\) that we can define based on the values of \(f\)
  for graphs \(G\) with at most \(p-1\) connected components.
  After the preprocessing phase, given a tuple \(\tw_1 \in A^{\ell_1}\),
  we can compute
  \(\sem{\delta^\sigma_{G_1, 2r+1}(\tx_1, \tw_1) \land \phi'(\tx_1, \tw_1)}^\A\)
  for any \(\phi'\) as described above in time \(f'''(\psi, \epsilon)\),
  and this holds analogously for \(G_2\) and \(\tw_2 \in A^{\ell_2}\)
  and for every \(G' \in \G_{k+\ell}\) with at most \(p-1\) connected components
  and for every \(\tw \in A^\ell\).

  Let \(\Delta\) be a Feferman--Vaught decomposition of \(\phi\) with respect to \(V_1\) and \(V_2\)
  computed via \cref{thm:FV}.
  We let \(\Delta'\) be the set containing the pair \((\alpha'_1, \alpha'_2)\)
  for every \((\alpha_1, \alpha_2) \in \Delta\),
  where \(\alpha'_i\) is obtained from \(\alpha_i\)
  by recursively replacing subformulas of the form \(\exists x \phi\) by
  \(\exists x \bigl(\phi \land \Lor_{y \in \free(\alpha_i)} \dist^\sigma_{\leq r}(x,y)\bigr)\).
  Then \(\Delta'\) is computable from \(\psi\), \(V_1\), and \(V_2\),
  and every formula occurring in \(\Delta'\) is \(r\)-local and has quantifier rank at most
  \(q = \qr(\phi) + \max_{R \in \sigma} \ar(R) \cdot \lceil\log(r+1)\rceil\).

  Thus, for \(\tv \in A^k\) and \(\tw \in A^\ell\) with
  \(\NrA{\tv \tw} = \NrA{\tv_1 \tw_1} \uplus \NrA{\tv_2 \tw_2}\),
  we have \(\NrA{\tv \tw} \models \phi[\tv, \tw]\)
  if and only if there is some \((\alpha_1, \alpha_2) \in \Delta\)
  such that \(\NrA{\tv_1 \tw_1} \models \alpha_1[\tv_1, \tw_1]\)
  and \(\NrA{\tv_2 \tw_2} \models \alpha_2[\tv_2, \tw_2]\),
  and this holds if and only if there is some \((\alpha'_1, \alpha'_2) \in \Delta'\)
  such that \(\NrA{\tv_1 \tw_1} \models \alpha'_1[\tv_1, \tw_1]\)
  and \(\NrA{\tv_2 \tw_2} \models \alpha'_2[\tv_2, \tw_2]\).
  Since \(\phi\) and the formulas in \(\Delta'\) are \(r\)-local,
  this implies that \(\A \models \phi[\tv, \tw]\)
  if and only if there is some \((\alpha'_1, \alpha'_2) \in \Delta'\)
  such that \(\A \models \alpha'_1[\tv_1, \tw_1]\)
  and \(\A \models \alpha'_2[\tv_2, \tw_2]\).
  Furthermore, \(\A \models \delta^\sigma_{G, 2r+1}[\tv, \tw]\)
  already implies that \(\NrA{\tv \tw} = \NrA{\tv_1 \tw_1} \uplus \NrA{\tv_2 \tw_2}\).

  This shows that, for all \(\tv \in A^k\) and \(\tw \in A^\ell\), we have
  \(\A \models \psi[\tv, \tw]\)
  if and only if
  \(\NrA{\tv \tw} = \NrA{\tv_1 \tw_1} \uplus \NrA{\tv_2 \tw_2}\)
  and there is some \((\alpha'_1, \alpha'_2) \in \Delta'\)
  such that \(\A \models \delta^\sigma_{G_1, 2r+1}[\tv_1, \tw_1] \land  \alpha'_1[\tv_1, \tw_1]\)
  and \(\A \models \delta^\sigma_{G_2, 2r+1}[\tv_2, \tw_2] \land  \alpha'_2[\tv_2, \tw_2]\).

  Let \(\G' \subset \G_{k+\ell}\) be the class of graphs \(G'\)
  with \(G'[V_1] = G_1\), \(G'[V_2] = G_2\), and \(G' \neq G\).
  For every such graph \(G'\), there has to be an edge between \(V_1\) and \(V_2\),
  so \(G'\) has at most \(p-1\) connected components.
  Moreover, for every \(\tv \in A^k\) and \(\tw \in A^\ell\)
  with \(\NrA{\tv \tw} \neq \NrA{\tv_1 \tw_1} \cup \NrA{\tv_2 \tw_2}\)
  (that is, there is an additional edge in \(\NrA{\tv \tw}\)
  between a vertex from \(\nrA{\tv_1 \tw_1}\) and a vertex from \(\nrA{\tv_2 \tw_2}\)),
  \(\NrA{\tv_1 \tw_1} \models \delta^\sigma_{G_1, 2r+1}[\tv_1, \tw_1]\),
  and \(\NrA{\tv_2 \tw_2} \models \delta^\sigma_{G_2, 2r+1}[\tv_2, \tw_2]\),
  it holds that \(\A \models \delta^\sigma_{G', r}[\tv, \tw]\)
  for some \(G' \in \G'\).

  By definition, we have
  \[\sem{\FOCCount{\tx}{\psi(\tx, \tw)}}^\A
  = \bigabs{\bigsetc{\tv \in A^k}{\A \models \delta^\sigma_{G, 2r+1}[\tv, \tw]
  \land \phi[\tv, \tw]}}.\]
  Since the formulas in \(\Delta\) are mutually exclusive, this is equal to
  \[\sum_{(\alpha'_1, \alpha'_2) \in \Delta'}
  \bigabs{\bigsetc{\tv \in A^k}{\A \models \delta^\sigma_{G, 2r+1}[\tv, \tw]
  \land \alpha'_1[\tv_1, \tw_1] \land \alpha'_2[\tv_2, \tw_2]}}.\]
  Finally, by the definition of \(\G'\) above, the term is equal to
  \begin{align*}
    &\sum_{(\alpha'_1, \alpha'_2) \in \Delta'}
       \Bigl(
       \bigabs{\bigsetc{\tv_1 \in A^{k_1}}{\A \models \delta^\sigma_{G_1, 2r+1}[\tv_1, \tw_1]
       \land \alpha'_1[\tv_1, \tw_1]}}\\
    &\mkern60mu\cdot
       \bigabs{\bigsetc{\tv_2 \in A^{k_2}}{\A \models \delta^\sigma_{G_2, 2r+1}[\tv_2, \tw_2]
       \land \alpha'_2[\tv_2, \tw_2]}}\\
    &-\mkern-5mu \sum_{G' \in \G'}\mkern-5mu
       \bigabs{\bigsetc{\tv \in A^k}{\A \models \delta^\sigma_{G', 2r+1}[\tv, \tw]
       \land \alpha'_1[\tv_1, \tw_1] \land \alpha'_2[\tv_2, \tw_2]}}\Bigr).
  \end{align*}
  Since the size of \(\Delta'\) and \(\G'\) only depends on \(\psi\),
  it suffices to bound the time to compute the value of each of the summands.
We have
  \begin{align*}
     &\ \bigabs{\bigsetc{\tv_1 \in A^{k_1}}{\A \models \delta^\sigma_{G_1, 2r+1}[\tv_1, \tw_1]
    \land \alpha'_1[\tv_1, \tw_1]}}\\
    =&\ \sem{\FOCCount{\tx_1}{\bigl(\delta^\sigma_{G_1, 2r+1}(\tx_1, \tw_1)
    \land \alpha'_1(\tx_1, \tw_1)\bigr)}}^\A,
  \end{align*}
  and this can be computed in time \(f'''(\psi, \epsilon)\)
  by the preprocessing based on the induction hypothesis above.
  Analogously, it holds that
  \begin{align*}
     &\ \bigabs{\bigsetc{\tv_2 \in A^{k_2}}{\A \models \delta^\sigma_{G_2, 2r+1}[\tv_2, \tw_2]
    \land \alpha'_2[\tv_2, \tw_2]}}\\
    =&\ \sem{\FOCCount{\tx_2}{\bigl(\delta^\sigma_{G_2, 2r+1}(\tx_2, \tw_2)
    \land \alpha'_2(\tx_2, \tw_2)\bigr)}}^\A
  \end{align*}
  can be computed in time \(f'''(\psi, \epsilon)\).
  Lastly, each of the terms
  \begin{align*}
    &\ \bigabs{\bigsetc{\tv \in A^k}{\A \models \delta^\sigma_{G', 2r+1}[\tv, \tw]
       \land \alpha'_1[\tv_1, \tw_1] \land \alpha'_2[\tv_2, \tw_2]}}\\
    =&\ \sem{\FOCCount{\tx}{\bigl(\delta^\sigma_{G', 2r+1}(\tx, \tw)
    \land \alpha'_1(\tx_1, \tw_1) \land \alpha'_2(\tx_2, \tw_2)\bigr)}}^\A
  \end{align*}
  with \(G' \in \G'\) can be computed in time \(f'''(\psi, \epsilon)\)
  by the preprocessing based on the induction hypothesis above,
  since \(G'\) has at most \(p-1\) connected components.

  All in all, we can define a function \(f\) based on
  \(f_{\ref*{thm:sparse-neighbourhood-cover}}\),
  \(f'\), \(f''\), and \(f'''\) such that,
  after preprocessing in time \(f(\psi, \epsilon) \cdot \abs{A}^{1+\epsilon}\),
  we can answer the following queries in time \(f(\psi, \epsilon)\):
  given a tuple \(\tw \in A^\ell\), output \(\sem{\FOCCount{\tx}{\psi(\tx, \tw)}}^\A\).
  This is the statement of \cref{lem:fo-dist-counting}.
\end{proof}

Based on this special case,
we can now prove \cref{thm:fo-counting-testing-locally-bounded-expansion}.

\begin{proof}[Proof of \cref{thm:fo-counting-testing-locally-bounded-expansion}]
  If \(\abs{\tx} = 0\), then, for all \(\tw \in A^{\abs{\ty}}\), it holds that
  \(\sem{\FOCCount{\tx}{\phi(\tx, \tw)}}^\A = \sem{\phi(\tw)}^\A\)
  Hence, in this case, \cref{thm:fo-counting-testing-locally-bounded-expansion}
  is a special case of \cref{thm:fo-testing}.
  In the following, we assume \(\abs{\tx} > 0\).

  Let \(f_{\ref*{thm:fo-testing}}\) be the function from \cref{thm:fo-testing} for \(\C\).
  and let \(f_{\ref*{lem:fo-dist-counting}}\)
  be the function from \cref{lem:fo-dist-counting} for \(\C\).
Let \(\phi(\tx, \ty)\) be an \(\FO\) formula,
  let \(\A\) be a \(\sigma(\phi)\)-structure \(G_\A \in \C\),
  and let \(\epsilon \in \Qpos\).

  First, we use \cref{thm:gaifman} to transform \(\phi\) into an equivalent formula \(\phi'\)
  in Gaifman normal form with radius at most \(r \deff 7^{\qr(\phi)}\).
  For every basic local sentence \(\chi\) occurring in \(\phi'\),
  we use \cref{thm:fo-testing} to check in time
  \(f_{\ref*{thm:fo-testing}}(\chi, \sigma(\phi), \epsilon) \cdot \abs{A}^{1+\epsilon}\)
  whether \(\A \models \chi\) holds.
  Let \(\phi'_\A(\tx, \ty)\) be the formula obtained from \(\phi'\) by replacing every
  basic local sentence \(\chi\) by \(\top()\) if \(\A \models \chi\),
  and replacing \(\chi\) by \(\bot()\) if \(\A \not\models \chi\).
  Then \(\phi'_\A(\tx, \ty)\) is an \(r\)-local formula,
  and it is equivalent to \(\Lor_{G \in \G_{\abs{\tx} + \abs{\ty}}} \psi_G(\tx, \ty)\)
  for \(\psi_G(\tx, \ty) \deff \delta^{\sigma(\phi)}_{G, 2r+1}(\tx, \ty) \land \phi'_\A(\tx, \ty)\).
  Moreover, it holds that \(\sem{\FOCCount{\tx}{\phi(\tx, \tw)}}^\A
  = \sum_{G \in \G_{\abs{\tx} + \abs{\ty}}} \sem{\FOCCount{\tx}{\psi_G(\tx, \tw)}}^\A\)
  for all \(\tw \in A^{\abs{\ty}}\).

  We run the preprocessing from \cref{lem:fo-dist-counting} for \(\epsilon\)
  and for each of the formulas \(\psi_G\) on \(\A\).
  This runs in time
  \(f_{\ref*{lem:fo-dist-counting}}(\psi_G, \epsilon) \cdot \abs{A}^{1+\epsilon}\).
After this preprocessing, given a tuple \(\tw \in A^{\abs{\ty}}\),
  we can compute \(\sem{\FOCCount{\tx}{\psi_G(\tx, \tw)}}^\A\)
  in time \(f_{\ref*{lem:fo-dist-counting}}(\psi_G, \epsilon)\).
  Thus, we can compute \(\sem{\FOCCount{\tx}{\phi(\tx, \tw)}}^\A\)
  in time \(\sum_{G \in \G_{\abs{\tx} + \abs{\ty}}}
  f_{\ref*{lem:fo-dist-counting}}(\psi_G, \epsilon)\).

  Let \(\Psi\) be the set of formulas obtained from \(\phi'\)
  by replacing the basic local sentences in \(\phi'\) by \(\top()\) or \(\bot()\),
  and let \(\Psi'\) be the set of all formulas of the form
  \(\delta^{\sigma(\phi)}_{G, 2r+1}(\tx, \ty) \land \psi(\tx, \ty)\)
  for \(G \in \G_{\abs{\tx} + \abs{\ty}}\) and \(\psi \in \Psi\).
  Then \(\psi_G \in \Psi'\) for every \(G \in \G_{\abs{\tx} + \abs{\ty}}\),
  so \(f_{\ref*{lem:fo-dist-counting}}(\psi_G, \epsilon)
  \leq \max_{\psi' \in \Psi'} f_{\ref*{lem:fo-dist-counting}}(\psi', \epsilon)\).

  All in all, this shows that there is a function \(f\) such that,
  given the formula \(\phi\), the structure \(\A\), and \(\epsilon\),
  the preprocessing described above runs in time \(f(\phi, \epsilon) \cdot \abs{A}^{1+\epsilon}\)
  and, after the preprocessing, given a tuple \(\tw \in A^{\abs{\ty}}\),
  we can compute \(\sem{\FOCCount{\tx}{\phi(\tx, \tw)}}^\A\)
  in time \(f(\phi, \epsilon)\).
\end{proof}

For the proof of \cref{lem:cgfoc-preprocessing-bounded-expansion},
in addition to \cref{thm:fo-counting-testing-locally-bounded-expansion},
we use the following result to compute the list of all cliques of a given size.

\begin{theorem}[Enumeration, {\cite[Lemma~2.2, Corollary~2.5]{SchweikardtSegoufinVigny_Enumeration2022}}]
  \label{thm:fo-enumeration}
  Let \(\C\) be a nowhere dense graph class.
  There is a function \(f\) and an algorithm that does the following.
  Given an \(\FO\) formula \(\phi(\tx)\),
  a \(\sigma\)-structure \(\A\) for some \(\sigma \supseteq \sigma(\phi)\)
  with \(G_\A \in \C\),
  and a rational \(\epsilon \in \Qpos\),
  after preprocessing in time \(f(\phi, \sigma, \epsilon) \cdot \abs{A}^{1+\epsilon}\),
  the algorithm enumerates all tuples \(\tv \in A^{\abs{\tx}}\) such that \(\A \models \phi[\tv]\)
  with \(f(\phi, \sigma, \epsilon)\) delay in lexicographic order, without duplicates.
  If \(\C\) is effectively nowhere dense, then \(f\) is computable.
\end{theorem}

Again, \cite[Corollary~5]{SegoufinVigny_LocallyBoundedExpansion2017}
states a similar result for classes of locally bounded expansion,
and we use \cite[Lemma~2.2, Corollary~2.5]{SchweikardtSegoufinVigny_Enumeration2022} instead,
because it gives a uniform algorithm for all \(\epsilon \in \Qpos\).

We can now prove \cref{lem:cgfoc-preprocessing-bounded-expansion}.

\cgfocPreprocessingBoundedExpansion*

\begin{proof}
  Let \(\C\) be a graph class of locally bounded expansion,
  let \(\epsilon \in \Qpos\),
  let \(\phi(\tx)\) be a \(\cgFOC\) formula,
  and let \(\sigma_\phi\), \(\phi'\), and \(\A_\phi\) be as in the proof
  of \cref{lem:cgfoc-preprocessing}.
We prove the result by structural induction on \(\phi\)
  and thereby describe a recursive algorithm for computing \(\A_\phi\).
  Let \(f_{\ref*{lem:nowhere-dense-cliques}}\),
  \(f_{\ref*{thm:fo-counting-testing-locally-bounded-expansion}}\),
  and \(f_{\ref*{thm:fo-enumeration}}\)
  be the functions for \(\C\) from
  \cref{lem:nowhere-dense-cliques},
  \cref{thm:fo-counting-testing-locally-bounded-expansion},
  and \cref{thm:fo-enumeration}, respectively.

  If \(\phi\) is not the result of applying rule~\eqref{def:cgfoc-P} of \cref{def:cgfoc},
  then the computation of \(\A_\phi\) as described in \cref{lem:cgfoc-preprocessing}
  trivially runs in time linear in \(\norm{\A_\phi}\),
  based on the recursively computed structure(s) for the respective subformula(s).
  Furthermore, by \cref{lem:nowhere-dense-cliques},
  \(\norm{\A_\phi} \leq \bigabs{\sigma_\phi} \cdot f_{\ref*{lem:nowhere-dense-cliques}}(k, \epsilon)
  \cdot \abs{A}^{1+\epsilon}\),
  where \(k\) is the maximum arity of a relation symbol in \(\sigma_\phi\),
  and \(\sigma_\phi\) and \(k\) can be computed from \(\phi\) and \(\sigma\).

  It remains to consider the case where
  \(\phi\) is the result of applying rule~\eqref{def:cgfoc-P},
  that is, \(\phi\) is of the form
  \(\Land_{i=1}^k R_i(z_{i,1}, \dots, z_{i, \ar(R_i)}) \land \Pred(t_1, \dots, t_m)\).
  Let \(\tx'\) be a tuple of pairwise distinct variables such that
  \(\tilde{x}' = \bigcup_{i=1}^m \free(t_i)\),
  and let \(\ell \deff \abs{\tx'}\).
  For the computation of \(\A_\phi\), we only need to compute the relation \(R_\phi(\A_\phi)\)
  from the proof of \cref{lem:cgfoc-preprocessing}.
Recall that \(R_\phi(\A_\phi)\)
  is the set of all tuples \(\tw \in A^\ell\) that form a clique in \(G_\A\) and where
  \((\sem{t_1}^{(\A, \beta_{\tw})}, \dots, \sem{t_m}^{(\A, \beta_{\tw})}) \in \sem{\Pred}\)
  for \(\beta_{\tw}(x'_i) \deff w_i\) for all \(i \in [\ell]\).

  There is an \(\FO[\sigma]\) formula \(\phi_{\textup{clique}}(\tx')\)
  that only depends on \(\sigma\) and \(\tx'\)
  such that, for all \(\tw \in A^{\abs{\tx'}}\),
  it holds that \(\A \models \phi_{\textup{clique}}[\tw]\)
  if and only if \(\tw\) forms a clique in \(G_\A\).
  Moreover, by \cref{lem:nowhere-dense-cliques}, the number of such tuples \(\tw\)
  is bounded by \(f_{\ref*{lem:nowhere-dense-cliques}}(\ell, \epsilon) \cdot \abs{A}^{1+\epsilon}\).
  Hence, by applying the enumeration result \cref{thm:fo-enumeration}
  to \(\A\) and \(\phi_{\textup{clique}}\),
  we can compute a list of all such tuples \(\tw\) in time
  \begin{align*}
    &f_{\ref*{thm:fo-enumeration}}(\phi_{\textup{clique}}, \sigma, \epsilon) \cdot \abs{A}^{1+\epsilon}
    + f_{\ref*{thm:fo-enumeration}}(\phi_{\textup{clique}}, \sigma, \epsilon) \cdot
    f_{\ref*{lem:nowhere-dense-cliques}}(\ell, \epsilon) \cdot \abs{A}^{1+\epsilon}\\
    =
    \ &f_{\ref*{thm:fo-enumeration}}(\phi_{\textup{clique}}, \sigma, \epsilon)
    \cdot \bigl(1+f_{\ref*{lem:nowhere-dense-cliques}}(\ell, \epsilon)\bigr) \cdot \abs{A}^{1+\epsilon}.
  \end{align*}

  It remains to check for every such tuple \(\tw\) whether we have
  \((\sem{t_1}^{(\A, \beta_{\tw})}, \dots, \sem{t_m}^{(\A, \beta_{\tw})}) \in \sem{\Pred}\).
  For this, note that the counting terms \(t_1, \dots, t_m\) are built using
  rules~\eqref{def:foc-countingterm}--\eqref{def:foc-plustimesterm},
  that is, they are built using integer constants, addition, multiplication,
  and \#-terms of the form \(\FOCCount{\ty}{\psi}\).
  We recursively compute \(\sigma_\psi\), \(\psi'\), and \(\A_\psi\).
  By the induction hypothesis, these can be computed in time
  \(f(\psi, \sigma, \epsilon) \cdot \abs{A}^{1+\epsilon}\).
  Then \(\sem{\FOCCount{\ty}{\psi}}^{(\A, \beta_{\tw})}
  = \sem{\FOCCount{\ty}{\psi'}}^{(\A_\psi, \beta_{\tw})}\),
  and \(\psi' \in \FO[\sigma_\psi]\).
  Hence, we can run the algorithm from \cref{thm:fo-counting-testing-locally-bounded-expansion}
  on \(\A_\psi\) and \(\psi'\).
  After preprocessing in time
  \(f_{\ref*{thm:fo-counting-testing-locally-bounded-expansion}}(\psi', \epsilon) \cdot \abs{A}^{1+\epsilon}\),
  we can compute \(\sem{\FOCCount{\ty}{\psi'}}^{(\A_\psi, \beta_{\tw})}\)
  for a single tuple \(\tw\) in time
  \(f_{\ref*{thm:fo-counting-testing-locally-bounded-expansion}}(\psi', \epsilon)\),
  so we can compute the results after the preprocessing in total time
  \(f_{\ref*{thm:fo-counting-testing-locally-bounded-expansion}}(\psi', \epsilon) \cdot
  f_{\ref*{lem:nowhere-dense-cliques}}(\ell, \epsilon) \cdot \abs{A}^{1+\epsilon}\)
  for all tuples \(\tw \in A^\ell\) that form a clique in \(G_\A\).

  Doing this for every \#-term occurring in \(t_1, \dots, t_m\) allows us to compute the values
  \(\sem{t_1}^{(\A, \beta_{\tw})}, \dots, \sem{t_m}^{(\A, \beta_{\tw})}\)
  for all tuples \(\tw \in A^\ell\) that form a clique in \(G_\A\)
  in time \(f'(\phi, \sigma, \epsilon) \cdot \abs{A}^{1+\epsilon}\)
  for a function \(f'\) that only depends on \(\C\).
  Once we have computed \(\sem{t_1}^{(\A, \beta_{\tw})}, \dots, \sem{t_m}^{(\A, \beta_{\tw})}\),
  we can check whether \(\tw \in R_\phi(\A_\phi)\)
  with an oracle call to \(\sem{\Pred}\) in time \(\bigO(1)\).

  Thus, all in all we can define a function \(f\)
  based on \(f_{\ref*{lem:nowhere-dense-cliques}}\),
  \(f_{\ref*{thm:fo-counting-testing-locally-bounded-expansion}}\),
  and \(f_{\ref*{thm:fo-enumeration}}\)
  such that we can compute \(\A_\phi\)
  in time \(f(\phi, \sigma, \epsilon) \cdot \abs{A}^{1+\epsilon}\).
  This is the statement of \cref{lem:cgfoc-preprocessing-bounded-expansion}.
\end{proof}

If the function in \cref{thm:fo-counting-testing-bounded-expansion}
is computable for every class of effectively bounded expansion,
then the functions in
\cref{lem:fo-dist-counting},
\cref{thm:fo-counting-testing-locally-bounded-expansion},
and
\cref{lem:cgfoc-preprocessing-bounded-expansion}
are computable for every class of effectively locally bounded expansion.

\subsection{Proof of Theorem~\ref{thm:answering-enumeration}}
\label{subsec:appendix-evaluation}

\answeringAndEnumeration*

\begin{proof}
  First, we prove \cref{thm:answering-enumeration}\ref{item:answering}.
  Let
  \(f_{\ref*{lem:nowhere-dense-cliques}}\),
  \(f_{\ref*{lem:cgfoc-preprocessing-bounded-expansion}}\),
  \(f_{\ref*{thm:fo-testing}}\),
  and \(f_{\ref*{thm:fo-counting-testing-locally-bounded-expansion}}\)
  be the functions from
  \cref{lem:nowhere-dense-cliques,lem:cgfoc-preprocessing-bounded-expansion}
  and \cref{thm:fo-testing,thm:fo-counting-testing-locally-bounded-expansion}.

  If \(\sigma \neq \sigma(\xi)\),
  then we can let \(\A'\) be the \(\sigma(\xi)\)-reduct of \(\A'\)
  and continue with \(\A'\) instead of \(\A\).
  By \cref{lem:nowhere-dense-cliques},
  \(\norm{\A} \leq \abs{\sigma} \cdot f_{\ref*{lem:nowhere-dense-cliques}}(k, \epsilon)
  \cdot \abs{A}^{1+\epsilon}\),
  where \(k\) is the maximum arity of a relation symbol in \(\sigma\).
  Since we can compute \(\A'\) from \(\A\) in time linear in \(\norm{\A}\),
  in the following, we may assume that \(\sigma = \sigma(\xi)\).

  If \(\xi(\tx) = \phi(\tx)\) for a \(\cgFOC\) formula \(\phi\),
  then we apply \cref{lem:cgfoc-preprocessing-bounded-expansion}
  to \(\A\) and \(\phi\), and, in time
  \(f_{\ref*{lem:cgfoc-preprocessing-bounded-expansion}}(\phi, \sigma, \epsilon)
  \cdot \abs{A}^{1+\epsilon}\),
  we obtain a signature \(\sigma_\phi \supseteq \sigma\),
  an \(\FO[\sigma_\phi]\) formula \(\phi'\),
  and a \(\sigma_\phi\)-expansion \(\A_\phi\) of \(\A\) with \(G_\A = G_{\A_\phi}\)
  such that for all \(\tv \in A^{\abs{\tx}}\),
  we have \(\A \models \phi[\tv]\) if and only if \(\A_\phi \models \phi'[\tv]\).
Then, since \(G_{\A_\phi} \in \C\),
  we can apply \cref{thm:fo-testing} to \(\A_\phi\) and \(\phi'\).
  Thus, after preprocessing in time
  \(f_{\ref*{thm:fo-testing}}(\phi', \sigma_\phi, \epsilon) \cdot \abs{A}^{1+\epsilon}\),
  given a tuple \(\tv \in A^{\abs{\tx}}\),
  we can output \(\sem{\phi'(\tv)}^{\A_\sigma} = \sem{\phi(\tv)}^\A\)
  in time \(f_{\ref*{thm:fo-testing}}(\phi', \sigma_\phi, \epsilon)\).
  Since \(\sigma_\phi\) and \(\phi'\) only depend on \(\phi\) and \(\sigma\),
  we can set
  \(f_a(\phi, \sigma, \epsilon) \deff
  f_{\ref*{lem:cgfoc-preprocessing-bounded-expansion}}(\phi, \sigma, \epsilon)
  + f_{\ref*{thm:fo-testing}}(\phi', \sigma_\phi, \epsilon)\),
  which then satisfies the requirements of \cref{thm:answering-enumeration}\ref{item:answering}.

  If \(\xi(\tx)\) is a \(\cgFOC\) counting term,
  then \(\xi\) is built using rules~\eqref{def:foc-countingterm}--\eqref{def:foc-plustimesterm}
  of \cref{def:cgfoc},
  that is, it is built using integer constants, addition, multiplication,
  and \#-terms of the form \(\FOCCount{\ty}{\psi(\tx', \ty)}\).
  For every such \#-term \(\FOCCount{\ty}{\psi(\tx', \ty)}\),
  we apply \cref{lem:cgfoc-preprocessing-bounded-expansion}
  to \(\A\) and \(\psi\), and, in time
  \(f_{\ref*{lem:cgfoc-preprocessing-bounded-expansion}}(\psi, \sigma, \epsilon)
  \cdot \abs{A}^{1+\epsilon}\),
  we obtain a signature \(\sigma_\psi \supseteq \sigma\),
  an \(\FO[\sigma_\psi]\) formula \(\psi'\),
  and a \(\sigma_\psi\)-expansion \(\A_\psi\) of \(\A\) with \(G_\A = G_{\A_\psi}\)
  such that for all \(\tv' \in A^{\abs{\tx'}}\),
  we have \(\A \models \psi[\tv']\) if and only if \(\A_\psi \models \psi'[\tv']\).
  Hence, it holds that
  \(\sem{\FOCCount{\ty}{\psi(\tv', \ty)}}^\A = \sem{\FOCCount{\ty}{\psi'(\tv', \ty)}}^{\A_\psi}\).
  Since \(\psi'\) is a first-order formula and \(G_{\A_\psi} \in \C\),
  we can apply \cref{thm:fo-counting-testing-locally-bounded-expansion}
  to \(\A_\psi\) and \(\psi'\).
  Thus, after preprocessing in time
  \(f_{\ref*{thm:fo-counting-testing-locally-bounded-expansion}}(\psi', \epsilon)
  \cdot \abs{A}^{1+\epsilon}\),
  given a tuple \(\tv' \in A^{\abs{\tx'}}\),
  we can compute
  \(\sem{\FOCCount{\ty}{\psi(\tv', \ty)}}^\A = \sem{\FOCCount{\ty}{\psi'(\tv', \ty)}}^{\A_\psi}\)
  in time \(f_{\ref*{thm:fo-counting-testing-locally-bounded-expansion}}(\psi', \epsilon)\).
  Hence, whenever we are given a tuple \(\tv \in A^{\abs{\tx}}\),
  we can efficiently compute the results of all \#-terms that \(\xi\) consists of,
  and then compute \(\sem{\xi(\tv)}^\A\) with additions and multiplications.
  Thus, all in all, we can set \(f_a(\xi, \sigma, \epsilon)\) such that the preprocessing
  takes time at most \(f_a(\xi, \sigma, \epsilon) \cdot \abs{A}^{1+\epsilon}\)
  and, given \(\tv \in A^{\abs{\tx}}\), the computation of \(\sem{\xi(\tv)}^\A\)
  takes time at most \(f_a(\xi, \sigma, \epsilon)\).

  Now, we prove \cref{thm:answering-enumeration}\ref{item:enumeration}.
  Let \(f_{\ref*{lem:cgfoc-preprocessing-bounded-expansion}}\)
  and \(f_{\ref*{thm:fo-enumeration}}\)
  be the functions from
  \cref{lem:cgfoc-preprocessing-bounded-expansion}
  and \cref{thm:fo-enumeration}.
  As above, we may assume that \(\sigma = \sigma(\xi)\).

  We apply \cref{lem:cgfoc-preprocessing-bounded-expansion}
  to \(\A\) and \(\xi\), and, in time
  \(f_{\ref*{lem:cgfoc-preprocessing-bounded-expansion}}(\xi, \sigma, \epsilon)
  \cdot \abs{A}^{1+\epsilon}\),
  we obtain a signature \(\sigma_\xi \supseteq \sigma\),
  an \(\FO[\sigma_\xi]\) formula \(\xi'\),
  and a \(\sigma_\xi\)-expansion \(\A_\xi\) of \(\A\) with \(G_\A = G_{\A_\xi}\)
  such that for all \(\tv \in A^{\abs{\tx}}\),
  we have \(\A \models \xi[\tv]\) if and only if \(\A_\phi \models \xi'[\tv]\).
Then, since \(G_{\A_\xi} \in \C\),
  we can apply \cref{thm:fo-enumeration} to \(\A_\xi\) and \(\xi'\).
  Thus, after preprocessing in time
  \(f_{\ref*{thm:fo-enumeration}}(\xi', \sigma_\xi, \epsilon) \cdot \abs{A}^{1+\epsilon}\),
  we can enumerate all tuples \(\tv \in A^{\abs{\tx}}\) such that \(\A \models \xi[\tv]\)
  with \(f_{\ref*{thm:fo-enumeration}}(\xi', \sigma_\xi, \epsilon)\) delay in lexicographic order,
  without duplicates.

  Since \(\sigma_\xi\) and \(\xi'\) only depend on \(\xi\) and \(\sigma\),
  we can set
  \(f_b(\xi, \sigma, \epsilon) \deff
  f_{\ref*{lem:cgfoc-preprocessing-bounded-expansion}}(\xi, \sigma, \epsilon)
  + f_{\ref*{thm:fo-enumeration}}(\xi', \sigma_\xi, \epsilon)\),
  which then satisfies the requirements of \cref{thm:answering-enumeration}\ref{item:enumeration}.

  Finally, it suffices to set \(f(\xi, \sigma, \epsilon)
  \deff \max\set{f_a(\xi, \sigma, \epsilon), f_b(\xi, \sigma, \epsilon)}\).
\end{proof}

\subsection{Details on Remark~\ref{rem:weaker-guards}}
\label{subsec:appendix-weaker-guards}

Using constructs of the form \(E(x,y) \land E(y,z) \land \Pred_=(t_1(x), t_2(z))\),
we give a transduction \((\phi_V(x), \phi_E(x_1,x_2))\)
that produces the class of all graphs from the class of coloured trees of height at most \(2\).

For this, we consider the signature \(\sigma \deff \set{E, \firstvertexrelation, \secondvertexrelation}\)
of coloured graphs, where \(E\) is a binary relation symbol,
and \firstvertexrelation{} and \secondvertexrelation{} are unary relation symbols.

Let \(n \in \Npos\), and let \(G\) be a graph with \(V(G) = [n]\).
In the following, we define a coloured tree \(T\)
with colours \firstvertexrelation{} and \secondvertexrelation{}
and with \(V(G) \subseteq V(T)\) such that,
for all \(v \in V(T)\), we have \(T \models \phi_V[v]\) if and only if \(v \in V(G)\).
Moreover, for all \(v_1, v_2 \in V(G)\),
we will have \(T \models \phi_E[v_1, v_2]\) if and only if \(\set{v_1, v_2} \in E(G)\).
The formulas \(\phi_V\) and \(\phi_E\) will not depend on \(G\).

For all \(i \in [n]\), let \(V_i \deff \setc{v_{i,j}}{j \in [i]}\),
and let \(V_V \deff \bigcup_{i \in [n]} V_i\).
Moreover, for all \(e = \set{w_1, w_2} \in E(G)\) with \(w_1 < w_2\),
let \(V_{e,1} \deff \setc{v_{e,1,j}}{j \in [w_1]}\),
\(V_{e,2} \deff \setc{v_{e,2,j}}{j \in [w_2]}\),
and let \(V_{E, i} \deff \bigcup_{e \in E(G)} V_{e, i}\) for \(i \in [2]\).

We set \(V(T) \deff \set{0} \cup V(G) \cup E(G) \cup V_V \cup V_{E,1} \cup V_{E,2}\) and
\begin{align*}
  E(T) \deff\ &\bigsetc{\set{0,v}}{v \in V(G) \cup E(G)}\\
               &\cup \bigsetc{\set{i,v}}{i \in [n], v \in V_i}\\
               &\cup \bigsetc{\set{e, v}}{e \in E(G), v \in V_{e,1} \cup V_{e,2}}.
\end{align*}
Furthermore, we let \(\firstvertexrelation(T) \deff V_V \cup V_{E,1}\)
and \(\secondvertexrelation(T) \deff V_{E,2}\).

Intuitively, the vertex \(0\) is the root of the tree,
and the vertices from \(V(G) \cup E(G)\) are its children.
In addition, every vertex \(i \in V(G)\) has \(i\) \firstvertexrelation-neighbours,
and every vertex \(\set{i,j} \in E(G) \subseteq V(T)\)
with \(i < j\)
has \(i\) \firstvertexrelation-neighbours
and \(j\) \secondvertexrelation-neighbours.
See \cref{fig:hardness-example-tree} for an example.

We set \(\phi_V(x) \deff \exists y \bigl(E(x,y) \land \firstvertexrelation(y)\bigr)
\land \neg \exists y \bigl(E(x,y) \land \secondvertexrelation(y)\bigr)\).
As desired, for all \(v \in V(T)\), we have \(T \models \phi_V(v)\) if and only if \(v \in V(G)\),
since the first condition is only satisfied for \(v \in V(G) \cup E(G)\),
and the second condition is not satisfied for \(v \in E(G)\).
Furthermore, we set
\[\phi_E(x_1,x_2) \deff \exists y \Bigl(\bigl(\phi_{E,1}(x_1,y) \land \phi_{E,2}(x_2,y)\bigr)
\lor \bigl(\phi_{E,1}(x_2,y) \land \phi_{E,2}(x_1,y)\bigr)\Bigr)\]
for
\begin{align*}
  t_{\smallfirstvertexrelation}(x)
  &\deff \FOCCount{(z)}{\bigl(E(x,z) \land \firstvertexrelation(z)\bigr)},\\
  t_{\smallsecondvertexrelation}(x)
  &\deff \FOCCount{(z)}{\bigl(E(x,z) \land \secondvertexrelation(z)\bigr)},\\
  \phi_{E,1}(x_1,y)
  &\deff \exists z_0 \Bigl(E(x_1,z_0) \land E(z_0,y)
  \land \Pred_{=}\bigl(t_{\smallfirstvertexrelation}(x_1),
  t_{\smallfirstvertexrelation}(y)\bigr)\Bigr),\\
  \intertext{and}
  \phi_{E,2}(x_2,y)
  &\deff \exists z_0 \Bigl(E(x_2,z_0) \land E(z_0,y)
  \land \Pred_{=}\bigl(t_{\smallfirstvertexrelation}(x_2),
  t_{\smallsecondvertexrelation}(y)\bigr)\Bigr).
\end{align*}
For \(e = \set{v_1, v_2} \in E(G)\), we have \(e \in V(T)\).
If \(v_1 < v_2\),
then it holds that \(T \models \phi_{E,1}[v_1, e] \land \phi_{E,2}[v_2, e]\)
by setting the variable \(z_0\) to the vertex \(0 \in V(T)\).
If \(v_2 < v_1\),
then we have \(T \models \phi_{E,1}[v_2, e] \land \phi_{E,2}[v_1, e]\).
This shows that \(T \models \phi_E[v_1, v_2]\)
and \(T \models \phi_E[v_2, v_1]\).
It can be easily verified that \(T \not\models \phi_E[v_1, v_2]\)
for all \(v_1, v_2 \in E(G)\) with \(\set{v_1, v_2} \not\in E(G)\).

This shows that \(G\) is the result of applying the transduction \((\phi_V, \phi_E)\) on \(T\).
Moreover, \(T\) can be computed from \(G\) in polynomial time.

\subsection{Details on Remark~\ref{rem:dense-classes}}
\label{subsec:appendix-dense-classes}

Using only constructs of the form \(E(x,y)\land \Pred_=(t_1(x), t_2(y))\),
we give a transduction\linebreak[4] \((\phi_V(x), \phi_E(x_1,x_2))\)
that produces the class of all graphs from a class of coloured graphs of shrub-depth at most \(2\).
The construction will be similar to the one given in \cref{subsec:appendix-weaker-guards}.

Let \(n \in \Npos\), and let \(G\) be a graph with \(V(G) = [n]\).
As above, for all \(i \in [n]\), let \(V_i \deff \setc{v_{i,j}}{j \in [i]}\),
and let \(V_V \deff \bigcup_{i \in [n]} V_i\).
Moreover, for all \(e = \set{w_1, w_2} \in E(G)\) with \(w_1 < w_2\),
let \(V_{e,1} \deff \setc{v_{e,1,j}}{j \in [w_1]}\),
\(V_{e,2} \deff \setc{v_{e,2,j}}{j \in [w_2]}\),
and let \(V_{E, i} \deff \bigcup_{e \in E(G)} V_{e, i}\) for \(i \in [2]\).

We set \(V(G') \deff V(G) \cup E(G) \cup V_V \cup V_{E,1} \cup V_{E,2}\) and
\begin{align*}
  E(G')
  &\deff \binom{V(G) \cup E(G)}{2}\\
  &\cup \bigsetc{\set{i,v}}{i \in [n], v \in V_i}\\
  &\cup \bigsetc{\set{e, v}}{e \in E(G), v \in V_{e,1} \cup V_{e,2}}.
\end{align*}
Furthermore, we let \(\firstvertexrelation(G') \deff V_V \cup V_{E,1}\)
and \(\secondvertexrelation(G') \deff V_{E,2}\).

Intuitively, the vertices in \(V(G) \cup E(G)\) form a clique.
In addition, every vertex \(i \in V(G)\) has \(i\) \firstvertexrelation-neighbours,
and every vertex \(\set{i,j} \in E(G) \subseteq V(G')\) with \(i < j\)
has \(i\) \firstvertexrelation-neighbours
and \(j\) \secondvertexrelation-neighbours.
See \cref{fig:hardness-example-dense} for an example.
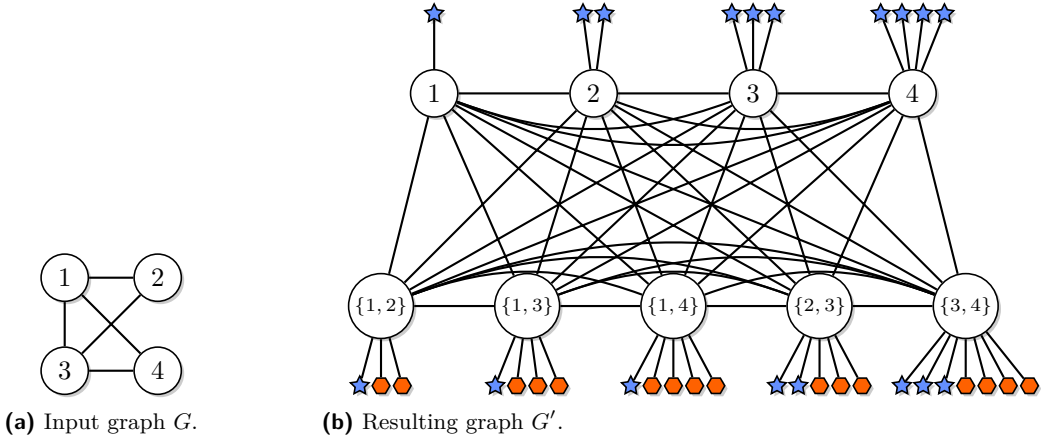
\begin{figure}
  \begin{subfigure}{.2\textwidth}
    \centering
    \begin{tikzpicture}
      \node[vertex]                                (1) {\(1\)};
\node[vertex, right of=1, node distance=3.5em] (2) {\(2\)};
\node[vertex, below of=1, node distance=3.5em] (3) {\(3\)};
\node[vertex, right of=3, node distance=3.5em] (4) {\(4\)};
\draw[edge] (1) -- (2) -- (3) -- (1) -- (4) -- (3);
     \end{tikzpicture}\caption{Input graph \(G\).}
  \end{subfigure}
  \hfill
  \begin{subfigure}{.7\textwidth}
    \centering
    \begin{tikzpicture}
      \coordinate (topcenter);
\node[vertex, left of=topcenter, node distance=9em]  (1) {\(1\)};
\node[vertex, left of=topcenter, node distance=3em]  (2) {\(2\)};
\node[vertex, right of=topcenter, node distance=3em] (3) {\(3\)};
\node[vertex, right of=topcenter, node distance=9em] (4) {\(4\)};

\node[edge vertex, below of=topcenter, node distance=8em] (e14) {\(\set{1,4}\)};
\node[edge vertex, left of=e14, node distance=11em]        (e12) {\(\set{1,2}\)};
\node[edge vertex, left of=e14, node distance=5.5em]         (e13) {\(\set{1,3}\)};
\node[edge vertex, right of=e14, node distance=5.5em]        (e23) {\(\set{2,3}\)};
\node[edge vertex, right of=e14, node distance=11em]       (e34) {\(\set{3,4}\)};

\node[first number vertex] at ($(1) + ( 0.0em, 3em)$) (1v1) {};
\node[first number vertex] at ($(2) + (-0.4em, 3em)$) (2v1) {};
\node[first number vertex] at ($(2) + ( 0.4em, 3em)$) (2v2) {};
\node[first number vertex] at ($(3) + (-0.8em, 3em)$) (3v1) {};
\node[first number vertex] at ($(3) + ( 0.0em, 3em)$) (3v2) {};
\node[first number vertex] at ($(3) + ( 0.8em, 3em)$) (3v3) {};
\node[first number vertex] at ($(4) + (-1.2em, 3em)$) (4v1) {};
\node[first number vertex] at ($(4) + (-0.4em, 3em)$) (4v2) {};
\node[first number vertex] at ($(4) + ( 0.4em, 3em)$) (4v3) {};
\node[first number vertex] at ($(4) + ( 1.2em, 3em)$) (4v4) {};
\draw[edge] (1) -- (1v1);
\draw[edge] (2) -- (2v1);
\draw[edge] (2) -- (2v2);
\draw[edge] (3) -- (3v1);
\draw[edge] (3) -- (3v2);
\draw[edge] (3) -- (3v3);
\draw[edge] (4) -- (4v1);
\draw[edge] (4) -- (4v2);
\draw[edge] (4) -- (4v3);
\draw[edge] (4) -- (4v4);

\node[first number vertex]  at ($(e12) + (-0.8em, -3em)$) (e12v11) {};
\node[second number vertex] at ($(e12) + ( 0.0em, -3em)$) (e12v21) {};
\node[second number vertex] at ($(e12) + ( 0.8em, -3em)$) (e12v22) {};
\node[first number vertex]  at ($(e13) + (-1.2em, -3em)$) (e13v11) {};
\node[second number vertex] at ($(e13) + (-0.4em, -3em)$) (e13v21) {};
\node[second number vertex] at ($(e13) + ( 0.4em, -3em)$) (e13v22) {};
\node[second number vertex] at ($(e13) + ( 1.2em, -3em)$) (e13v23) {};
\node[first number vertex]  at ($(e14) + (-1.6em, -3em)$) (e14v11) {};
\node[second number vertex] at ($(e14) + (-0.8em, -3em)$) (e14v21) {};
\node[second number vertex] at ($(e14) + ( 0.0em, -3em)$) (e14v22) {};
\node[second number vertex] at ($(e14) + ( 0.8em, -3em)$) (e14v23) {};
\node[second number vertex] at ($(e14) + ( 1.6em, -3em)$) (e14v24) {};
\node[first number vertex]  at ($(e23) + (-1.6em, -3em)$) (e23v11) {};
\node[first number vertex]  at ($(e23) + (-0.8em, -3em)$) (e23v12) {};
\node[second number vertex] at ($(e23) + ( 0.0em, -3em)$) (e23v21) {};
\node[second number vertex] at ($(e23) + ( 0.8em, -3em)$) (e23v22) {};
\node[second number vertex] at ($(e23) + ( 1.6em, -3em)$) (e23v23) {};
\node[first number vertex]  at ($(e34) + (-2.4em, -3em)$) (e34v11) {};
\node[first number vertex]  at ($(e34) + (-1.6em, -3em)$) (e34v12) {};
\node[first number vertex]  at ($(e34) + (-0.8em, -3em)$) (e34v13) {};
\node[second number vertex] at ($(e34) + ( 0.0em, -3em)$) (e34v21) {};
\node[second number vertex] at ($(e34) + ( 0.8em, -3em)$) (e34v22) {};
\node[second number vertex] at ($(e34) + ( 1.6em, -3em)$) (e34v23) {};
\node[second number vertex] at ($(e34) + ( 2.4em, -3em)$) (e34v24) {};
\draw[edge] (e12) -- (e12v11);
\draw[edge] (e12) -- (e12v21);
\draw[edge] (e12) -- (e12v22);
\draw[edge] (e13) -- (e13v11);
\draw[edge] (e13) -- (e13v21);
\draw[edge] (e13) -- (e13v22);
\draw[edge] (e13) -- (e13v23);
\draw[edge] (e14) -- (e14v11);
\draw[edge] (e14) -- (e14v21);
\draw[edge] (e14) -- (e14v22);
\draw[edge] (e14) -- (e14v23);
\draw[edge] (e14) -- (e14v24);
\draw[edge] (e23) -- (e23v11);
\draw[edge] (e23) -- (e23v12);
\draw[edge] (e23) -- (e23v21);
\draw[edge] (e23) -- (e23v22);
\draw[edge] (e23) -- (e23v23);
\draw[edge] (e34) -- (e34v11);
\draw[edge] (e34) -- (e34v12);
\draw[edge] (e34) -- (e34v13);
\draw[edge] (e34) -- (e34v21);
\draw[edge] (e34) -- (e34v22);
\draw[edge] (e34) -- (e34v23);
\draw[edge] (e34) -- (e34v24);

\draw[edge] (1) -- (2) -- (3) -- (4);
\draw[edge] (1) edge [bend right=20] (3);
\draw[edge] (2) edge [bend right=20] (4);
\draw[edge] (1) edge [bend right=20] (4);

\draw[edge] (e12) -- (e13) -- (e14) -- (e23) -- (e34);
\draw[edge] (e12) edge [bend left=20] (e14);
\draw[edge] (e13) edge [bend left=20] (e23);
\draw[edge] (e14) edge [bend left=20] (e34);
\draw[edge] (e12) edge [bend left=20] (e23);
\draw[edge] (e13) edge [bend left=20] (e34);
\draw[edge] (e12) edge [bend left=20] (e34);

\draw[edge] (1) -- (e12) (2) -- (e12) (3) -- (e12) (4) -- (e12)
            (1) -- (e13) (2) -- (e13) (3) -- (e13) (4) -- (e13)
            (1) -- (e14) (2) -- (e14) (3) -- (e14) (4) -- (e14)
            (1) -- (e23) (2) -- (e23) (3) -- (e23) (4) -- (e23)
            (1) -- (e34) (2) -- (e34) (3) -- (e34) (4) -- (e34);
     \end{tikzpicture}\caption{Resulting graph \(G'\).}
  \end{subfigure}
  \caption{Construction of a coloured graph \(G'\) based on a graph \(G\)
    such that the transduction \((\phi_V, \phi_E)\) from \cref{subsec:appendix-dense-classes},
  applied on \(G'\), produces \(G\).}
  \label{fig:hardness-example-dense}
\end{figure}

We set \(\phi_V(x) \deff \exists y \bigl(E(x,y) \land \firstvertexrelation(y)\bigr)
\land \neg \exists y \bigl(E(x,y) \land \secondvertexrelation(y)\bigr)\).
As desired, for all \(v \in V(G')\), we have \(G' \models \phi_V(v)\) if and only if \(v \in V(G)\),
since the first condition is only satisfied for \(v \in V(G) \cup E(G)\),
and the second condition is not satisfied for \(v \in E(G)\).
Furthermore, we set
\[\phi_E(x_1,x_2) \deff \exists y \Bigl(\bigl(\phi_{E,1}(x_1,y) \land \phi_{E,2}(x_2,y)\bigr)
\lor \bigl(\phi_{E,1}(x_2,y) \land \phi_{E,2}(x_1,y)\bigr)\Bigr)\]
for
\begin{align*}
  t_{\smallfirstvertexrelation}(x)
  &\deff \FOCCount{(z)}{\bigl(E(x,z) \land \firstvertexrelation(z)\bigr)},\\
  t_{\smallsecondvertexrelation}(x)
  &\deff \FOCCount{(z)}{\bigl(E(x,z) \land \secondvertexrelation(z)\bigr)},\\
  \phi_{E,1}(x_1,y)
  &\deff E(x_1,y)
  \land \Pred_{=}\bigl(t_{\smallfirstvertexrelation}(x_1),
  t_{\smallfirstvertexrelation}(y)\bigr),\\
  \intertext{and}
  \phi_{E,2}(x_2,y)
  &\deff E(x_2,y)
  \land \Pred_{=}\bigl(t_{\smallfirstvertexrelation}(x_2),
  t_{\smallsecondvertexrelation}(y)\bigr).
\end{align*}

For \(e = \set{v_1, v_2} \in E(G)\), we have \(e \in V(G')\).
If \(v_1 < v_2\),
then it holds that \(G' \models \phi_{E,1}[v_1, e] \land \phi_{E,2}[v_2, e]\).
If \(v_2 < v_1\),
then we have \(G' \models \phi_{E,1}[v_2, e] \land \phi_{E,2}[v_1, e]\).
This shows that \(G' \models \phi_E[v_1, v_2]\)
and \(G' \models \phi_E[v_2, v_1]\).
It can be easily verified that \(G' \not\models \phi_E[v_1, v_2]\)
for all \(v_1, v_2 \in E(G)\) with \(\set{v_1, v_2} \not\in E(G)\).

This shows that \(G\) is the result of applying the transduction \((\phi_V, \phi_E)\) on \(G'\).
Moreover, \(G'\) can be computed from \(G\) in polynomial time.

Finally, we show that the class of \(\set{E}\)-reducts of the resulting coloured graphs \(G'\)
has shrub-depth at most \(2\).
As defined in \cite{GanianHNOMR_Shrubdepth},
a class of graphs has shrub-depth at most \(d\)
if there is a number \(m\) of labels such that for every graph \(G'\) in the class,
there is a tree \(T\) with the following properties.
\begin{enumerate}
  \item The set of leaves of \(T\) is exactly \(V(G')\).
  \item The length of each root-to-leaf path in \(T\) is exactly \(d\).
  \item Each leaf of \(T\) is assigned one of \(m\) labels.
  \item The existence of an edge \(\set{v_1, v_2} \in E(G')\)
    depends solely on the labels of \(v_1\) and \(v_2\)
    and the distance between \(v_1\) and \(v_2\) in \(T\).
\end{enumerate}
We show that the \(\set{E}\)-reduct of every coloured graph \(G'\) resulting from our construction above
can be defined via such a tree with \(2\) labels.
For this, we let the leaves of the tree \(T\) be
\(V(G') = V(G) \cup E(G) \cup V_V \cup V_{E,1} \cup V_{E,2}\).
The leaves in \(V(G) \cup E(G)\) are given the label \(1\),
and the leaves in \(V_V \cup V_{E,1} \cup V_{E,2}\) are given the label \(2\).
In addition to the leaves,
we let \(V(T)\) contain a fresh root node \(r\)
as well as a fresh node \(i'\) for every \(i \in V(G)\)
and a fresh node \(e'\) for every \(e \in E(G)\).
We let the fresh nodes \(i'\) and \(e'\) be the children of \(r\).
For a node \(i'\) with \(i \in V(G)\),
we let the children of \(i'\) be the vertices in \(V_i\) as well as the vertex \(i\).
For a node \(e'\) with \(e \in V(G)\),
we let the children of \(e'\) be the vertices in \(V_{e,1} \cup V_{e,2}\)
as well as the vertex \(e\).
The edges of \(G'\) are exactly the pairs of vertices
of distance \(2\) in \(T\) with labels \(1\) and \(2\)
(these are the connections between the vertices in \(V(G) \cup E(G)\)
with their \firstvertexrelation{}- and \secondvertexrelation{}-neighbours)
and the pairs of vertices of distance \(4\) in \(T\) with labels \(1\) and \(1\)
(this is the clique on \(V(G) \cup E(G)\)).
 \section{Proofs Omitted in Section~\ref{sec:graph-dimension}}
\label{sec:appendix-graph-dimension}

\subsection{Proof of Lemma~\ref{lem:cgfoc-rank}}
\label{subsec:appendix-cgfoc-rank}

For proving \cref{lem:cgfoc-rank},
we show the following result for \(\FO\) instead of \(\cgFOC\).

\begin{lemma}
  \label{lem:fo-rank}
  There is a computable function \(T \colon \FO \times \N \to \N\) such that,
  for every \(\FO\) formula \(\phi(\tx, \ty, \tz)\),
  every \(p \in \N\),
  every \(\sigma(\phi)\)-structure \(\A\),
  for \(r \deff 7^{\qr(\phi)}\),
  and for all \(V, W \subseteq A\) that are \(\bigl((2r+1)(\abs{\tz} + 1)\bigr)\)-separated
  in \(G_\A\) by a set of size at most \(p\),
  we have \(\rank\bigl(\MTypeTAVW\bigr) \leq T(\phi, p)\)
  for \(t(\tx, \ty) \deff \FOCCount{\tz}{\phi(\tx, \ty, \tz)}\).
\end{lemma}

Before we prove \cref{lem:fo-rank},
we first show how to obtain \cref{lem:cgfoc-rank} from it.

\cgfocRank*

\begin{proof}
  Let \(t(\tx, \ty)\) be a \(\cgFOC\) counting term,
  let \(\A\) be a \(\sigma\)-structure for some \(\sigma \supseteq \sigma(t)\),
  and let \(V, W \subseteq A\) be sets that are \(r(t)\)-separated in \(G_\A\)
  by a set of size at most \(m\),
  where \(r \colon \cgFOC \to \N\) will be specified below.

  If \(\sigma \neq \sigma(t)\),
  then we let \(\A'\) be the \(\sigma(t)\)-reduct of \(\A\).
  It holds that \(\MTypeTAVW = \MTypeTVW{\A'}\),
  and \(V\) and \(W\) are also \(r(t)\)-separated in \(G_{\A'}\)
  by a set of size at most \(m\).
  Thus, it suffices to provide a computable upper bound \(T(t, m)\)
  for \(\rank\bigl(\MTypeTVW{\A'}\bigr)\),
  since the same value is also an upper bound for \(\rank\bigl(\MTypeTAVW\bigr)\).
  Hence, in the following, we may assume that \(\sigma = \sigma(t)\).

  We prove the result by structural induction on \(t(\tx, \ty)\).
  As a \(\cgFOC\) counting term, \(t\) is the result of applying
  rule~\eqref{def:foc-countingterm},
  \eqref{def:foc-constterm},
  or~\eqref{def:foc-plustimesterm}
  of \cref{def:cgfoc}.

  If \(t\) is the result of applying rule~\eqref{def:foc-countingterm},
  then it is a \#-term of the form
  \(t(\tx, \ty) = \FOCCount{\tz}{\psi(\tx, \ty, \tz)}\) for a \(\cgFOC\) formula \(\psi\).
  We apply \cref{lem:cgfoc-preprocessing} to \(\A\) and \(\psi\),
  and we obtain a signature \(\sigma_\psi \supseteq \sigma\),
  an \(\FO[\sigma_\psi]\) formula \(\psi'\),
  and a \(\sigma_\psi\)-expansion \(\A_\psi\) of \(\A\) with \(G_\A = G_{\A_\psi}\)
  such that for all \(\tv \in A^{\abs{\tx}}\), \(\tw \in A^{\abs{\ty}}\),
  and \(\tc \in A^{\abs{\tz}}\),
  we have \(\A \models \psi[\tv, \tw, \tc]\) if and only if \(\A_\psi \models \psi'[\tv, \tw, \tc]\).
  Hence, for \(t'(\tx, \ty) \deff \FOCCount{\tz}{\psi'(\tx, \ty, \tz)}\),
  it holds that \(\sem{t(\tv, \tw)}^\A = \sem{t'(\tv, \tw)}^{\A_\psi}\),
  so \(\MTypeTAVW = \MType{t'}{\A_\psi}{V/W}\).
  Let \(T_\FO \colon \FO \times \N \to \N\) be the computable function from \cref{lem:fo-rank}.
  We set \(r(t) \deff (2 \cdot 7^{\qr(\psi')}+1)(\abs{\tz}+1)\)
  and \(T(t, m) \deff T_\FO(\psi', m)\).
  Note that, since \(\psi'\) is computable from \(t\),
  \(r(t)\) is also computable from \(t\),
  and \(T(t, m)\) is computable from \(t\) and \(m\).
  By \cref{lem:fo-rank}, since \(\psi' \in \FO\),
  and since \(V\) and \(W\) are \(r(t)\)-separated in \(G_\A\) by a set of size at most \(m\)
  if and only if they are \(r(t)\)-separated in \(G_{\A_\phi}\) by a set of size at most \(m\),
  we have
  \(\rank\bigl(\MTypeTAVW\bigr) = \rank\bigl(\MType{t'}{\A_\psi}{V/W}\bigr)
  \leq T_\FO(\psi', m) = T(t, m)\).

  If \(t\) is the result of applying rule~\eqref{def:foc-constterm},
  then \(t = i\) for some \(i \in \Z\).
  Hence, \(\bigl(\MTypeTAVW\bigr)_{\tv, \tw} = i\)
  for all \(\tv \in V^{\abs{\tx}}\) and \(\tw \in W^{\abs{\ty}}\).
  Thus, we have \(\rank\bigl(\MTypeTAVW\bigr) \leq 1\),
  and we do not rely on \(V\) and \(W\) being separated,
  so we can set \(r(t) \deff 0\) and \(T(t,m) \deff 1\).

  Lastly, if \(t\) is the result of applying rule~\eqref{def:foc-plustimesterm},
  then \(t = t_1 + t_2\) or \(t = t_1 \cdot t_2\) for \(\cgFOC\) counting terms \(t_1\) and \(t_2\).
  By the induction hypothesis, for \(i \in [2]\),
  \(r(t_i)\) is computable from \(t_i\) and
  \(T(t_i, m)\) is computable from \(t_i\) and \(m\)
  such that \(\rank\bigl(\MTypeAVW{t_i}\bigr) \leq T(t_i, m)\)
  for all sets \(V, W \subseteq A\) that are \(r(t_i)\)-separated by a set of size at most \(m\).
  We set \(r(t) \deff \max\set{r(t_1), r(t_2)}\).
If \(t = t_1 + t_2\), then \(\MTypeTAVW = \MTypeAVW{t_1} + \MTypeAVW{t_2}\), so
  \begin{align*}
    \rank\bigl(\MTypeTAVW\bigr)
    &\leq \rank\bigl(\MTypeAVW{t_1}\bigr) + \rank\bigl(\MTypeAVW{t_2}\bigr)\\
    &\leq T(t_1, m) + T(t_2, m) \ffed T(t, m).
  \end{align*}
If \(t = t_1 \cdot t_2\), then \(\MTypeTAVW = \MTypeAVW{t_1} \odot \MTypeAVW{t_2}\),
  where \(\odot\) denotes the element-wise multiplication of two matrices,
  for which it holds that \(\rank(M_1 \odot M_2) \leq \rank(M_1) \cdot \rank(M_2)\)
  for matrices \(M_1, M_2\) (see, \eg, \cite[Theorem~5.1.7]{HadamardProduct}).
  Thus,
  \begin{align*}
    \rank\bigl(\MTypeTAVW\bigr)
    &\leq \rank\bigl(\MTypeAVW{t_1}\bigr) \cdot \rank\bigl(\MTypeAVW{t_2}\bigr)\\
    &\leq T(t_1, m) \cdot T(t_2, m) \ffed T(t, m).
  \end{align*}

  All in all, \(T\) and \(r\) as defined above are suitable computable functions,
  and this concludes the proof.
\end{proof}

Towards \cref{lem:fo-rank}, we first prove the following result,
where the formula \(\delta^\sigma_{G, r}\) is defined as in
\cref{subsec:appendix-cgfoc-preprocessing-bounded-expansion}.

\begin{lemma}
  \label{lem:fo-dist-rank}
  There is a computable function \(T \colon \FO \to \N\) such that, for every \(\FO\) formula
  \(\psi(\tx, \ty, \tz) = \delta^{\sigma(\phi)}_{G, 2r+1}(\tx, \ty, \tz) \land \phi(\tx, \ty, \tz)\)
  for some \(r \in \N\),
  \(G \in \G_{\abs{\tx} + \abs{\ty} + \abs{\tz}}\),
  and an \(r\)-local \(\FO\) formula \(\phi\),
  for every \(\sigma(\phi)\)-structure \(\A\),
  and all \(V, W \subseteq A\) that are \(\bigl((2r+1)(\abs{\tz} + 1)\bigr)\)-separated in \(G_\A\),
  we have \(\rank\bigl(\MTypeTAVW\bigr) \leq T(\psi)\)
  for \(t(\tx, \ty) \deff \FOCCount{\tz}{\psi(\tx, \ty, \tz)}\).
\end{lemma}
\begin{proof}
  Let \(k \deff \abs{\tx}\), \(\ell \deff \abs{\ty}\), and \(m \deff \abs{\tz}\).
  If \(k = 0\), then \(\MTypeTAVW\) only has a single row,
  so \(\rank\bigl(\MTypeTAVW\bigr) \leq 1 \ffed T(\psi)\)
  for all \(V, W \subseteq A\).
  Analogously, if \(\ell = 0\), then \(\MTypeTAVW\) only has a single column,
  and we again have \(\rank\bigl(\MTypeTAVW\bigr) \leq 1 \ffed T(\psi)\)
  for all \(V, W \subseteq A\).
  Thus, in the following, we assume \(k > 0\) and \(\ell > 0\).

  We prove the result by induction on the number of connected components of \(G\).
  For the base case, let \(G\) be a connected graph.
  Then it holds that \(\A \not\models \delta^{\sigma(\phi)}_{G, 2r+1}[\tv, \tw, \tc]\)
  for all \(\tv \in V^k\), \(\tw \in W^\ell\), and \(\tc \in A^m\),
  since \(V\) and \(W\) are \(\bigl((2r+1)(m+1)\bigr)\)-separated,
  which implies that \(\dist^A(\tv, \tw) > (2r+1)(m+1)\).
  Thus, we have \(\sem{t(\tv, \tw)}^\A = 0\) for all \(\tv \in V^k\) and \(\tw \in W^\ell\),
  so \(\rank\bigl(\MTypeTAVW\bigr) = 0 \ffed T(\psi)\).

  Now suppose that \(G\) has \(p > 1\) connected components,
  and suppose the statement of \cref{lem:fo-dist-rank} holds for all graphs \(G\)
  with at most \(p-1\) connected components.

  If there is a connected component of \(G\) that has a non-empty intersection
  with both \([k]\) and \([k+1, \ell]\),
  then, analogously to the case that \(G\) is connected,
  it holds that \(\sem{t(\tv, \tw)}^\A = 0\) for all \(\tv \in V^k\) and \(\tw \in W^\ell\),
  and \(\rank\bigl(\MTypeTAVW\bigr) = 0 \ffed T(\psi)\).

  Otherwise, let \(V_1 \subseteq V(G)\) be the vertices of all connected components of \(G\)
  that contain a vertex from \([k]\),
  and let \(V_2 \deff V(G) \setminus V_1\).
  Then \([k] \subseteq V_1 \subseteq [k] \cup [k+\ell+1, k+\ell+m]\),
  \([k+1, \ell] \subseteq V_2 \subseteq [k+1, k+\ell+m]\),
  and \(G_1 \deff G[V_1]\) and \(G_2 \deff G[V_2]\) have at most \(p-1\) connected components.
  The following argumentation is similar to the one in the proof of \cref{lem:fo-dist-counting}.

  For \(i \in [2]\), we let \(\tz_i\) be the tuple obtained from \(\tz\)
  by keeping only those entries at index \(j\) with \(k + \ell + j \in V_i\),
  we analogously define \(\tc_i\) for a tuple \(\tc \in A^m\),
  and we set \(m_i \deff \abs{\tz_i}\),

  Let \(\Delta\) be a Feferman--Vaught decomposition of \(\phi\) with respect to \(V_1\) and \(V_2\)
  computed via \cref{thm:FV}.
  We let \(\Delta'\) be the set containing the pair \((\alpha'_1, \alpha'_2)\)
  for every \((\alpha_1, \alpha_2) \in \Delta\),
  where \(\alpha'_i\) is obtained from \(\alpha_i\)
  by recursively replacing subformulas of the form \(\exists x \phi\) by
  \(\exists x \bigl(\phi \land \Lor_{y \in \free(\alpha_i)} \dist^\sigma_{\leq r}(x,y)\bigr)\).
  Then \(\Delta'\) is computable from \(\psi\), \(V_1\), and \(V_2\),
  and every formula occurring in \(\Delta'\) is \(r\)-local.

  Thus, for \(\tv \in A^k\), \(\tw \in A^\ell\), and \(\tc \in A^m\) with
  \(\NrA{\tv \tw \tc} = \NrA{\tv \tc_1} \uplus \NrA{\tw \tc_2}\),
  it holds that \(\NrA{\tv \tw \tc} \models \phi[\tv, \tw, \tc]\)
  if and only if there is some \((\alpha_1, \alpha_2) \in \Delta\)
  with \(\NrA{\tv \tc_1} \models \alpha_1[\tv, \tc_1]\)
  and \(\NrA{\tw \tc_2} \models \alpha_2[\tw, \tc_2]\),
  and this holds if and only if there is some \((\alpha'_1, \alpha'_2) \in \Delta'\)
  such that \(\NrA{\tv \tc_1} \models \alpha'_1[\tv, \tc_1]\)
  and \(\NrA{\tw \tc_2} \models \alpha'_2[\tw, \tc_2]\).
  Since \(\phi\) and the formulas in \(\Delta'\) are \(r\)-local,
  this implies that \(\A \models \phi[\tv, \tw, \tc]\)
  if and only if there is some \((\alpha'_1, \alpha'_2) \in \Delta'\)
  such that \(\A \models \alpha'_1[\tv, \tc_1]\)
  and \(\A \models \alpha'_2[\tw, \tc_2]\).
  Furthermore, \(\A \models \delta^\sigma_{G, 2r+1}[\tv, \tw, \tc]\)
  implies that \(\NrA{\tv \tw \tc} = \NrA{\tv \tc_1} \uplus \NrA{\tw \tc_2}\).

  This shows that, for all \(\tv \in A^k\), \(\tw \in A^\ell\), \(\tc \in A^m\),
  we have \(\A \models \psi[\tv, \tw, \tc]\)
  if and only if
  \(\NrA{\tv \tw \tc} = \NrA{\tv \tc_1} \uplus \NrA{\tw \tc_2}\)
  and there is some \((\alpha'_1, \alpha'_2) \in \Delta'\)
  such that \(\A \models \delta^\sigma_{G_1, 2r+1}[\tv, \tc_1]
  \land \alpha'_1[\tv, \tc_1]\)
  and \(\A \models \delta^\sigma_{G_2, 2r+1}[\tw, \tc_2]
  \land \alpha'_2[\tw, \tc_2]\).

  Let \(\G' \subset \G_{k+\ell+m}\) be the class of graphs \(G'\)
  with \(G'[V_1] = G_1\), \(G'[V_2] = G_2\), and \(G' \neq G\).
  For every such graph \(G'\), there has to be an edge between \(V_1\) and \(V_2\),
  so \(G'\) has at most \(p-1\) connected components.
  Moreover, for every \(\tv \in A^k\), \(\tw \in A^\ell\), and \(\tc \in A^m\)
  with \(\NrA{\tv \tw \tc} \neq \NrA{\tv \tc_1} \cup \NrA{\tw \tc_2}\)
  (that is, there is an additional edge in \(\NrA{\tv \tw \tc}\)
  between a vertex from \(\nrA{\tv \tc_1}\) and a vertex from \(\nrA{\tw \tc_2}\)),
  \(\NrA{\tv \tc_1} \models \delta^\sigma_{G_1, 2r+1}[\tv, \tc_1]\),
  and \(\NrA{\tw \tc_2} \models \delta^\sigma_{G_2, 2r+1}[\tw, \tc_2]\),
  it holds that \(\A \models \delta^\sigma_{G'\!, 2r+1}[\tv, \tw, \tc]\)
  for some \(G' \in \G'\).

  By definition, we have
  \[\sem{\FOCCount{\tz}{\psi(\tv, \tw, \tz)}}^\A
  = \bigabs{\bigsetc{\tc \in A^m}{\A \models \delta^\sigma_{G, 2r+1}[\tv, \tw, \tc]
  \land \phi[\tv, \tw, \tc]}}.\]
  Since the formulas in \(\Delta'\) are mutually exclusive, this is equal to
  \[\sum_{(\alpha'_1, \alpha'_2) \in \Delta'} \bigabs{\bigl\{\tc \in A^m \bigmid
  \A \models \delta^\sigma_{G, 2r+1}[\tv, \tw, \tc]
  \land \alpha'_1[\tv, \tc_1] \land \alpha'_2[\tw, \tc_2]\bigr\}}.\]
  Lastly, by the definition of \(\G'\) above, the term is equal to
  \begin{align*}
    &\sum_{(\alpha'_1, \alpha'_2) \in \Delta'}
       \Bigl(
       \bigabs{\bigsetc{\tc_1 \in A^{m_1}}{\A \models \delta^\sigma_{G_1, 2r+1}[\tv, \tc_1]
       \land \alpha'_1[\tv, \tc_1]}}\\
    &\mkern60mu\cdot
       \bigabs{\bigsetc{\tc_2 \in A^{m_2}}{\A \models \delta^\sigma_{G_2, 2r+1}[\tw, \tc_2]
       \land \alpha'_2[\tw, \tc_2]}}\\
    &- \sum_{G' \in \G'}
       \bigabs{\bigsetc{\tc \in A^m}{\A \models \delta^\sigma_{G'\!, 2r+1}[\tv, \tw, \tc]
       \land \alpha'_1[\tv, \tc_1] \land \alpha'_2[\tw, \tc_2]}}\Bigr).
  \end{align*}
For every \((\alpha'_1, \alpha'_2) \in \Delta'\), we set
  \begin{align*}
    \psi_{\alpha'_1}(\tx, \tz_1)
    &\deff \delta^\sigma_{G_1, 2r+1}(\tx, \tz_1) \land \alpha'_1(\tx, \tz_1),\\
    \psi_{\alpha'_2}(\ty, \tz_2)
    &\deff \delta^\sigma_{G_2, 2r+1}(\ty, \tz_2) \land \alpha'_2(\ty, \tz_2),
  \end{align*}
  and we define
  \(t_{\alpha'_1}(\tx) \deff \FOCCount{\tz_1}{\psi_{\alpha'_1}(\tx, \tz_1)}\),
  \(t_{\alpha'_2}(\ty) \deff \FOCCount{\tz_2}{\psi_{\alpha'_2}(\ty, \tz_2)}\),
  \(u_{\alpha'_1}(\tx, \ty) \deff t_{\alpha'_1}(\tx)\),
  and
  \(u_{\alpha'_2}(\tx, \ty) \deff t_{\alpha'_2}(\ty)\).
  Furthermore, for all \(G' \in \G'\), we set
  \[\psi_{\alpha'_1, \alpha'_2, G'}(\tx, \ty, \tz) \deff \delta^\sigma_{G'\!, 2r+1}(\tx, \ty, \tz)
  \land \alpha'_1(\tx, \tz_1) \land \alpha'_2(\ty, \tz_2)\]
  and
  \(t_{\alpha'_1, \alpha'_2, G'}(\tx, \ty) \deff
  \FOCCount{\tz}{\psi_{\alpha'_1, \alpha'_2, G'}(\tx, \ty, \tz)}\).
  Then
  \[\sem{t(\tv, \tw)}^\A = \sum_{(\alpha'_1, \alpha'_2) \in \Delta'} \Bigl(
  \sem{u_{\alpha'_1}(\tv, \tw)}^\A \cdot \sem{u_{\alpha'_2}(\tv, \tw)}^\A
  - \sum_{G' \in \G'} \sem{t_{\alpha'_1, \alpha'_2, G'}(\tv, \tw)}^\A\Bigr)\]
  for all \(\tv \in V^k\) and \(\tw \in W^\ell\).
  Thus,
  \[\MTypeTAVW = \sum_{(\alpha'_1, \alpha'_2) \in \Delta'}
  \Bigl(\MTypeAVW{u_{\alpha'_1}} \odot \MTypeAVW{u_{\alpha'_2}}
  - \sum_{G' \in \G'} \MTypeAVW{t_{\alpha'_1, \alpha'_2, G'}}\Bigr),\]
  where \(\odot\) denotes the element-wise multiplication of two matrices.
By the induction hypothesis applied to \(\psi_{\alpha'_1}\), \(\psi_{\alpha'_2}\),
  and \(\psi_{\alpha'_1, \alpha'_2, G'}\),
  since \(G_1\), \(G_2\), and \(G'\) have at most \(p-1\) connected components
  and \(\alpha'_1\) and \(\alpha'_2\) are \(r\)-local,
  it holds that
  \(\rank\bigl(\MTypeAVW{u_{\alpha'_i}}\bigr) \leq T(\psi_{\alpha'_i})\) for \(i\in\set{1,2}\) and
  \(\rank\bigl(\MTypeAVW{t_{\alpha'_1, \alpha'_2, G'}}\bigr)
  \leq T(\psi_{\alpha'_1, \alpha'_2, G'})\).
  Hence, \(\rank\bigl(\MTypeTAVW\bigr)\) is at most
  \begin{align*}
    &\sum_{(\alpha'_1, \alpha'_2) \in \Delta'} \Bigl(
      \rank\bigl(\MTypeAVW{u_{\alpha'_1}}\bigr)
      \cdot \rank\bigl(\MTypeAVW{u_{\alpha'_2}}\bigr)
      + \sum_{G' \in \G'} \rank\bigl(\MTypeAVW{t_{\alpha'_1, \alpha'_2, G'}}\bigr)
    \Bigr)\\
    \leq
    &\ \sum_{(\alpha'_1, \alpha'_2) \in \Delta'} \Bigl(
      T(\psi_{\alpha'_1})
      \cdot T(\psi_{\alpha'_2})
      + \sum_{G' \in \G'} T(\psi_{\alpha'_1, \alpha'_2, G'})\bigr)
      \Bigr)\\
      \ffed
    &\ T(\psi).
  \end{align*}
  Since \(\Delta'\) is computable from \(\psi\)
  and, by the induction hypothesis,
  \(T(\psi')\) is computable from \(\psi'\)
  for all \(\psi'\) defined via a graph \(G'\) with at most \(p-1\) connected components,
  \(T(\psi)\) is computable from \(\psi\).

  All in all, this shows that \(T\) as defined above is a suitable computable function,
  and this concludes the proof.
\end{proof}

With \cref{lem:fo-dist-rank} at hand, we can now prove \cref{lem:fo-rank}.

\begin{proof}[Proof of \cref{lem:fo-rank}]
  Let \(\phi(\tx, \ty, \tz) \in \FO\) be an \(\FO\) formula,
  let \(k \deff \abs{\tx}\), \(\ell \deff \abs{\ty}\), \(m \deff \abs{\tz}\),
  let \(t(\tx, \ty) \deff \FOCCount{\tz}{\phi(\tx, \ty, \tz)}\),
  and \(r \deff 7^{\qr(\phi)}\).
  In order to prove \cref{lem:fo-rank},
  we define computable values \(T(\phi, p)\) such that, for every \(p \in \N\),
  every \(\sigma(\phi)\)-structure \(\A\),
  and all \(V, W \subseteq A\) that are \(\bigl((2r+1)(m+1)\bigr)\)-separated in \(G_\A\)
  by a set of size at most \(p\),
  we have \(\rank\bigl(\MTypeTAVW\bigr) \leq T(\phi, p)\).

  The cases \(k = 0\) and \(\ell = 0\) can be handled analogously
  to the corresponding cases in the proof of \cref{lem:fo-dist-rank}
  by setting \(T(\phi, p) \deff 1\).
  Hence, in the following, we assume \(k > 0\) and \(\ell > 0\).

  Using \cref{thm:gaifman},
  we transform \(\phi(\tx, \ty, \tz)\) into an equivalent formula
  \(\psi(\tx, \ty, \tz)\) in Gaifman normal form with radius at most \(r\).

  In the following, by induction on \(p \in \N\),
  we define \(T(\phi, p)\) such that \(T\) is computable
  and has the properties mentioned above.
Consider the base case \(p=0\).
  Let \(\Psi\) be the set of all formulas obtained from \(\psi\)
  by replacing every basic local sentence occurring in \(\psi\) by \(\top\) or \(\bot\).
  Then every formula in \(\Psi\) is \(r\)-local,
  and \(\Psi\) is a finite set that is computable from \(\phi\).
  Let \(\G \deff \G_{k + \ell + m}\).
  For every formula \(\psi'\) in \(\Psi\) and every graph \(G \in \G\),
  we let \(\psi'_G(\tx, \ty, \tz)
  \deff \delta^{\sigma(\phi)}_{G, 2r+1}(\tx, \ty, \tz) \land \psi'(\tx, \ty, \tz)\)
  and \(t_{\psi', G}(\tx, \ty) \deff \FOCCount{\tz}{\psi'_G(\tx, \ty, \tz)}\).

  Let \(T_{\ref*{lem:fo-dist-rank}}\) be the computable function from \cref{lem:fo-dist-rank}.
  For every \(\sigma(\phi)\)-structure \(\A\) and for all sets \(V, W \subseteq A\)
  that are \(\bigl((2r+1)(\abs{\tz}+1)\bigr)\)-separated in \(G_\A\), it holds that
  \(\rank\bigl(\MTypeAVW{t_{\psi', G}}\bigr) \leq T_{\ref*{lem:fo-dist-rank}}(\psi'_G)\).
  Let
  \(T(\phi, 0) \deff \max_{\psi' \in \Psi} \sum_{G \in \G} T_{\ref*{lem:fo-dist-rank}}(\psi'_G)\),
  which is computable from \(\phi\).

  For a \(\sigma(\phi)\)-structure \(\A\),
  we let \(\psi_\A \in \Psi\) be the formula obtained from \(\psi\)
  where we replace every basic local sentence \(\chi\) occurring in \(\psi\)
  by \(\top\) if \(\A \models \chi\),
  and we replace it by \(\bot\) if \(\A \not\models \chi\).
  Then \(\psi_\A \in \Psi\).
  Furthermore, for all \(\tv \in A^k\), \(\tw \in A^\ell\), and \(\tc \in A^m\),
  it holds that \(\A \models \phi[\tv, \tw, \tc]\) if and only if
  \(\A \models \psi_\A[\tv, \tw, \tc]\).
  Moreover, for \(\psi_{\A, G}(\tx, \ty, \tz) = \delta^{\sigma(\phi)}_{G, 2r+1}(\tx, \ty, \tz)
  \land \psi_\A(\tx, \ty, \tz)\)
  and \(t_{\A, G}(\tx, \ty) = \FOCCount{\tz}{\psi_{\A, G}(\tx, \ty, \tz)}\),
  we have
  \(\sem{t(\tv, \tw)}^\A = \sum_{G \in \G} \sem{t_{\A, G}(\tv, \tw)}^\A\).

  Let \(V, W \subseteq A\) be two sets that are
  \(\bigl((2r+1)(\abs{\tz}+1)\bigr)\)-separated in \(G_\A\).
  Then it holds that
  \(\rank\bigl(\MTypeAVW{t}\bigr) \leq \sum_{G \in \G} T_{\ref*{lem:fo-dist-rank}}(\psi_{\A, G})
  \leq T(\phi, 0)\).
  This finishes the proof for \(p = 0\).

  Now let \(p \in \Npos\),
  and suppose the result holds for all \(p' < p\).
  Let \(V, W \subseteq A\) be two sets
  that are \(\bigl((2r+1)(\abs{\tz}+1)\bigr)\)-separated in \(G_\A\)
  by a set \(S \subseteq A\) of size at most \(p\).
  If \(\abs{S} < p\), then the result already holds by induction hypothesis,
  so suppose \(\abs{S} = p\) with \(S = \set{s_1, \dots, s_p}\).

  Intuitively, in the following, we remove \(s_p\) from \(\A\),
  and we encode the tuples in relations of \(\A\) involving \(s_p\)
  by using new relation symbols.
  We call the resulting structure \(\A'\).
  The values in \(\MTypeTAVW\) will correspond to values in
  \(\MType{t'}{\A'}{V'/W'}\) for
  \(V' \deff V \setminus \set{s_p}\),
  \(W' \deff W \setminus \set{s_p}\),
  and for counting terms \(t'\) that can be computed from \(t\).
  Hence, to bound the rank of \(\MTypeTAVW\),
  it suffices to bound the rank of \(\MType{t'}{\A'}{V'/W'}\).
  Since \(V'\) and \(W'\) are \(\bigl((2r+1)(\abs{\tz}+1)\bigr)\)-separated in \(G_{\A'}\)
  by the set \(\set{s_1, \dots, s_{p-1}}\) of size \(p-1\),
  we can provide a computable bound by the induction hypothesis.

  Formally, we proceed as follows.
  We let \(\sigma'\) be the signature obtained from \(\sigma\)
  by replacing every relation symbol \(R \in \sigma\)
  by fresh relation symbols \(R_f\), for every \(f \colon [\ar(R)] \to \set{0,1}\),
  of arity \(\ar(R_f) \deff \abs{\setc{i \in [\ar(R)]}{f(i) = 0}}\).
  We let \(\A'\) be the \(\sigma'\)-structure obtained from \(\A\)
  with \(A' \deff A \setminus \set{s_p}\) by letting
  \(R_f(\A')\) contain all tuples that are obtained from a tuple \(\tv \in R(\A)\),
  with \(v_i = s_p\) if and only if \(f(i) = 1\),
  by dropping all entries from \(\tv\) that are equal to \(s_p\)
  (or, equivalently, all entries with indices \(i\) where \(f(i) = 1\)).
  For every formula \(\psi \in \FO\) and every function \(g \colon \free(\psi) \to \set{0,1}\),
  we let \(\psi_g\) be the formula obtained from \(\psi\) as follows.
  \begin{itemize}
    \item If \(\psi\) is of the form \(z_1 = z_2\) and
      \begin{itemize}
        \item \(g(z_1) = g(z_2) = 0\), then we set \(\psi_g \deff z_1 = z_2\);
        \item \(g(z_1) = g(z_2) = 1\), then we set \(\psi_g \deff \top\);
        \item \(g(z_1) \neq g(z_2)\), then we set \(\psi_g \deff \bot\).
      \end{itemize}
    \item If \(\psi\) is of the form \(R(\tz)\),
      then we set \(\psi_g \deff R_h(\tz_g)\),
      where \(h \colon [\ar(R)] \to \set{0,1}\)
      with \(h(i) = 0\) if and only if \(g(z_i) = 0\),
      and \(\tz_g\) is obtained from \(\tz\) by dropping all variables \(z_i\)
      with \(g(z_i) = 1\).
    \item If \(\psi\) is of the form \(\neg \psi'\),
      then we set \(\psi_g \deff \neg \psi'_g\).
    \item If \(\psi\) is of the form \(\psi' \lor \psi''\),
      then we set \(\psi_g \deff \psi'_g \lor \psi''_g\).
    \item If \(\psi\) is of the form \(\exists x\, \psi'\),
      then we set \(\psi_g \deff \psi'_{g'} \lor \exists x\, \psi'_{g''}\),
      where \(g', g'' \colon \free(\psi) \uplus \set{x} \to \set{0,1}\),
      \(g'(z) \deff g''(z) \deff g(z)\) for all \(z \in \free(\psi)\)
      and \(g'(x) \deff 1\) and \(g''(x) \deff 0\).
  \end{itemize}
  All in all, it holds that \(\free(\psi_g) \subseteq \setc{x \in \free(\psi)}{g(x) = 0}\)
  and \(\qr(\psi_g) = \qr(\psi)\).
  Moreover, for every formula \(\psi(z_1, \dots, z_q) \in \FO[\sigma]\)
  and every tuple \(\tv \in A^q\),
  it holds that \(\A \models \psi[\tv]\) if and only if
  \(\A' \models \psi_g[\tv_g]\),
  where \(g \colon \free(\psi) \to \set{0,1}\),
  with \(g(z_i) = 1\) if and only if \(v_i = s_p\),
  and \(\tv_g\) is obtained from \(\tv\) by dropping all indices \(i\)
  with \(g(z_i) = 1\), or, equivalently, all entries that are equal to \(s_p\).
In addition, we have \(G_{\A'} = G_\A[A'] = G_\A[A \setminus \set{s_p}]\).

Let \(F\) be the set of all functions \(f \colon \free(\phi) \to \set{0,1}\).
We set \(T(\phi, p) \deff \sum_{f \in F} T(\phi_f, p-1)\).
  By the induction hypothesis, the values \(T(\phi_f, p-1)\) are computable,
  so \(T(\phi, p)\) is computable as well.

  Now let \(\A\) be a \(\sigma(\phi)\)-structure,
  and let \(V, W \subseteq A\) be two sets that are \(\bigl((2r+1)(\abs{\tz}+1)\bigr)\)-separated
  in \(G_\A\) by a set \(S = \set{s_1, \dots, s_p} \subseteq A\) of size \(\abs{S} = p\).
  We shall prove that \(\rank\bigl(\MTypeTAVW\bigr) \leq T(\phi, p)\).

  Let \(\A'\) be the \(\sigma'\)-structure based on \(\A\) and \(s_p\) as described above.
  Then \(V' \deff V \setminus \set{s_p}\) and \(W' \deff W \setminus \set{s_p}\)
  are \(\bigl((2r+1)(\abs{\tz}+1)\bigr)\)-separated in \(G_{\A'}\)
  by the set \(S' \deff \set{s_1, \dots, s_{p-1}}\) of size \(p-1\).

  Let \(F_x\) be the set of functions \(f_x \colon \tilde{x} \to \set{0,1}\),
  let \(F_y\) be the set of functions \(f_y \colon \tilde{y} \to \set{0,1}\),
  and let \(F_z\) be the set of functions \(f_z \colon \tilde{z} \to \set{0,1}\).
  For every \(f_z \in F_z\),
  we let \[A^m_{f_z} \deff \bigsetc{\tc \in A^m}{c_i = s_p \iff f_z(z_i) = 1},\]
  and we let \(M_{f_z} \in \Z^{V^k \times W^\ell}\) with
  \[(M_{f_z})_{\tv, \tw} \deff \bigabs{\bigsetc{\tc \in A^m_{f_z}}{\A \models \phi[\tv, \tw, \tc]}}.\]
  Then \(A^m = \biguplus_{f_z \in F_z} A^m_{f_z}\) and \(\MTypeTAVW = \sum_{f_z \in F_z} M_{f_z}\).
  This shows that
  \(\rank\bigl(\MTypeTAVW\bigr) \leq \sum_{f_z \in F_z} \rank(M_{f_z})\).

  For \(f_x \in F_x\), \(f_y \in F_y\), and \(f_z \in F_z\),
  let \(V^k_{f_x} \deff \bigsetc{\tv \in V^k}{v_i = s_p \iff f_x(x_i) = 1}\),
  \(W^\ell_{f_y} \deff \bigsetc{\tw \in W^\ell}{w_i = s_p \iff f_y(y_i) = 1}\),
  and let \(M_{f_x, f_y, f_z} \in \Z^{V^k_{f_x} \times W^\ell_{f_y}}\) with
  \[(M_{f_x, f_y, f_z})_{\tv, \tw} \deff \bigabs{\bigsetc{\tc \in A^m_{f_z}}
  {\A \models \phi[\tv, \tw, \tc]}}.\]
  Let
  \(\phi_{f_x, f_y, f_z}(\tx_f, \ty_f, \tz_f) \deff \phi_f(\tx_f, \ty_f, \tz_f)\)
  for the function \(f \in F\)
  with \(f(x_i) = f_x(x_i)\), \(f(y_i) = f_y(y_i)\), and \(f(z_i) = f_z(z_i)\),
  and let
  \(t_{f_x, f_y, f_z}(\tx_f, \ty_f) \deff t_f(\tx_f, \ty_f)
  \deff \FOCCount{\tz_f}{\phi_f(\tx_f, \ty_f, \tz_f)}\).
  Then \((M_{f_x, f_y, f_z})_{\tv, \tw} = \bigl(\MType{t_f}{\A'}{V'/W'}\bigr)_{\tv_f, \tw_f}\),
  so \(\rank(M_{f_x, f_y, f_z}) = \rank\bigl(\MType{t_f}{\A'}{V'/W'}\bigr)\).
  Moreover, by the induction hypothesis, since \(V'\) and \(W'\)
  are \(\bigl((2r+1)(\abs{\tz}+1)\bigr)\)-separated in \(G_{\A'}\) by a set of size at most \(p-1\),
  and \(\qr(\phi_f) = \qr(\phi)\) and \(\abs{\tz_f} \leq \abs{\tz}\),
  we have \(\rank\bigl(\MType{t_f}{\A'}{V'/W'}\bigr) \leq T(\phi_f, p-1)\).
  Lastly, \(M_{f_z}\) consists of the submatrices \(M_{f_x, f_y, f_z}\)
  for all \(f_x \in F_x\) and \(f_y \in F_y\),
  so \(\rank(M_{f_z}) \leq \sum_{f_x \in F_x, f_y \in F_y} \rank(M_{f_x, f_y, f_z})\).
All in all, this shows that
  \[\rank\bigl(\MTypeTAVW\bigr) \leq \sum_{f_z \in F_z} \rank(M_{f_z})
  \leq \sum_{f \in F} T(\phi_f, p-1) = T(\phi, p),\]
  which concludes the proof.
\end{proof}

\subsection{Proofs of Theorems~\ref{thm:ladder-index-term} and~\ref{thm:graph-dimension}}
\label{subsec:appendix-ladder-index-term}

For the proof of \cref{thm:ladder-index-term}, we use the following result.

\begin{theorem}[{Uniform quasi-wideness for tuples \cite[Theorem~2.9]{PilipczukSiebertzTorunczyk_Types2018}}]
  \label{thm:uniform-quasi-wideness}
  Let \(p, r \in \N\),
  and let \(\C\) be a class of graphs that do not include \(K_p\) as a depth-\(18r\) minor.
  For every \(d \in \N\),
  there is a number \(s\) and a polynomial \(N \colon \N \to \N\)
  computable from \(p, r\), and \(d\) with the following property.

  For every \(G \in \C\), every \(m \in \N\),
  and every set \(X \subseteq (V(G))^d\) with \(\abs{X} \geq N(m)\),
  there are sets \(S \subseteq V(G)\) and \(Y \subseteq X\)
  with \(\abs{S} \leq s\) and \(\abs{Y} \geq m\)
  such that all distinct \(\tv, \tv' \in Y\)
  are \(r\)-separated by \(S\) in \(G\).
\end{theorem}

We can now prove \cref{thm:ladder-index-term}.

\ladderIndexTerm*

\begin{proof}
  Let \(r \colon \cgFOC \to \N\) and \(T \colon \cgFOC \times \N \to \N\)
  be the functions from \cref{lem:cgfoc-rank}.
  We set \(g \colon \cgFOC \to \N, \xi \mapsto 18r(\xi)\).

  Let \(t(\tx, \ty)\) be a \(\cgFOC\) counting term,
  let \(p \in \N\),
  and let \(\C\) be the class of graphs excluding \(K_p\) as a depth-\(g(t)\) minor.
  Let \(d \deff \abs{\tx} + \abs{\ty}\),
  and let \(s \in \N\) be the number and \(N \colon \N \to \N\) be the polynomial
  computed from \(r(t)\), \(p\), and \(d\) using \cref{thm:uniform-quasi-wideness}.
  Moreover, let \(L \deff f(t, p) \deff N(2T(t, s) + 6)\).
  Since \(N\) and \(s\) are computable from \(\phi\) and \(p\),
  we have that \(f\) is a computable function.

  Towards a contradiction, suppose there is a signature \(\sigma \supseteq \sigma(t)\),
  a \(\sigma\)-structure \(\A\) with \(G_\A \in \C\),
  and tuples \(\tv_1, \dots, \tv_L \in A^{\abs{\tx}}\)
  and \(\tw_1, \dots, \tw_L \in A^{\abs{\ty}}\)
  that form a \(t\)-ladder in \(\A\),
  that is, there exists a function \(g \colon \set{\tw_1, \dots, \tw_L} \to \Z\)
  such that \(\sem{t(\tv_i, \tw_j)}^\A = g(\tw_j)\) if and only if \(i \leq j\).
  In particular, the tuples \(\tv_1, \dots, \tv_L\) are pairwise distinct,
  and the same holds for the tuples \(\tw_1, \dots, \tw_L\).
  Let \(X \deff \setc{\tv_i \tw_i}{i \in [L]} \subseteq A^d\).
  By \cref{thm:uniform-quasi-wideness},
  for \(m \deff 2T(t, s) + 6\), since \(\abs{X} = L \geq N(m)\),
  there are sets \(S \subseteq A\) and \(Y \subseteq X\)
  with \(\abs{S} \leq s\) and \(\abs{Y} \geq m\) such that all distinct \(\tu, \tu' \in Y\)
  are \(r(t)\)-separated by \(S\) in \(G_\A\).
  Let \(Y' \subseteq Y\) with \(\abs{Y'} = m\),
  and let \(I \deff \setc{i \in [L]}{\tv_i \tw_i \in Y'}\).
  Let \(I_1, I_2\) be an alternating partition of \(I\), that is,
  for all successive \(i, j \in I_1\), there is exactly one \(k \in I_2\) with \(i < k < j\),
  and the same holds the other way round.
  Note that \(\abs{I} = \abs{Y'} = m\),
  and hence \(\abs{I_1} = \abs{I_2} = m/2 = T(t, s) + 3\).
  Let \(V \subseteq A\) be the set of vertices
  appearing in a tuple \(\tv_i \tw_i\) with \(i \in I_1\),
  and let \(W \subseteq A\) be the set of vertices
  appearing in a tuple \(\tv_i \tw_i\) with \(i \in I_2\),
  Since all distinct \(\tu, \tu' \in Y\) are \(r(t)\)-separated by \(S\) in \(G_\A\),
  it also holds that \(V\) and \(W\) are \(r(t)\)-separated by \(S\) in \(G_\A\).

  Now we can apply \cref{lem:cgfoc-rank} to \(V\) and \(W\),
  which shows that \(\rank\bigl(\MTypeTAVW\bigr) \leq T(t, s) = \abs{I_1} - 3\).

  On the other hand, we have that \((\tv_i)_{i \in I_1}\), \((\tw_i)_{i \in I_2}\)
  forms a \(t\)-ladder in \(\A\) via the function
  \(g' \colon \setc{\tw_j}{j \in I_2} \to \Z, \tw \mapsto g(\tw)\).
  Let \(M_{g'} \in \Z^{V^{\abs{\tx}} \times W^{\abs{\ty}}}\)
  with \((M_{g'})_{\tv, \tw_j} \deff -g'(\tw_j)\) for all \(\tv \in V^{\abs{\tx}}\) and \(j \in I_2\),
  and let \((M_{g'})_{\tv, \tw_j} \deff 0\) for all \(\tv \in V^{\abs{\tx}}\) and \(j \not\in I_2\).
  Furthermore, let \(M' \deff \MTypeTAVW + M_{g'}\).
  Then, for all \(i \in I_1\) and \(j \in I_2\),
  we have \(M'_{\tv_i, \tw_j} = 0\) if and only if \(i \leq j\).
  This shows that \(M'\) has rank at least \(\abs{I_1} - 1\),
  and \(M_{g'}\) has rank \(1\),
  so \(\MTypeTAVW\) has rank at least \(\abs{I_1} - 2\).
  This is a contradiction.

  Thus, this shows that there is no \(t\)-ladder in \(\A\) of size at least \(L = f(t, p)\),
  so the ladder index of \(t\) in \(\A\) is at most \(f(t, p)\).
\end{proof}

\graphDimension*

\begin{proof}
  Let \(\C\) be a nowhere dense graph class.
  By \cref{def:nowhere-dense},
  there is a function \(p \colon \N \to \N\) such that for all \(r \in \N\) and \(G \in \C\),
  it holds that \(G\) does not contain \(K_{p(r)}\) as a depth-\(r\) minor.

  Let \(f' \colon \cgFOC \times \N \to \N\) and \(g \colon \cgFOC \to \N\)
  be the computable functions from \cref{thm:ladder-index-term}.
  We set \(f \colon \cgFOC \to \N, \xi \mapsto f'\bigl(\xi, p(g(\xi))\bigr)\).
  If \(\C\) is effectively nowhere dense, then \(p\) is computable,
  so \(f\) is computable as a composition of computable functions.

  Let \(t\) be a \(\cgFOC\) counting term
  and let \(\A\) be a \(\sigma\)-structure for some \(\sigma \supseteq \sigma(t)\)
  such that \(G_\A \in \C\).
  Then \(G_\A\) does not contain \(K_{p(g(t))}\) as a depth-\(g(t)\) minor.
  Hence, by \cref{thm:ladder-index-term}, the ladder index of \(t\) in \(\A\)
  is at most \(f'\bigl(t, p(g(t))\bigr) = f(t)\).

  Let \(d \deff f(t)\).
  Then there is no \(t\)-ladder in \(\A\) of length \(d\).
  Thus, for all \(\tw_1, \dots, \tw_d \in A^{\abs{\ty}}\)
  and all functions \(g \colon \set{\tw_1, \dots, \tw_d} \to \Z\),
  there are no \(\tv_1, \dots, \tv_d \in A^{\abs{\tx}}\)
  such that \(\sem{t(\tv_i, \tw_j)}^\A = g(\tw_j)\) if and only if \(i \leq j\).
  Hence, for \(W \deff \set{\tw_1, \dots, \tw_d}\),
  and every function \(g \colon W \to \Z\),
  there is an index \(i \in [d]\)
  such that there is no function \(h \in \STypeTA{A/A}\)
  with \(h(\tw) = g(\tw)\) if \(\tw \in \set{\tw_i, \dots, \tw_d}\)
  and \(h(\tw) \neq g(\tw)\) if \(\tw \in W \setminus \set{\tw_i, \dots, \tw_d}\).
  Since \(g\) is arbitrary, this shows that \(W\) is not graph shattered by \(\STypeTA{A/A}\).
  Furthermore, since \(W\) was chosen arbitrarily,
  there is no set of size \(d\) that is graph shattered by \(\STypeTA{A/A}\),
  so the graph dimension of \(t\) in \(\A\) is less than \(d = f(t)\).
\end{proof}

\subsection{Proof of Corollary~\ref{cor:VCdimension}}
\label{subsec:appendix-VC-dimension}

For a formula \(\phi(\tx, \ty)\) and \(L \in \Npos\), a \emph{\(\phi\)-ladder} of length \(L\)
in a structure \(\A\) is a sequence of tuples \(\tv_1, \dots, \tv_L, \tw_1, \dots, \tw_L\)
such that \(\tv_i \in A^{\abs{\tx}}\) and \(\tw_i \in A^{\abs{\ty}}\) for all \(i \in [L]\)
and, for all \(i, j \in [L]\), it holds that \(\A \models \phi[\tv_i, \tw_j]\)
if and only if \(i \leq j\).
The smallest \(L\) for which there is no \(\phi\)-ladder of length \(L\) in \(\A\)
is called the \emph{ladder index of \(\phi(\tx,\ty)\) in \(\A\)}.

\VCdimension*

\Cref{cor:VCdimension} is a direct consequence of \cref{thm:graph-dimension,thm:ladder-index-term}
via the following result.

\begin{lemma}
  For every \(\cgFOC\) formula \(\phi(\tx, \ty)\),
  every \(\sigma\)-structure \(\A\) with \(\sigma \supseteq \sigma(\phi)\),
  and for the \(\cgFOC\) counting term \(t(\tx, \ty) \deff \FOCCount{()}{\phi}\),
  the VC dimension of \(\phi(\tx, \ty)\) in \(\A\)
  is equal to the graph dimension of \(t(\tx, \ty)\) in \(\A\),
  and the ladder index of \(\phi(\tx, \ty)\) in \(\A\)
  is at most the ladder index of \(t(\tx, \ty)\) in \(\A\).
\end{lemma}
\begin{proof}
  For a set \(X \subseteq A^{\abs{\ty}}\),
  we let \(h_X \colon A^{\abs{\ty}} \to \Z\) be the indicator function for \(X\)
  with \(h_X(x) \deff 1\) if \(x \in X\), and \(h_X(x) \deff 0\) otherwise.
  We set
  \begin{align*}
    \Hypo \deff \STypePhiA{A/A}
    &= \setc{\tp^\phi_\A(\tv/A)}{\tv \in A^{\abs{\tx}}}\\
    &= \bigsetc{\setc{\tw \in A^{\abs{\ty}}}{\A \models \phi[\tv, \tw]}}{\tv \in A^{\abs{\tx}}}
    \intertext{and}
    \Hypo' \deff \STypeTA{A/A}
    &= \setc{\tp^t_\A(\tv/A)}{\tv \in A^{\abs{\tx}}}\\
    &= \bigsetc{A^{\abs{\ty}} \to \Z, \tw \mapsto \sem{t(\tv, \tw)}^\A}{\tv \in A^{\abs{\tx}}}\\
    &= \bigsetc{A^{\abs{\ty}} \to \Z, \tw \mapsto \sem{\phi(\tv, \tw)}^\A}{\tv \in A^{\abs{\tx}}}.
  \end{align*}
  By definition, the VC dimension of \(\phi(\tx, \ty)\) in \(\A\) is the VC dimension of \(\Hypo\),
  and the graph dimension of \(t(\tx, \ty)\) in \(\A\) is the graph dimension of \(\Hypo'\).
  Furthermore, \(\Hypo' = \setc{h_X}{X \in \Hypo}\).
  Hence, by \cref{lem:vc-dim-vs-graph-dim},
  the VC dimension of \(\Hypo\) is equal to the graph dimension of \(\Hypo'\),
  so the VC dimension of \(\phi(\tx, \ty)\) in \(\A\)
  is equal to the graph dimension of \(t(\tx, \ty)\) in \(\A\).

  Now let \(\tv_1, \dots, \tv_L, \tw_1, \dots, \tw_L\) be a \(\phi\)-ladder in \(\A\).
  Then, for all \(i, j \in [L]\), it holds that \(\A \models \phi[\tv_i, \tw_j]\)
  if and only if \(i \leq j\).
  Hence, for \(g \colon \set{\tw_1, \dots, \tw_L} \to \Z\)
  with \(g(\tw_j) \deff 1\) for all \(j \in [L]\),
  we have \(\sem{t(\tv_i, \tw_j)}^\A = g(\tw_j)\) if and only if \(i \leq j\).
  This shows that \(\tv_1, \dots, \tv_L, \tw_1, \dots, \tw_L\) is also a \(t\)-ladder in \(\A\).
  Since this holds for every \(\phi\)-ladder in \(\A\),
  the ladder index of \(\phi(\tx, \ty)\) in \(\A\)
  is at most the ladder index of \(t(\tx, \ty)\) in \(\A\).
\end{proof}

\subsection{Proof of Theorem~\ref{thm:VCdensity}}
\label{subsec:appendix-number-of-types-nowhere-dense}

\VCdensity*

Our proof will rely on the following stronger version of \cref{lem:cgfoc-rank} for formulas:

\begin{corollary}
  \label{cor:number-of-types}
  There are computable functions \(T \colon \cgFOC \times \N \to \N\)
  and \(r \colon \cgFOC \to \N\) such that,
  for every \(\cgFOC\) formula \(\phi(\tx, \ty)\),
  every \(m \in \N\),
  every \(\sigma\)-structure \(\A\) for some \(\sigma \supseteq \sigma(\phi)\),
  and all \(V, W \subseteq A\) that are \(r(\phi)\)-separated in \(G_\A\)
  by a set of size at most \(m\),
  we have \(\absSTypePhiAVW \leq T(\phi, m)\).
\end{corollary}

\Cref{cor:number-of-types} follows immediately from the statement of \cref{lem:cgfoc-rank}
for the counting term \(t \deff \FOCCount{()}{\phi}\),
using the following fact on Boolean matrices over the real numbers \(\R\).

\begin{lemma}
  \label{lem:boolean-number-of-rows}
  Let \(M \in \set{0,1}^{m \times n}\) with \(\rank(M) \ffed c\) be a Boolean matrix
  with \(m\) pairwise distinct rows. Then \(m \leq 2^c\).
\end{lemma}
\begin{proof}
  For sets \(S \subseteq [m]\) and \(T \subseteq [n]\),
  let \(M_{S\times T}\) be the submatrix of \(M\) with entries
  whose coordinates are in \(S \times T\).

  \begin{claim*}
    There are sets \(S \subseteq [m]\) and \(T \subseteq [n]\)
    of size \(\abs{S} = \abs{T} = c\) with \(\rank(M_{S\times T}) = c\).
  \end{claim*}
  \begin{claimproof}
    Let \(S' \subseteq [m]\) be a set of \(\abs{S'} < c\) row indices
    so that all rows in \(S'\) are linearly independent.
    Since \(\abs{S'} < c\), there is an index \(i \in [m] \setminus S\)
    so that all rows in \(M_{(S' \cup \set{i}) \times [n]}\) are linearly independent.
    Therefore, we can inductively expand \(S'\) to a set \(S\) of size exactly \(c\)
    that satisfies \(\rank(M_{S \times [n]}) = c\).
    By applying the same reasoning to the columns of \(M_{S \times [n]}\),
    we obtain a set \(T \subseteq [n]\) with \(\rank(M_{S \times T}) = c\).
  \end{claimproof}

  Let \(S \subseteq [m]\) and \(T \subseteq [n]\)
  of size \(\abs{S} = \abs{T} = c\) with \(\rank(M_{S \times T}) = c\).
  The matrix \(M_{[m] \times T}\) contains at most \(2^c\) distinct rows.
  We show that it cannot contain two identical rows, which implies that \(m \leq 2^c\).

  Since \(M_{S \times T}\) has full rank, the equation system \(x \cdot M_{S \times T} = b\)
  has a unique solution for every \(b \in \set{0,1}^c\).
  Therefore, the system \(x' \cdot M_{S \times [n]} = b'\) has at most one solution
  \(x' \in \R^c\) for every \(b' \in \set{0,1}^n\).
  This implies that there can be no two rows in \(M\) that agree on \(T\)
  but differ on \([n] \setminus T\), since every row \(b' \in \set{0,1}^n\) of \(M\)
  can be expressed as \(x' \cdot M_{S \times [n]}\) for some \(x' \in \R^c\)
  due to \(M_{S \times [n]}\) having the same rank as the entire matrix \(M\).
\end{proof}

In the remainder of this section, using \cref{cor:number-of-types},
we prove the following theorem, which directly implies \cref{thm:VCdensity}.

\begin{theorem}
  \label{thm:number-of-types-nowhere-dense}
  Let \(\C\) be a nowhere dense graph class,
  and let \(\phi(\tx, \ty)\) be a \(\cgFOC\) formula.
  For every \(\epsilon > 0\), there exists a constant \(c \in \N\)
  such that for every \(\sigma\)-structure \(\A\) with \(\sigma \supseteq \sigma(\phi)\)
  and \(G_\A \in \C\)
  and every non-empty \(W \subseteq A\),
  we have \(\absSTypePhiA{A/W} \leq c \cdot \abs{W}^{\abs{\tx} + \epsilon}\).
\end{theorem}

For the proof of \cref{thm:number-of-types-nowhere-dense},
we rely on the following lemma on the neighbourhood complexity
in nowhere dense graph classes.
Let \(G\) be a graph, and let \(X \subseteq V(G)\).
For vertices \(v \in X\) and \(w \in V(G)\),
a path \(P\) from \(v\) to \(w\) in \(G\) is called \emph{\(X\)-avoiding}
if all vertices on the path except for \(v\) are not contained in \(X\).
For an \(r \in \N\) and \(w \in V(G)\),
the \emph{\(r\)-projection of \(w\) on \(X\)}, denoted by \(M_r^G(w,X)\),
is the set of all vertices \(v \in X\) that are connected to \(w\)
by an \(X\)-avoiding path of length at most \(r\).

\begin{lemma}[{\cite[Lemmas~21 and 22]{EickmeyerGKKPRS_NeighbourhoodComplexity}},
  see also {\cite[Lemma~3.3]{PilipczukSiebertzTorunczyk_Types2018}}]
  \label{lem:nowhere-dense-closure}
  Let \(\C\) be a nowhere dense graph class,
  There is a function \(f_\closure \colon \N \times \Qpos \to \N\)
  and a polynomial-time algorithm that,
  given a graph \(G \in \C\), \(X \subseteq V(G)\),
  \(r \in \N\), and \(\delta \in \Qpos\),
  computes a set \(\closure_{r, \delta}(X)\),
  called the \emph{\(r\)-closure of \(X\) with respect to \(\delta\)},
  with the following properties.
  \begin{enumerate}
    \item
      \label{item:closure-superset}
      \(X \subseteq \closure_{r, \delta}(X) \subseteq V(G)\),
    \item
      \label{item:closure-size}
      \(\abs{\closure_{r, \delta}(X)} \leq f_\closure(r, \delta) \cdot \abs{X}^{1+\delta}\), and
    \item
      \label{item:closure-projection-size}
      \(\abs{M_r^G\bigl(u, \closure_{r, \delta}(X)\bigr)} \leq f_\closure(r, \delta) \cdot \abs{X}^\delta\)
      for all \(u \in V(G) \setminus \closure_{r, \delta}(X)\).
  \end{enumerate}
  Moreover, for all \(X \subseteq V(G)\), it holds that
  \begin{enumerate}
    \setcounter{enumi}{3}
    \item
      \label{item:closure-projection-number}
      \(\abs{\bigsetc{M_r^G(u, X)}{u \in V(G)}} \leq f_\closure(r, \delta) \cdot \abs{X}^{1+\delta}\).
  \end{enumerate}
\end{lemma}

\begin{proof}[Proof of \cref{thm:number-of-types-nowhere-dense}]
  The proof is similar to the proof of the analogous result
  for graphs in \cite{PilipczukSiebertzTorunczyk_Types2018},
  using \cref{cor:number-of-types} instead of the corresponding result for graphs.

  Let \(\C\) be a nowhere dense graph class,
  let \(\phi(\tx, \ty)\) be a \(\cgFOC\) formula,
  and let \(\epsilon > 0\).
  Without loss of generality, we assume \(\epsilon \leq 1\).
  Let \(k \deff \abs{\tx}\),
  \(\ell \deff \abs{\ty}\),
  let \(r \colon \cgFOC \to \N\) and \(T \colon \cgFOC \times \N \to \N\)
  be the functions from \cref{cor:number-of-types},
  let \(p \colon \N \to \N\) be the function from \cref{def:nowhere-dense} for \(\C\),
  and let \(r \deff r(\phi)\) and \(p \deff p(36r)\).
  It holds that no graph \(G \in \C\) contains \(K_p\) as a depth-\(36r\) minor.

  By \cref{thm:uniform-quasi-wideness},
  there is a number \(s \in \N\) and a polynomial \(N \colon \N \to \N\)
  such that, for every graph \(G \in \C\), every \(m \in \N\),
  and every set \(X \subseteq (V(G))^k\) with \(\abs{X} \geq N(m)\),
  there are sets \(S \subseteq V(G)\) and \(Y \subseteq X\)
  with \(\abs{S} \leq s\) and \(\abs{Y} \geq m\)
  such that all distinct \(\tv, \tv' \in Y\)
  are \(2r\)-separated by \(S\) in \(G\).
  Let \(d\) be the degree of \(N\).

  Let \(\A\) be a \(\sigma\)-structure with \(\sigma \supseteq \sigma(\phi)\)
  and \(G_\A \in \C\),
  and let \(W \subseteq A\) be a non-empty set.
We set \(\delta \deff \frac{\epsilon}{4k+4d}\),
  and we let \(W' \deff \closure_{r, \delta}(W)\)
  be the \(r\)-closure of \(W\) with respect to \(\delta\),
  obtained via \cref{lem:nowhere-dense-closure} applied to \(G_\A\).
We shall prove that
  \begin{equation}
    \label{eqn:number-of-types-closure}
    \absSTypePhiA{A/W'} \in \bigO_{\epsilon, \phi}\bigl(\abs{W'}^{k+\epsilon'}\bigr)
    \text{ for } \epsilon' \deff \epsilon/2 > 0,
    \tag{\(\star\)}
  \end{equation}
  where \(\bigO_{\epsilon, \phi}(\cdot)\) omits factors depending only on \(\epsilon\) and \(\phi\).
  Since \(W \subseteq W'\), we have \(\absSTypePhiA{A/W} \leq \absSTypePhiA{A/W'}\).
  Moreover, by \cref{lem:nowhere-dense-closure},
  we have \(\abs{W'} = \abs{\closure_{r, \delta}(W)}
  \leq f_\closure(r, \delta) \cdot \abs{W}^{1+\delta}\),
  and we have
  \((1+\delta)(k + \epsilon') = (1+\delta)(k + \epsilon/2) \leq k + \epsilon\)
  by the choice of \(\delta\),
  so
  \[\absSTypePhiA{A/W} \in
    \bigO_{\epsilon, \phi}\Bigl(\bigl(f_\closure(r, \delta) \cdot \abs{W}^{1+\delta}\bigr)^{k+\epsilon'}\Bigr)
  \subseteq \bigO_{\epsilon, \phi}\bigl(\abs{W}^{k+\epsilon}\bigr),\]
  which is the statement of \cref{thm:number-of-types-nowhere-dense}.

  It remains to prove \eqref{eqn:number-of-types-closure}.
  Recall that
  \(\STypePhiA{A/W'} = \bigsetc{\tp^\phi_\A(\tv/W')}{\tv \in A^k}\).
  We partition the tuples
  \(\tv = (v_1, \dots, v_k) \in A^k\)
  based on their projection \(M_r^{G_\A}(\tv, W') \deff \bigcup_{i=1}^k M_r^{G_\A}(v_i, W')\)
  into sets \(V_1, \dots, V_t\).
  That is, two tuples \(\tv, \tv' \in A^k\)
  are contained in the same set \(V_j\) for some \(j \in [t]\) if and only if
  \(M_r^{G_\A}(\tv, W') = M_r^{G_\A}(\tv', W')\).
  By \cref{item:closure-projection-number} of \cref{lem:nowhere-dense-closure},
  there are at most \(f_\closure(r, \delta) \cdot \abs{W'}^{1+\delta}\)
  different projections of vertices in \(A = V(G_\A)\) on \(W'\), so we have
  \(t \in \bigO_{\epsilon, \phi}\bigl(\abs{W'}^{(1+\delta) k}\bigr)\).
  Hence, to prove \eqref{eqn:number-of-types-closure},
  it suffices to show that
  \begin{equation}
    \label{eqn:number-of-types-closure-projection}
    \bigabs{\bigsetc{\tp^\phi_\A(\tv/W')}{\tv \in V_j}}
    \in \bigO_{\epsilon, \phi}\bigl(\abs{W'}^{\epsilon''}\bigr),
    \tag{\(\star\star\)}
  \end{equation}
  for \(\epsilon'' \deff \epsilon' - k \delta > 0\)
  and all \(j \in [t]\), implying
  \(\absSTypePhiA{A/W'} \in \bigO_{\epsilon, \phi}\bigl(\abs{W'}^{(1+\delta)k}\abs{W'}^{\epsilon'-k\delta}\bigr)
  = \bigO_{\epsilon, \phi}\bigl(\abs{W'}^{k + \epsilon'}\bigr)\).

  Let \(j \in [t]\),
  and let \(X \deff M_r^{G_\A}(\tv, W')\) be the \(r\)-projection of \(\tv\) on \(W'\)
  for any (and, due to the definition of \(V_j\), for all) \(\tv \in V_j\).
  By \cref{item:closure-projection-size} of \cref{lem:nowhere-dense-closure},
  we have \(\abs{X} \leq k \cdot f_\closure(r, \delta) \cdot \abs{W}^\delta
  \in \bigO_{\epsilon, \phi}\bigl(\abs{W}^\delta\bigr)\).

  Let \(V'_j\) be a maximal subset of \(V_j\) such that
  all pairwise distinct tuples \(\tv, \tv'\) from \(V'_j\)
  have different types
  \(\tp^\phi_\A(\tv/W') \neq \tp^\phi_\A(\tv'/W')\).
  Note that
  \(\bigabs{\bigsetc{\tp^\phi_\A(\tv/W')}{\tv \in V_j}} = \bigabs{V'_j}\).
  Now let \(m \in \N\) be the maximum number with
  \(\bigabs{V'_j} \geq N(m)\).
  Then \(\bigabs{V'_j} < N(m+1) \in \bigO_{\epsilon, \phi}(m^d)\).

  By \cref{thm:uniform-quasi-wideness},
  as described above,
  there are sets \(S \subseteq A = V(G_\A)\) and \(Y \subseteq V'_j\)
  with \(\abs{S} \leq s\) and \(\abs{Y} \geq m\)
  such that all distinct \(\tv, \tv' \in Y\)
  are \(2r\)-separated by \(S\) in \(G_\A\).

  We partition \(Y\) into two sets \(Y_1 \uplus Y_2\),
  where \(Y_1\) contains all tuples that are \(r\)-separated by \(S\) from \(W'\),
  and \(Y_2\) contains the remaining tuples.
  By \cref{cor:number-of-types}, since all tuples in \(Y_1\) are \(r\)-separated by \(S\) from \(W'\),
  and all tuples in \(Y_1\) have distinct types,
  we know that \(\abs{Y_1} \leq T(\phi, s) \in \bigO_{\epsilon, \phi}(1)\).
  Moreover, for every tuple \(\tv \in Y_2\),
  there is a vertex \(w \in W'\) such that \(\tv\) and \(w\)
  are not \(r\)-separated by \(S\) in \(G_\A\).
  Note that we can choose \(w\) to be contained in \(X\).
  Furthermore, since all tuples in \(Y_2\) are mutually \(2r\)-separated by \(S\) in \(G_\A\),
  we know that for two distinct tuples \(\tv, \tv' \in Y_2\),
  the vertices in \(X\) connected to them by paths of length at most \(r\)
  avoiding \(S\) must also be distinct.
  This shows that \(\abs{Y_2} \leq \abs{X}\).
  Combined, we obtain that \(\abs{Y} \in \bigO_{\epsilon, \delta}(\abs{X})\).
  Furthermore, since \(\abs{Y} \geq m\),
  we have
  \[\abs{V'_j} \in \bigO_{\epsilon, \phi}(m^d)
    \ \subseteq
    \ \bigO_{\epsilon, \phi}(\abs{Y}^d)
    \ \subseteq
    \ \bigO_{\epsilon, \phi}(\abs{X}^d)
    \ \subseteq
    \ \bigO_{\epsilon, \phi}(\abs{W'}^{d\delta})
    \ \subseteq
    \ \bigO_{\epsilon, \phi}(\abs{W'}^{\epsilon''}),
  \]
  where the last inclusion holds because
  \(d\delta \leq \epsilon/4 \leq \epsilon/2 - k\delta = \epsilon''\)
  by the choice of \(\delta\).
  This proves \eqref{eqn:number-of-types-closure-projection}, which,
  as discussed above, implies the statement of \cref{thm:number-of-types-nowhere-dense}.
\end{proof}
 \section{Proofs Omitted in Section~\ref{sec:learning-formulas}}
\label{sec:appendix-learning}

\subsection{Proof of Theorem~\ref{thm:pac-learn-locally-bounded-expansion}}
\label{subsec:appendix-pac-learn}

\begin{lemma}[Sauer--Shelah--Perles \cite{Sauer1972,Shelah1972,VapnikChervonenkis1971}]
  \label{lem:sauer-shelah-perles}
  Let \(X\) be a set, and let \(\Hypo \subseteq 2^X\) have finite VC dimension \(d \in \N\).
  Then, for every finite subset \(Y \subseteq X\), it holds that
  \(\abs{\bigsetc{Y \cap H}{H \in \Hypo}} \leq \sum_{i=0}^d \binom{\abs{Y}}{i}\).
\end{lemma}

This implies the following.

\begin{lemma}
  \label{lem:vc-dimension-union}
  Let \(X\) be a set, and let \(\Hypo_1, \Hypo_2 \subseteq 2^X\)
  such that \(\Hypo_1\) has finite VC dimension \(d_1 \in \N\)
  and \(\Hypo_2\) has finite VC dimension \(d_2 \in \N\).
  It holds that \(\Hypo_1 \cup \Hypo_2\) has VC dimension at most \(d_1 + d_2 + 1\).
\end{lemma}
\begin{proof}
  Let \(X' \subseteq X\) be a finite set of size \(\abs{X'} = d_1 + d_2 + 2\).
  By \cref{lem:sauer-shelah-perles}, it holds that
  \(\abs{\bigsetc{X' \cap H}{H \in \Hypo_i}} \leq \sum_{i=0}^{d_j} \binom{d_1 + d_2 + 2}{i}\)
  for \(j \in [2]\).
  Hence, we have
  \begin{align*}
    &\abs{\bigsetc{X' \cap H}{H \in \Hypo_1 \cup \Hypo_2}}\\
    =~
    &\abs{\bigsetc{X' \cap H}{H \in \Hypo_1} \cup \bigsetc{X' \cap H}{H \in \Hypo_2}}\\
    \leq~
    &\abs{\bigsetc{X' \cap H}{H \in \Hypo_1}} + \abs{\bigsetc{X' \cap H}{H \in \Hypo_2}}\\
    \leq~
    &\sum_{j=0}^{d_1} \binom{d_1 + d_2 + 2}{j} + \sum_{j=0}^{d_2} \binom{d_1 + d_2 + 2}{j}\\
    =~
    &\sum_{j=0}^{d_1} \binom{d_1 + d_2 + 2}{j}
    + \sum_{j=0}^{d_2} \binom{d_1 + d_2 + 2}{d_1 + d_2 + 2 - j}\\
    <~
    &\sum_{j=0}^{d_1 + d_2 + 2} \binom{d_1 + d_2 + 2}{j}
    = 2^{d_1 + d_2 + 2}.
  \end{align*}
  This shows that \(X'\) is not shattered by \(\Hypo_1 \cup \Hypo_2\).
  Since \(X'\) was chosen arbitrarily,
  this implies that the VC dimension of \(\Hypo_1 \cup \Hypo_2\) is at most \(d_1 + d_2 + 1\).
\end{proof}

We now prove an analogous result on the graph dimension.

\begin{lemma}
  \label{lem:graph-dimension-union}
  Let \(X\) be a set, and let \(\Hypo_1, \Hypo_2\) be two sets of functions
  \(h \colon X \to \Z\) such that \(\Hypo_1\) has finite graph dimension \(d_1 \in \N\)
  and \(\Hypo_2\) has finite graph dimension \(d_2 \in \N\).
  It holds that \(\Hypo_1 \cup \Hypo_2\) has graph dimension at most \(d_1 + d_2 + 1\).
\end{lemma}

\begin{proof}
  We prove the statement based on the analogous statement for the VC dimension
  in \cref{lem:vc-dimension-union}.
  Let \(X' \subseteq X\) be a finite set of size \(\abs{X'} = d_1 + d_2 + 2\),
  and let \(g \colon X' \to \Z\).
  We show that for every such \(X'\) and \(g\),
  there is a set \(Y \subseteq X'\) such that there is no function
  \(h \in \Hypo \deff \Hypo_1 \cup \Hypo_2\) with
  \(h(x) = g(x)\) for all \(x \in Y\)
  and \(h(x) \neq g(x)\) for all \(x \in X' \setminus Y\).
  This shows that there is no set \(X'\) of size \(d_1 + d_2 + 2\)
  that is graph shattered by \(\Hypo\),
  so the graph dimension of \(\Hypo\) is at most \(d_1 + d_2 + 1\).

  Let \(g' \colon X \to \Z\) be a function with \(g'(x) = g(x)\) for all \(x \in X'\),
  let \(\Hypo'_i \deff \bigsetc{\setc{x \in X}{h(x) = g'(x)}}{h \in \Hypo_i} \subseteq 2^X\)
  for \(i \in [2]\),
  and let \(\Hypo' \deff \Hypo'_1 \cup \Hypo'_2 =
  \bigsetc{\setc{x \in X}{h(x) = g'(x)}}{h \in \Hypo}\).

  Since \(\Hypo_1\) has graph dimension \(d_1\),
  for every set \(X'_1 \subseteq X\) of size \(\smallabs{X'_1} = d_1 + 1\),
  there is no function \(g_1 \colon X'_1 \to \Z\) such that
  for every set \(Y_1 \subseteq X'_1\), there is a function \(h \in \Hypo_1\)
  with \(h(x) = g_1(x)\) for all \(x \in Y_1\)
  and \(h(x) \neq g_1(x)\) for all \(x \in X'_1 \setminus Y_1\).
  In particular, there is a set \(Y_1 \subseteq X'_1\)
  such that there is no function \(h \in \Hypo_1\)
  with \(h(x) = g'(x)\) for all \(x \in Y_1\)
  and \(h(x) \neq g'(x)\) for all \(x \in X'_1 \setminus Y_1\).
  This shows that \(Y_1 \not\in \setc{X'_1 \cap H}{H \in \Hypo'_1}\),
  so \(\Hypo'_1\) does not shatter \(X'_1\).
  Since this holds for every set \(X'_1 \subseteq X\) of size \(\smallabs{X'_1} = d_1 + 1\),
  this implies that \(\Hypo'_1\) has VC dimension at most \(d_1\).
  Analogously, it holds that \(\Hypo'_2\) has VC dimension at most \(d_2\).

  Using \cref{lem:vc-dimension-union}, we obtain that \(\Hypo'\)
  has VC dimension at most \(d_1 + d_2 + 1\),
  so \(\Hypo'\) does not shatter \(X'\).
  Thus, there is a set \(Y \subseteq X'\)
  with \(Y \not\in \setc{X' \cap H}{H \in \Hypo'}\),
  so there is no function \(h \in \Hypo\) with
  \(h(x) = g'(x) = g(x)\) for all \(x \in Y\)
  and \(h(x) \neq g'(x) = g(x)\) for all \(x \in X' \setminus Y\).
  As described above, since \(X'\) and \(g\) were chosen arbitrarily,
  this shows that \(\Hypo\) has graph dimension at most \(d_1 + d_2 + 1\).
\end{proof}

For a set \(X\), a probability distribution \(\D\) on \(X \times \set{0,1}\),
a sequence \(S = (x_1, \lambda_1), \dots, (x_s, \lambda_s)\)
over \(X \times \set{0,1}\),
and a subset \(X' \subseteq X\) of \(X\),
we let \(\err_\D(X') \deff \Pr_{(x, \lambda) \sim \D} \bigl(x \in X' \iff \lambda = 0\bigr)\)
and \(\err_S(X') \deff \frac{1}{s} \cdot \bigabs{\setc{i \in [s]}{x_i \in X' \iff \lambda_i = 0}}\).

\begin{lemma}[Uniform Convergence,
  {\cite[Sections~6.4 and~6.5]{Shalev-ShwartzBen-David_UnderstandingMachineLearning}}]
  \label{lem:uniform-convergence-vc}
  There is a computable function \(s_\UC\) with the following property.
  Let \(X\) be a set, let \(\Hypo \subseteq 2^X\) be of finite VC dimension \(c \in \N\),
  let \(\D\) be a probability distribution on \(X \times \set{0,1}\),
  and let \(\epsilon, \delta \in \Qpos\).
  For every sequence \(S = (x_1, \lambda_1), \dots, (x_s, \lambda_s)\)
  of length \(s \geq s_\UC(c, \epsilon, \delta)\) drawn \iid from \(\D\),
  where \((x_i, \lambda_i) \in X \times \set{0,1}\),
  with probability at least \(1-\delta\),
  it holds that \(\bigabs{\err_\D(h) - \err_S(h)} \leq \epsilon\) for all \(h \in \Hypo\).
  Moreover, \(s_\UC(c, \epsilon, \delta)\) is monotonically increasing in \(c\)
  and monotonically decreasing in \(\epsilon\) and \(\delta\).
\end{lemma}

For a set \(X\), a probability distribution \(\D\) on \(X \times \Z\),
a sequence \(S = (x_1, \lambda_1), \dots, (x_s, \lambda_s)\)
over \(X \times \Z\),
and a function \(h \colon X \to \Z\),
we let \(\err_\D(h) = \Pr_{(x, \lambda) \sim \D} \bigl(h(x) \neq \lambda\bigr)\)
and \(\err_S(h) = \frac{1}{s} \cdot \bigabs{\setc{i \in [s]}{h(x_i) \neq \lambda_i}}\).
We prove the following analogue of \cref{lem:uniform-convergence-vc} for the graph dimension.

\begin{lemma}
  \label{lem:uniform-convergence-graph-dimension}
  There is a computable function \(s_{\UC}\) with the following property.
  Let \(X\) be a set, let \(\Hypo\) be a set of functions \(h \colon X \to \Z\)
  such that \(\Hypo\) has finite graph dimension \(c \in \N\),
  let \(\D\) be a probability distribution on \(X \times \Z\),
  and let \(\epsilon, \delta \in \Qpos\).
  For every sequence \(S = (x_1, \lambda_1), \dots, (x_s, \lambda_s)\) over \(X \times \Z\)
  of length \(s \geq s_\UC(c, \epsilon, \delta)\) drawn \iid from \(\D\),
  with probability at least \(1-\delta\),
  it holds that \(\bigabs{\err_\D(h) - \err_S(h)} \leq \epsilon\) for all \(h \in \Hypo\).
  Moreover, \(s_\UC(c, \epsilon, \delta)\) is monotonically increasing in \(c\)
  and monotonically decreasing in \(\epsilon\) and \(\delta\).
\end{lemma}
\begin{proof}
  Let \(s_\UC\) be the computable function from \cref{lem:uniform-convergence-vc},
  let \(s \in \Npos\) with \(s \geq s_\UC(c, \epsilon, \delta)\),
  and let \(\CS\) be the set of all sequences \(S = (x_1, \lambda_1), \dots, (x_s, \lambda_s)\)
  over \(X \times \Z\) such that
  \(\abs{\err_\D(h) - \err_S(h)} \leq \epsilon\) for all \(h \in \Hypo\).
  We show that \(\D^s(\CS) \geq 1-\delta\),
  which proves \cref{lem:uniform-convergence-graph-dimension}.\footnote{\(\D^s\) is the probability distribution on \((X \times \Z)^s\)
    where we draw \(s\) times \iid from \(\D\).}

  We set \(X' \deff X \times \Z\).
  For a function \(h \colon X \to \Z\),
  let \(X'_h \deff \bigsetc{\bigl(x, h(x)\bigr)}{x \in X} \subset X'\),
  and set \(\Hypo' \deff \setc{X'_h}{h \in \Hypo} \subseteq 2^{X'}\).

  \begin{claim*}
    \(\Hypo'\) has VC dimension at most \(c\).
  \end{claim*}
  \begin{claimproof}
    Let \(Y' \subseteq X'\) be a set of size \(\abs{Y'} = c+1\).
    If there are \((x, i_1), (x, i_2) \in Y'\) for some \(x \in X\) and \(i_1, i_2 \in \Z\)
    with \(i_1 \neq i_2\),
    then there is no \(X'_h \in \Hypo'\) with \(Y' \subseteq X'_h\),
    since either \(i_1 \neq h(x)\) or \(i_2 \neq h(x)\),
    so \(Y' \not\in \setc{Y' \cap X'_h}{X'_h \in \Hypo'}\).
    This shows that \(Y'\) is not shattered by \(\Hypo'\).

    In the other case, for every \(x \in X\),
    there is at most one \(i \in \Z\) with \((x,i) \in Y'\).
    Let \(Y\) be the set of all \(x \in X\) such that \((x,i) \in Y'\) for some \(i \in \Z\).
    Then \(\abs{Y} = \abs{Y'} = c + 1\).
    Let \(g \colon Y \to \Z\) be the function with \(g(x) \deff i\)
    for every \((x,i) \in Y'\).
    Since \(\Hypo\) has graph dimension \(c\),
    we know that \(Y\) is not graph shattered by \(\Hypo\).
    Hence, there exists a set \(A \subseteq Y\) such that for every function \(h \in \Hypo\),
    we have \(g(x) \neq h(x)\) for some \(x \in A\)
    or \(g(x) = h(x)\) for some \(x \in Y \setminus A\).
    Let \(A' \subseteq Y'\) be the set
    that contains exactly those \((x,i) \in Y'\) with \(x \in A\).
    Then, for every \(X'_h \in \Hypo'\),
    we have \(g(x) \neq h(x)\) for some \(x \in A\),
    which implies \((x,i) \in A'\) and \((x,i) \not\in X'_h\),
    or we have \(g(x) = h(x)\) for some \(x \in Y \setminus A\),
    which implies \((x,i) \not\in A'\) and \((x,i) \in X'_h\).
    Thus, for every \(X'_h \in \Hypo'\),
    it holds that \(A' \neq X'_h \cap Y'\).
    As in the first case, this shows that \(Y'\) is not shattered by \(\Hypo'\).

    Hence, in all cases, \(Y'\) is not shattered by \(\Hypo'\).
    Since \(Y'\) was chosen arbitrarily of size \(c+1\),
    there is no set of size \(c+1\) that is shattered by \(\Hypo'\),
    implying that \(\Hypo'\) has VC dimension at most \(c\).
  \end{claimproof}

  Let \(\D'\) be the probability distribution on \(X' \times \set{0,1}\)
  obtained from \(\D\) by setting \(\D'(A) \deff \D\bigl(\setc{x \in X'}{(x,1) \in A}\bigr)\)
  for every \(A \subseteq X' \times \set{0,1}\).
  Moreover, let \(\CS'\) be the set of all sequences
  \(S' = (x_1, \lambda_1), \dots, (x_s, \lambda_s)\)
  over \(X' \times \set{0,1}\) with
  \(\abs{\err_{\D'}(h') - \err_{S'}(h')} \leq \epsilon\) for all \(h' \in \Hypo'\).
  Equivalently, for a sequence \(S' \in \CS'\),
  we have that \(\bigabs{\err_{\D'}(X'_h) - \err_{S'}(X'_h)} \leq \epsilon\) for all \(h \in \Hypo\).
  Since \(\Hypo'\) has VC dimension at most \(c\)
  and \(s \geq s_\UC(c, \epsilon, \delta)\), it holds that
  \[\Pr_{S' \sim (\D')^s}\bigl(\abs{\err_{\D'}(h') - \err_{S'}(h')}
  \leq \epsilon \text{ for all } h' \in \Hypo'\bigr) \geq 1-\delta\]
  by \cref{lem:uniform-convergence-vc}.
  Hence, \((\D')^s(\CS') \geq 1-\delta\).
  Let \(\CS'' \subseteq \CS'\) be the set that contains only those sequences
  \(S' = (x_1, \lambda_1), \dots, (x_s, \lambda_s)\)
  from \(\CS'\) with \(\lambda_i = 1\) for all \(i \in [s]\).
  By the definition of \(\D'\), we have
  \((\D')^s(\CS'') = (\D')^s(\CS')\).
  Furthermore, for a hypothesis \(h \in \Hypo\), it holds that
  \begin{align*}
    \err_{\D'}(X'_h)
    &= \Pr_{((x, i), \lambda) \sim \D'}\bigl((x,i) \not\in X'_h \iff \lambda = 1\bigr)\\
    &= \Pr_{(x,i) \sim \D}\bigl((x,i) \not\in X'_h\bigr)\\
    &= \Pr_{(x,i) \sim \D}\bigl(h(x) \neq i\bigr)\\
    &= \err_\D(h).
  \end{align*}
  Moreover, for a sequence \(S = (x_1, \lambda_1), \dots, (x_s, \lambda_s)\)
  over \(X \times \Z\) and the sequence
  \(S' \deff \bigl((x_1, \lambda_1), 1\bigr), \dots, \bigl((x_s, \lambda_s), 1\bigr)\),
  it holds that
  \begin{align*}
    \err_{S'}(X'_h)
    &= {\textstyle \frac{1}{s}} \cdot \abs{\setc{i \in [s]}{(x_i, \lambda_i) \not\in X'_h}}\\
    &= {\textstyle \frac{1}{s}} \cdot \abs{\setc{i \in [s]}{h(x_i) \neq \lambda_i}}\\
    &= \err_S(h).
  \end{align*}
  Hence, \(\abs{\err_\D(h) - \err_S(h)} = \bigabs{\err_{\D'}(X'_h) - \err_{S'}(X'_h)}\)
  for all \(h \in \Hypo\),
  so \(S \in \CS\) if and only if \(S' \in \CS''\).
  This shows that \(\D^s(\CS) = (\D')^s(\CS'')\).
  Thus, all in all, it holds that \(\D^s(\CS) \geq 1-\delta\),
  which finishes this proof.
\end{proof}

In order to prove \cref{thm:pac-learn-locally-bounded-expansion}, we first give an ERM algorithm that can then be turned into an agnostic-PAC-learning algorithm.

\begin{proposition}
  \label{prop:erm-locally-bounded-expansion}
  For every graph class \(\C\) of locally bounded expansion,
  there is a function \(f\) and an algorithm that does the following.
  Given a signature \(\sigma\), a \(\sigma\)-structure \(\A\) with \(G_\A \in \C\),
  natural numbers \(k \in \N\) and \(\ell, m \in \Npos\) with \(k + \ell \leq m\),
  a rational number \(\alpha \in \Qpos\),
  a finite set of integers \(I\),
  and a sequence \(S = (\tw_1, \lambda_1), \dots, (\tw_s, \lambda_s)\)
  over \(A^\ell \times \Z\)
  for some \(s \in \Npos\),
  the algorithm runs in time \(f(\alpha, \sigma, I, m) \cdot s \cdot \abs{A}^{\max(k, 1+\alpha)}\),
  and it returns a list consisting of the tuples \(\bigl(t, \tv, \err_S(t, \tv)\bigr)\)
  for all \(t(\tx, \ty) \in T[\sigma, k, \ell, m, I]\) and \(\tv \in A^{\abs{\tx}}\),
  sorted in ascending lexicographic order on \(\bigl(\err_S(t, \tv), \abs{\tv}, \abs{t}\bigr)\).
\end{proposition}

\begin{proof}
  Let \(f_{\ref*{thm:answering-enumeration}}\) be the function from \cref{thm:answering-enumeration}.
  Moreover, for every \(i \in [0,s]\), let \(L_i\) be an empty list.

  For every counting term \(t(\tx, \ty) \in T[\sigma, k, \ell, m, I]\),
  in ascending lexicographic order on \((\abs{\tx}, \abs{t})\), we do the following.
  First, we run the preprocessing from \cref{thm:answering-enumeration} on \(t\), \(\A\), and \(\alpha\).
  This takes time \(f_{\ref*{thm:answering-enumeration}}(t, \sigma(t), \alpha) \cdot \abs{A}^{1+\alpha}\).
  Next, for every tuple \(\tv \in A^{\abs{\tx}}\) and every \(i \in [s]\),
  we compute \(\sem{t(\tv, \tw_i)}^\A\) based on the preprocessing above.
  This gives us the value
  \(\err_S(t, \tv) =
  \frac{1}{s} \cdot \abs{\setc{i \in [s]}{\sem{t(\tv, \tw_i)}^\A \neq \lambda_i}}\),
  and we append \(\bigl(t, \tv, \err_S(t, \tv)\bigr)\) to the list \(L_{s \cdot \err_S(t, \tv)}\).
  Over all tuples \(\tv \in A^{\abs{\tx}}\) and all \(i \in [s]\), this takes time
  \(f_{\ref*{thm:answering-enumeration}}(t, \sigma(t), \alpha) \cdot s \cdot \abs{A}^{\abs{\tx}}
  \leq f_{\ref*{thm:answering-enumeration}}(t, \sigma(t), \alpha) \cdot s \cdot \abs{A}^k\).

  In the end, we output the concatenation of the lists \(L_0, L_1, \dots, L_s\).
  All in all, since the set \(T[\sigma, k, \ell, m, I]\) is computable
  from \(\sigma\), \(m\), and \(I\), there is a function \(f\) such that the described algorithm
  runs in time \(f(\alpha, \sigma, I, m) \cdot s \cdot \abs{A}^{\max(k, 1+\alpha)}\).
\end{proof}

We are now ready to prove \cref{thm:pac-learn-locally-bounded-expansion}.

\PACLearnLocallyBoundedExpansion*

\begin{proof}
  Let \(\C\) be a graph class of effectively locally bounded expansion,
  and let \(f_{\textup{graph}}\) be the computable function
  from \cref{thm:graph-dimension} for \(\C\)
  such that for every signature \(\sigma\), every \(\cgFOC[\sigma]\) counting term \(t(\tx, \ty)\),
  and every \(\sigma\)-structure \(\A\) with \(G_\A \in \C\),
  the graph dimension of \(t(\tx, \ty)\) in \(\A\) is at most \(f_{\textup{graph}}(t)\).
  Let \(s_{\UC}\) be the computable function from \cref{lem:uniform-convergence-graph-dimension},
  let \(c \deff \bigabs{T[\sigma, k, \ell, m, I]} +
  \sum_{t \in T[\sigma, k, \ell, m, I]} f_{\textup{graph}}(t)\),
  and \(s \deff s_{\UC}(c, \epsilon, \delta)\).
  Then \(s\) is computable from \(\sigma\), \(k\), \(\ell\), \(m\), \(I\),
  \(\epsilon\), and \(\delta\).

  Our algorithm for \cref{thm:pac-learn-locally-bounded-expansion} draws a sequence \(S\)
  of \(s\) samples \((\tw_1, \lambda_1), \dots, (\tw_s, \lambda_s)\) \iid from \(\D\)
  using the oracle and then runs the algorithm from \cref{prop:erm-locally-bounded-expansion}
  on \(\A\), \(k\), \(\ell\), \(m\), \(\alpha\), \(I\), and \(S\).
  The running-time bound and the order of the output
  for \cref{thm:pac-learn-locally-bounded-expansion} directly follow from
  \cref{prop:erm-locally-bounded-expansion}.

  Let \(\Hypo \deff \bigcup_{t \in T[\sigma, k, \ell, m, I]} S^t_\A(A/A)\).
  Then, by \cref{lem:graph-dimension-union}, the graph dimension of \(\Hypo\) is at most \(c\).
  Hence, by \cref{lem:uniform-convergence-graph-dimension},
  with probability at least \(1-\delta\) over the oracle calls to \(\D\),
  it holds that \(\abs{\err_\D(h) - \err_S(h)} \leq \epsilon\)
  for all \(h \in \Hypo\),
  so \(\abs{\err_\D(t, \tv) - \err_S(t, \tv)} \leq \epsilon\)
  for all \(t(\tx, \ty) \in T[\sigma, k, \ell, m, I]\)
  and \(\tv \in A^{\abs{\tx}}\).
  By \cref{prop:erm-locally-bounded-expansion},
  the triples that we output are of the form \(\bigl(t, \tv, \err_S(t, \tv)\bigr)\).
  Thus, the error bound from \cref{thm:pac-learn-locally-bounded-expansion} is also met.
\end{proof}

\subsection{Proof of Theorem~\ref{thm:pac-learn-bounded-degree}}

For the proof of \cref{thm:pac-learn-bounded-degree},
we use the following two results from~\cite{KuskeSchweikardt_FOCN}.

\begin{theorem}[{\cite[Theorem~3.2]{KuskeSchweikardt_FOCN}}]
  \label{thm:hanf-normal-form}
  For every \(d \in \N\) and every \(\FOC\) formula \(\phi(\tx)\),
  there is an \(\FOC[\sigma(\phi)]\) formula \(\phi'(\tx)\)
  that is computable from \(d\) and \(\phi\) such that
  \(\phi\) and \(\phi'\) are \(d\)-equivalent and
  \(\phi'(\tx)\) is a Boolean combination of \(\FO[\sigma(\phi)]\) formulas
  and \(\FOC[\sigma(\phi)]\) sentences.
\end{theorem}

\begin{theorem}[{\cite[Corollary~5.7]{KuskeSchweikardt_FOCN}}]
  \label{thm:foc-enumeration}
  There are a computable function \(f\) and an algorithm that does the following.
  Given a \(\sigma\)-structure \(\A\) for some signature \(\sigma\),
  given \(d \in \N\) such that \(G_\A\) has degree at most \(d\),
  and given an \(\FOC[\sigma]\) formula \(\phi(\tx)\),
  after preprocessing in time \(f(\abs{\phi}, d) \cdot \abs{A}\),
  the algorithm enumerates all tuples \(\tv \in A^{\abs{\tx}}\)
  such that \(\A \models \phi[\tv]\) with \(f(\abs{\phi}, d)\) delay.
\end{theorem}

We can now prove \cref{thm:pac-learn-bounded-degree}.

\label{subsec:appendix-pac-learn-bounded-degree}

\PACLearnBoundedDegree*

\begin{proof}
  Let \(g_{\ref*{thm:ladder-index-term}}\) and \(f_{\ref*{thm:ladder-index-term}}\)
  be the computable functions from \cref{thm:ladder-index-term},
  and let \(f_{\ref*{thm:fo-enumeration}}\) be the computable function
  from \cref{thm:fo-enumeration}.

  Let \(\sigma\) be a signature, let \(d \in \N\),
  and let \(\A\) be a \(\sigma\)-structure such that \(G_\A\) has degree at most \(d\).
  Then, for every \(r \in \N\) and every \(v \in A\),
  we have \(\bigabs{\neighbA{r+1}{v}} < (d+1)^{r+2} + 1 \ffed p(r,d)\),
  so \(G_\A\) excludes \(K_{p(r,d)}\) as a depth-\(r\) minor.
  Hence, by \cref{thm:ladder-index-term},
  analogously to the proof of \cref{thm:graph-dimension},
  the graph dimension of a \(\cgFOC[\sigma]\) counting term \(t\) in \(\A\) is at most
  \(f_{\textup{graph}}(t,d) \deff f_{\ref*{thm:ladder-index-term}}
  \bigl(t, p\bigl(g_{\ref*{thm:ladder-index-term}}(t), d\bigr) \bigr)\),
  which is computable from \(t\) and \(d\).

  Let \(T \deff T_{\FOC}[\sigma, k, \ell, m, I]\),
  and let \(t(\tx, \ty) \in T\).
  Then \(t\) is a polynomial of integers from \(I\)
  and of \#-terms \(u\) of the form \(\FOCCount{\tz}{\phi}\)
  for an \(\FOC[\sigma]\) formula \(\phi\).
  For every such \#-term \(u\), we apply \cref{thm:hanf-normal-form} to \(\phi\) and \(d\),
  and we obtain an \(\FOC[\sigma]\) formula \(\phi'\)
  that is a Boolean combination of \(\FO[\sigma]\) formulas and \(\FOC[\sigma]\) sentences.
  Let \(\Phi\) be the set of all \(\FO[\sigma]\) formulas obtained from \(\phi'\)
  by replacing the \(\FOC\) sentences by \(\top()\) or \(\bot()\),
  let \(T_u\) be the set of all \(\cgFOC[\sigma]\) counting terms
  of the form \(\FOCCount{\tz}{\psi}\) for a formula \(\psi \in \Phi\),
  and let \(T'_t\) be the set of all \(\cgFOC[\sigma]\) counting terms obtained from \(t\)
  by replacing every \#-term \(u\) in \(t\) by a counting term from \(T_u\).
  Then \(T'_t\) can be computed from \(t\) and \(d\).
  Furthermore, for every \(\sigma\)-structure \(\A\) of degree at most \(d\),
  there is a counting term \(t' \in T'_t\) that is equivalent to \(t\) on \(\A\).
Let \(T' \deff \bigcup_{t \in T} T'_t\).
  Then \(T' \subseteq \cgFOC[\sigma]\) is computable from
  \(\sigma\), \(k\), \(\ell\), \(m\), \(I\), and \(d\).

  Let \(s_{\UC}\) be the computable function from \cref{lem:uniform-convergence-graph-dimension},
  let \(c \deff \abs{T'} + \sum_{t' \in T'} f_{\textup{graph}}(t', d)\),
  and \(s \deff s_{\UC}(c, \epsilon, \delta)\).
  Then \(s\) is computable from \(\sigma\), \(k\), \(\ell\), \(m\), \(I\), \(d\),
  \(\epsilon\), and \(\delta\).

  Our algorithm for \cref{thm:pac-learn-bounded-degree} computes \(T\) and \(s\).
  Next, it draws a sequence \(S = (\tw_1, \lambda_1), \dots, (\tw_s, \lambda_s)\)
  from \(\D\) using the oracle.

  Let \(\sigma' \deff \sigma \uplus \setc{R_{i,j}}{i \in [s], j \in [\ell]}\)
  for fresh unary relation symbols \(R_{i,j}\),
  and let \(\A_S\) be the \(\sigma'\)-expansion of \(\A\)
  with \(R_{i,j}(\A_S) \deff \set{(\tw_i)_j}\).

  For every counting term \(t(\tx, \ty) \in T\) and every subsequence \(S'\) of \(S\),
  we define an \(\FOC[\sigma']\) formula \(\phi_{t, S'}(\tx)\)
  such that, for every tuple \(\tv \in A^{\abs{\tx}}\),
  we have \(\A_S \models \phi_{t, S'}[\tv]\)
  if and only if \(\sem{t(\tv, \tw)}^\A = \lambda\) for all \((\tw, \lambda) \in S'\)
  and \(\sem{t(\tv, \tw)}^\A \neq \lambda\) for all \((\tw, \lambda) \in S \setminus S'\).
  For this, we let
  \[\phi_{t, S'}(\tx) \deff \forall y_1 \dots \forall y_\ell \Bigl(
    \mkern-5mu\Land_{(\tw_i, \lambda_i) \in S'}\mkern-5mu \psi_{t, i, +}(\tx, \ty)
    \land
    \mkern-18mu\Land_{(\tw_i, \lambda_i) \in S \setminus S'}\mkern-5mu \psi_{t, i, -}(\tx, \ty)
  \Bigr)\]
  with
  \[\psi_{t, i, +}(\tx, \ty) \deff \bigl(\Land_{j \in [\ell]} R_{i,j}(y_j)\bigr)
  \rightarrow \Pred_=\bigl(t(\tx, \ty), \lambda_i\bigr)\]
  and
  \[\psi_{t, i, -}(\tx, \ty) \deff \bigl(\Land_{j \in [\ell]} R_{i,j}(y_j)\bigr)
  \rightarrow \neg \Pred_=\bigl(t(\tx, \ty), \lambda_i\bigr).\]
  It can easily be verified that \(\phi_{t, S'}\) satisfies the properties above,
  and there is a bound for the length \(\smallabs{\phi_{t, S'}}\) of \(\phi_{t, S'}\)
  that can be computed from \(t\), \(s\), and \(\ell\),
  which are computable from \(\sigma\), \(k\), \(\ell\), \(m\), \(I\), \(d\),
  \(\epsilon\), and \(\delta\).

  In our algorithm for \cref{thm:pac-learn-bounded-degree},
  for every \(t(\tx, \ty) \in T\) and every subsequence \(S'\) of \(S\),
  we run the preprocessing from the \(\FOC\)-enumeration result \cref{thm:foc-enumeration}
  on \(\A_S\) and \(\phi_{t, S'}\)
  in time \(f_{\ref*{thm:foc-enumeration}}\bigl(\smallabs{\phi_{t, S'}}, d\bigr) \cdot \abs{A}\).
  This finishes the preprocessing of our algorithm,
  and there is a computable function \(f_1\)
  such that the preprocessing takes time at most
  \(f_1(\epsilon, \delta, \sigma, I, m, d) \cdot \abs{A}\).

  In the enumeration phase, we iterate over all subsequences \(S'\) of \(S\)
  with decreasing length, starting with the longest subsequence \(S' = S\).
  We iterate over all counting terms \(t(\tx, \ty) \in T\),
  in ascending lexicographic order on \((\abs{\tx}, \abs{t})\).
  For every such counting term,
  based on the preprocessing above,
  we iterate over all tuples \(\tv \in A^{\abs{\tx}}\)
  that satisfy \(\A_S \models \phi_{t, S'}[\tv]\),
  and we output \((t, \tv, 1 - \abs{S'} / \abs{S})\) for every such tuple.
  Because of the preprocessing above, for every counting term \(t\) and subsequence \(S'\),
  the enumeration has
  \(f_{\ref*{thm:foc-enumeration}}\bigl(\smallabs{\phi_{t, S'}}, d\bigr)\) delay.
  Since an upper bound on the number of counting terms \(t \in T\)
  and the number of subsequences \(S'\) of \(S\) can be computed from
  \(\epsilon\), \(\delta\), \(\sigma\), \(I\), \(m\), and \(d\),
  there is a computable function \(f_2\) such that the enumeration overall
  has at most \(f_2(\epsilon, \delta, \sigma, I, m, d)\) delay.
  Setting \(f(\epsilon, \delta, \sigma, I, m, d)
  \deff \max\bigl(f_1(\epsilon, \delta, \sigma, I, m, d),
  f_2(\epsilon, \delta, \sigma, I, m, d)\bigr)\) finishes the proof of the running-time bounds.

  Lastly, the set
  \(\Hypo \deff \bigcup_{t \in T} S^t_\A(A/A) \subseteq \bigcup_{t \in T'} S^t_\A(A/A)\)
  has graph dimension at most \(c\) by \cref{lem:graph-dimension-union}.
  Thus, analogously to the proof of \cref{thm:pac-learn-locally-bounded-expansion},
  the error bounds follow from \cref{lem:uniform-convergence-graph-dimension}.
\end{proof}

\subsection{Proof of Theorem~\ref{thm:learning-formulas-pac-enumeration}}
\label{subsec:appendix-learning-formulas}

For a signature \(\sigma\), numbers \(k \in \N\), \(\ell, m \in \Npos\),
and a finite set of integers \(I \subset \Z\),
let \(\Phi[\sigma, k, \ell, m, I]\) be the set of all \(\cgFOC[\sigma]\) formulas
\(\phi(x_1, \dots, x_{k'}, y_1, \dots, y_\ell)\)
with \(k' \leq k\), \(\abs{\phi} \leq m\)
and such that \(\phi\) only uses variables from
\(\set{x_1, \dots, x_k, y_1, \dots, y_\ell, z_1, \dots, z_m}\)
and integers from \(I\).
Analogously to the set \(T[\sigma, k, \ell, m, I]\),
the set \(\Phi[\sigma, k, \ell, m, I]\) is finite and computable from
\(\sigma\), \(k\), \(\ell\), \(m\), and  \(I\).

\learningFormulasPacEnumeration*

\begin{proof}
  Let \(\C\) be a graph class of effectively locally bounded expansion,
  and let \(\sigma\) be a signature.
  By \cref{cor:VCdimension},
  there is a computable function \(f_{\VC} \colon \cgFOC \to \N\) such that
  for every \(\cgFOC[\sigma]\) formula \(\phi\)
  and every \(\sigma\)-structure \(\A\) with \(G_\A \in \C\),
  the VC dimension of \(\phi\) in \(\A\) is a most \(f_{\VC}(\phi)\).
  Hence, given \(k\), \(\ell\), \(m\), and \(I\),
  we can compute \(\Phi \deff \Phi[\sigma, k, \ell, m, I]\)
  and \(c \deff \abs{\Phi} + \sum_{\phi \in \Phi} f_{\VC}(\phi)\).
  Then, for every \(\sigma\)-structure \(\A\) with \(G_\A \in \C\),
  by \cref{lem:vc-dimension-union},
  the set \(\Hypo_\A \deff \bigcup_{\phi \in \Phi} S^\phi_\A(A/A)\) has VC dimension at most \(c\).

  Let \(s_{\UC}\) be the computable function from \cref{lem:uniform-convergence-vc},
  let \(s \deff s_{\UC}(c, \epsilon, \delta)\),
  and let \(S = (\tw_1, \lambda_1), \dots, (\tw_s, \lambda_s)\)
  be a sequence drawn \iid from \(\D\).
  Then, by \cref{lem:uniform-convergence-vc},
  with probability at least \(1-\delta\),
  it holds that \(\abs{\err_\D(\phi, \tv) - \err_S(\phi, \tv)} \leq \epsilon\)
  for every formula \(\phi(\tx, \ty) \in \Phi\) and every tuple \(\tv \in A^k\),
  where
  \[\err_S(\phi, \tv) \deff
  {\textstyle \frac{1}{s}} \cdot \abs{\setc{i \in [s]}{\sem{\phi(\tv, \tw_i)}^\A \neq \lambda_i}}.\]

  Let \(\sigma' \deff \sigma \uplus \setc{R_{i,j}}{i \in [s], j \in [\ell]}\)
  for fresh unary relation symbols \(R_{i,j}\).
  We let \(\A_S\) be the \(\sigma'\)-expansion of \(\A\)
  with \(R_{i,j}(\A_S) \deff \set{(\tw_i)_j}\).
For every subsequence \(S'\) of \(S\) and every formula \(\phi(\tx, \ty) \in \Phi\), we set
  \begin{align*}
    \phi_{S'}(\tx)
    &\deff \forall y_1 \dots \forall y_\ell
    \Bigl(\mkern-5mu\Land_{(\tw_i, \lambda_i) \in S'}\mkern-5mu \psi_{i, +}(\tx, \ty)
    \land\mkern-25mu \Land_{(\tw_i, \lambda_i) \in S \setminus S'}\mkern-5mu \psi_{i, -}(\tx, \ty)\Bigr)\\
    \intertext{with}
    \psi_{i, +}(\tx, \ty)
    &\deff \bigl(\Land_{j \in [\ell]} R_{i,j}(y_j)\bigr)
    \rightarrow \phi_{\lambda_i}(\tx, \ty),\\
    \psi_{i, -}(\tx, \ty)
    &\deff \bigl(\Land_{j \in [\ell]} R_{i,j}(y_j)\bigr)
    \rightarrow \phi_{1-\lambda_i}(\tx, \ty),\\
    \phi_0(\tx, \ty)
    &\deff \neg \phi(\tx, \ty), \quad\text{and}\quad
    \phi_1(\tx, \ty)
    \deff \phi(\tx, \ty).
  \end{align*}
  Then, for every \(\tv \in A^k\), it holds that
  \(\A_S \models \phi_{S'}[\tv]\)
  if and only if \(\sem{\phi(\tv, \tw)}^\A = \lambda\) for all \((\tw, \lambda)\) in \(S'\)
  and \(\sem{\phi(\tv, \tw)}^\A \neq \lambda\) for all \((\tw, \lambda)\) in \(S \setminus S'\).
  Hence, for all \(\tv \in A^k\) with \(\A_S \models \phi_{S'}[\tv]\),
  we have \(\err_S(\phi, \tv) = 1- \abs{S'}/\abs{S}\).

  In the algorithm for \cref{thm:learning-formulas-pac-enumeration},
  the preprocessing works as follows.
  First, we compute \(\Phi\), \(c\), and \(s \deff s_{\UC}(c, \epsilon, \delta)\) as described above.
  Next, we draw a sequence \(S = (\tw_1, \lambda_1), \dots, (\tw_s, \lambda_s)\)
  of length \(s\) from \(\D\) using the oracle,
  and we compute the structure \(\A_S\).
Finally, for every subsequence \(S'\) of \(S\), we run the preprocessing
  from \cref{thm:answering-enumeration} on \(\A_S\) and \(\phi_{S'}\).
  This allows us to enumerate all tuples \(\tv\) that satisfy \(\A_S \models \phi_{S'}[\tv]\)
  with constant delay.
  Since \(s\) and \(\Phi\) only depend on \(\sigma\), \(k\), \(\ell\), \(m\), \(I\),
  \(\epsilon\), and \(\delta\),
  using the running-time bounds from \cref{thm:answering-enumeration},
  there is a function \(f_1\) such that the preprocessing in total takes time at most
  \(f_1(\alpha, \epsilon, \delta, \sigma, I, m) \cdot \abs{A}^{1+\alpha}\).

  For the enumeration, we iterate over all subsequences \(S'\) of \(S\) with decreasing length,
  starting with the longest subsequence \(S' = S\).
  For every subsequence \(S'\), we iterate over all formulas \(\phi(\tx, \ty) \in \Phi\),
  in ascending lexicographic order on \((\abs{\tx}, \abs{\phi})\).
  For every such formula, using the algorithm from \cref{thm:answering-enumeration}\ref{item:enumeration}
  based on our preprocessing above,
  we iterate over all tuples \(\tv\) that satisfy \(\A_S \models \phi_{S'}[\tv]\),
  and we output \((\phi, \tv, 1 - \abs{S'}/\abs{S})\) for every such tuple.
  Again, using the running-time bounds from \cref{thm:answering-enumeration},
  there is a function \(f_2\) such that we enumerate the tuples with delay at most
  \(f_2(\alpha, \epsilon, \delta, \sigma, I, m)\).
  The statement of \cref{thm:learning-formulas-pac-enumeration} then follows with
  \(f(\alpha, \epsilon, \delta, \sigma, I, m) \deff \max\bigl(
    f_1(\alpha, \epsilon, \delta, \sigma, I, m),
    f_2(\alpha, \epsilon, \delta, \sigma, I, m)
  \bigr)\).
\end{proof}

\end{document}